\documentclass{article}

\usepackage[utf8]{inputenc}

\usepackage{amsmath, amsthm, amssymb}
\usepackage{bm}
\usepackage{mathtools}

\usepackage{graphicx}
\usepackage{verbatim}
\usepackage{natbib}
\usepackage{caption}
\usepackage{fancyvrb}
\usepackage{enumerate}
\usepackage{relsize}
\usepackage{eufrak}
\usepackage{float}
\usepackage{booktabs}
\usepackage{footmisc}

\usepackage[margin=1.5in]{geometry}

\usepackage{hyperref}
\hypersetup{colorlinks,citecolor=blue,urlcolor=blue,linkcolor=blue}

\usepackage{enumitem}
\usepackage{algorithm2e}
\RestyleAlgo{boxruled}
\LinesNumbered
\DontPrintSemicolon
\SetAlFnt{\normalsize}
\SetKwInOut{Input}{Input}

\usepackage{xparse}
\let\oldeqref\eqref
\makeatletter
\RenewDocumentCommand\eqref{s m}{%
  \IfBooleanTF#1%
  {\oldeqref{#2}}
  {\oldeqref{#2}}
}
\makeatother

\newcommand{\citepsupplement}{\citep}
\newcommand{\citetsupplement}{\citet}

\usepackage{adjustbox}

\usepackage{tikzexternal}
\usepackage{stefan_tex}

\theoremstyle{plain}
\newtheorem{prop}{Proposition}

\newtheorem{coro}[prop]{Corollary}
\newtheorem{lemm}[prop]{Lemma}
\newtheorem{theo}[prop]{Theorem}

\theoremstyle{definition}

\theoremstyle{remark}

\title{Confidence Intervals for \\ Nonparametric Empirical Bayes Analysis}
\author{Nikolaos Ignatiadis \\ \texttt{ignat@stanford.edu} \and Stefan Wager \\ \texttt{swager@stanford.edu}}
\date{September 2021}
          
\begin{document}

\maketitle

\begin{abstract}
In an empirical Bayes analysis, we use data from repeated sampling to imitate inferences made by an oracle Bayesian with extensive knowledge of the data-generating distribution. Existing results provide a comprehensive characterization of when and why empirical Bayes point estimates accurately recover oracle Bayes behavior. In this paper, we develop flexible and practical confidence intervals that provide asymptotic frequentist coverage of empirical Bayes estimands, such as the posterior mean or the local false sign rate. The coverage statements hold even when the estimands are only partially identified or when empirical Bayes point estimates converge very slowly.
\end{abstract}

\section{Introduction}
\label{sec:intro}

Empirical Bayes methods enable frequentist estimation that emulates a Bayesian oracle. Suppose
we observe $\Zo$ generated as below, and want to estimate $\theta_G(\zo)$,
\begin{equation}
\label{eq:EB}
\mu \sim \gprior, \ \ \Zo \sim p(\cdot \cond \mu), \ \ \theta_{\gprior}(\zo) = \EE[\gprior]{h(\mu) \cond \Zo = \zo},
\end{equation}
for some known function $h(\cdot) \in \RR$.  Given knowledge of $\gprior$, $\theta_{\gprior}(\zo)$ can be directly evaluated via
Bayes' rule. An empirical Bayesian does not know $\gprior$,
but seeks an approximately optimal estimator $\smash{\htheta(\zo) \approx \theta_{\gprior}(\zo)}$ using independent
draws $\Zo_1, \, \Zo_2, \, ..., \, \Zo_n$ from the distribution \eqref{eq:EB}.

The empirical Bayes approach was first introduced by \citet{robbins1956empirical} and has proven to be successful
in a wide variety of settings with repeated observations of similar phenomena, such as genomics~\citep{efron2001empirical, love2014moderated}, education~\citep{lord1969estimating, gilraine2020new} and actuarial science~\citep{buhlmann2006course}. Table~\ref{tab:eb_examples} provides concrete applications of model~\eqref{eq:EB} for these subject areas. In all examples, the posterior mean \smash{$\theta_{\gprior}(\zo) = \EE[\gprior]{\mu \cond \Zo=z}$} is a statistic of interest, as it describes the (mean squared error) optimal shrinkage rule for estimating $\mu$. In the genomics application, it is also of interest to determine the local false-sign rate
\smash{$\theta_{\gprior}(\zo) = \PP{\mu \zo \leq 0 \cond \Zo = \zo}$}, i.e., the posterior probability that $\mu$ has a different sign than $\Zo$. 

\renewcommand{\arraystretch}{1.2} 
\begin{table}
\centering
\begin{tabular}{@{}lllll@{}} \toprule
Subject & $i$ & $Z_i$ & $\mu_i$ & $Z_i \cond \mu_i \; \approxdot$  \\ \midrule 
Actuarial science & Contract & Number of insurance claims & Risk profile & $\Poisson{\mu_i}$  \\ 
Education & Student & Score in test with $N$ & Latent ability  & $\text{Binom}(N, \mu_i)$  \\
& & multiple choice questions &  & \\
Genomics & Gene & t-statistic comparing & Standardized  & $\nn(\mu_i,1)$  \\
& &  expression between conditions & effect size & 
\\ \bottomrule
\end{tabular}
\caption{Example applications for empirical Bayes inference in model~\eqref{eq:EB}.}
\label{tab:eb_examples}
\end{table}

As elaborated later, there is by now a large literature proposing a suite of estimators \smash{$\htheta(\zo)$}
for $\theta_{\gprior}(\zo)$. Many of these estimators have theoretical guarantees under nonparametric specification of $\gprior$, say $\gprior \in \gcal$, where $\gcal$ is a convex class of distributions. The goal of this paper is to move past point estimation, and develop nonparametric confidence intervals for $\theta_G(\zo)$, i.e., intervals with the following property:
\begin{equation}
\label{eq:CI}
\ii_\alpha(\zo) = \sqb{\htheta^-_\alpha(\zo), \, \htheta^+_\alpha(\zo)}, \ \ 
\liminf_{n \rightarrow \infty} \PP[\gprior]{\theta_{\gprior}(\zo) \in \ii_\alpha(\zo)} \geq 1 - \alpha \text{ for all } \gprior \in \gcal. 
\end{equation}
Despite widespread use of empirical Bayes methods, the problem has received surprisingly little attention. In fact, we are not aware of confidence intervals with property~\eqref{eq:CI} beyond two special cases: one proposal by~\citet{lord1975empirical} for inference about the posterior mean in the binomial model and another by~\citet{robbins1980empirical} for the same task in the Poisson model.

\subsection{Motivating application: Predicting automobile insurance claims}
\label{subsec:bichsel}
To motivate our interest in confidence intervals of the form~\eqref{eq:CI}, we revisit the historical work of~\citet{bichsel1964erfahrungs}. Bichsel developed a theoretical framework for assigning automobile insurance premium rates, in a way that accounts for the claims experience of each individual. He analyzed a dataset (Table~\ref{tab:bichsel_intro}) of claims made in the year 1961 by holders of a Swiss automobile insurance policy. Bichsel posited that $\Zo_i(t)$, the number of claims made in year $t$ by the $i$-th insurance holder, is distributed as $\Poisson{\mu_i}$, where $\mu_i$ is $i$'s latent risk. $\mu_i$ is a random draw from a distribution $\gprior$ that captures the heterogeneity of the insurance portfolio. Bichsel further assumed that the number of claims $Z_i(t), Z_i(t')$ in different years $t \neq t'$ are i.i.d. conditionally on $\mu_i$. Given these assumptions, Bichsel sought to estimate the expected number of claims in the next year, among all insurance holders that made $\Zo_i(1961) = \zo$ claims in 1961,
\begin{equation}
\label{eq:posterior_mean_poisson}
\theta_{\gprior}(\zo) = \EE[\gprior]{ \Zo(1962) \cond \Zo(1961)=\zo} = \EE[\gprior]{\mu \cond \Zo(1961) =\zo}.
\end{equation}
Bichsel reasoned, that if $\theta_{\gprior}(\zo)$ were known, it could be used by the insurance company for policy decisions, such as increasing or decreasing the premium of a policy holder with $\zo$ claims in 1961. Since $\gprior$, and consequently $\theta_{\gprior}(\zo)$, were not known to Bichsel, he considered an empirical Bayes approach.\footnote{\label{footnote:eb_intervals}$\theta_{\gprior}(\zo)$ is a property of $\gprior$, i.e., of the portfolio heterogeneity. The goal is to best assess how many claims will be made across all individuals in the portfolio that made $\zo$ claims in 1961, and not to reason about the risk $\mu_i$ of any individual policy holder with $\Zo_i(1961)=\zo$. The problem of forming intervals containing the true $\mu_i$ (for individuals) is of scientific importance, see e.g.,~\citet*{morris1983parametric, laird1987empirical, armstrong2020robust, koenker2020} for some proposals; however it is not the problem we consider in this work.}  

The problem of point estimation for $\theta_\gprior(\zo)$ is well understood. One popular nonparametric solution\footnote{In his work, Bichsel modeled $\gprior$ parametrically as a Gamma distribution with unknown parameters.} is to first estimate $\gprior$ through the nonparametric maximum likelihood estimator (NPMLE) of~\citet{kiefer1956consistency} and~\citet{simar1976maximum}: one estimates \smash{$\hG$} as the maximizer of the marginal log-likelihood \smash{$\sum_{i} \log(f_{\gprior}(\Zo_i))$}, \smash{$f_{\gprior}(\zo)= \int \exp(-\mu)\mu^{\zo}/\zo!\,d\gprior(\mu)$}, among all possible prior distributions $\gprior$. Then, with \smash{$\hG$} in hand, one estimates \smash{$\theta_{\gprior}(\zo)$} through the plug-in principle, i.e., \smash{$\htheta_{\gprior}(\zo) = \theta_{\hG}(\zo)$} (shown in the third column of Table~\ref{tab:bichsel_intro}). 

However, in so far as $\theta_{\gprior}(\zo)$ may be used for policy decisions of the insurance company, it is also important to assess the uncertainty in estimating it. In this paper, we develop two complementary approaches that address the problem of inference for empirical Bayes estimands, and enable the construction of intervals with the property~\eqref{eq:CI} under the general model~\eqref{eq:EB}. The last two columns of Table~\ref{tab:bichsel_intro} show the two confidence intervals that we propose for Bichsel's data. The assumption we make in forming these intervals is that $\gprior$ is supported on $[0,5]$. The `F-localization' intervals (third column of Table~\ref{tab:bichsel_intro}) have simultaneous coverage for all $\zo$, while the `\Amari' intervals (fourth column) have pointwise coverage. We next provide a high-level overview of our two constructions.

\renewcommand{\arraystretch}{1.2} 
\begin{table}
\centering
\begin{tabular}{@{}rllll@{}} \toprule
z & $\#\cb{\Zo_i =\zo}$  & $\htheta_{\text{NPMLE}}(\zo)$ & F-localization $\ii_\alpha(\zo)$ & \Amari~$\ii_\alpha(\zo)$ \\ \midrule
0 & 103704  & 0.14 & 0.13 -- 0.14 & 0.13 -- 0.14 \\
1 & 14075 & 0.25 & 0.23 -- 0.27 & 0.24 -- 0.26 \\
2 & 1766  & 0.44 & 0.36 -- 0.53 & 0.38 -- 0.49 \\
3 & 255  & 0.69 & 0.48 -- 0.94 & 0.53 -- 0.91 \\
4 & 45  & 0.82 & 0.52 -- 1.64 & 0.58 -- 1.39 \\
$\geq$ 5 & 8 &  &  &  \\ \bottomrule
\end{tabular}
\caption{\textbf{Empirical Bayes confidence intervals in an actuarial application.} The first two columns report the number of claims $z$ made in 1961 by 119853 holders of a Swiss automobile insurance policy~\citep{bichsel1964erfahrungs}. The next column reports the nonparametric maximum likelihood (NPMLE) point estimates of the posterior mean \smash{$\theta_{\gprior}(\zo) = \EE[\gprior]{\mu \cond \Zo =\zo}$} under model~\eqref{eq:EB} with \smash{$\Zo \cond \mu \; \sim \Poisson{\mu}$}. The last two columns report the two types of confidence intervals developed in this work; the $F$-localization intervals (Section~\ref{sec:Flocalization}) and \Amari~intervals (Section~\ref{sec:amari}).}
\label{tab:bichsel_intro}
\end{table}

\subsection{Empirical Bayes confidence intervals}
In our approach the data analyst first specifies~\eqref{eq:EB}, i.e., \smash{$\theta_{\gprior}(\zo)$}, the empirical Bayes estimand of interest (e.g., the posterior mean~\eqref{eq:posterior_mean_poisson}) and the conditional distribution of $\Zo$ given $\mu$ (e.g., $\Poisson{\mu}$), which we represent by its conditional density $p(\cdot \cond \mu)$ with respect to a $\sigma$-finite measure $\lambda$ on a subset of $\RR$ (e.g., the counting measure on $\NN_{\geq 0}$). We also require the data analyst to specify a convex class of priors $\gcal$ such that $\gprior \in \gcal$. For example, for our analysis of Bichsel's data in Table~\ref{tab:bichsel_intro}, we assumed that $\gprior \in \gcal = \pp([0, 5])$, where:\footnote{We provide more guidance for choosing $\gcal$ in Section~\ref{sec:gcal}.}
\begin{equation}
\label{eq:all_dbns}
\pp(\Ksupport)  := \cb{G \text{ distribution}:\; \support(G) \subset \Ksupport} \text{ for } \Ksupport \subset \RR.
\end{equation}

\subsubsection{$F$-localization}
Our first confidence interval construction is based on the notion of \emph{$F$-localization}. The key idea is to construct a confidence set for the marginal distribution of $\Zo$ and then determine all $\gprior \in \gcal$ consistent with this confidence set. Let us denote the marginal distribution of $\Zo$ by  $F_{\gprior}$ and its $d\lambda$-density by $f_{\gprior}$, i.e.,
\begin{equation}
\label{eq:marginal_density}
f_{\gprior}(\zo) = \int p(\zo \cond \mu)d\gprior(\mu)\;,\;\;\;  F_{\gprior}(t) = \PP[\gprior]{\Zo \leq t} = \int \ind\p{\zo \leq t} f_{\gprior}(\zo)d\lambda(\zo).
\end{equation}
We then define an $F$-localization as an (asymptotic) $1-\alpha$ confidence set $\ff_n(\alpha)$ of distributions, i.e., a set such that
\begin{equation}
\label{eq:dbn_nbhood}
\liminf_{n \to \infty}\cb{ \PP[G]{ F_G \in \ff_n(\alpha)} \, - \, (1-\alpha)}\, \geq \, 0.
\end{equation}
With an $F$-localization $\ff_n(\alpha)$ in hand, and deferring the construction of such to Section~\ref{sec:Flocalization}, we can form confidence intervals \smash{$\ii_{\alpha}(\zo) = [\htheta^-_\alpha(\zo),\htheta^+_\alpha(\zo)]$} for $\theta_{\gprior}(\zo)$ by letting,
\begin{align}
&\htheta^-_\alpha(\zo) = \inf \cb{ \theta_{\gprior}(\zo) \mid \gprior \in \gcal\p{\ff_n(\alpha)}},\, \htheta^+_\alpha(\zo) = \sup \cb{ \theta_{\gprior}(\zo) \mid \gprior \in \gcal\p{\ff_n(\alpha)}}, \label{eq:nbhood_worst_case}\\
&\text{where } \gcal(\ff) = \cb{ \gprior \in \gcal \cond F_{\gprior} \in \ff}. \label{eq:localized_gcal}
\end{align}
The intervals~\eqref{eq:nbhood_worst_case} satisfy~\eqref{eq:CI}, since \smash{$\mathbb P_{\gprior}[  \theta_{\gprior}(\zo) \in [\htheta^-_\alpha(\zo),\htheta^+_\alpha(\zo)]] \geq \PP[\gprior]{F_{\gprior} \in \ff_n(\alpha)},$}
and the same argument also demonstrates that coverage holds simultaneously over all possible empirical Bayes estimands \smash{$\theta_{\gprior}(\zo) = \EE[\gprior]{h(\mu) \cond \Zo=\zo}$}, where both $\zo$ and $h$ can vary. In Section~\ref{sec:Flocalization} we explain how for common choices of $\gcal$ and $\ff_n(\alpha)$,~\eqref{eq:nbhood_worst_case} can be computed by solving two linear programs.

It is interesting to consider the $F$-localization approach in the context of the dichotomy of~\citet{efron2014two, efron2019bayes} on  $F$- versus $G$-modeling. Under model~\eqref{eq:EB}, it is typically  straightforward to estimate $F_{\gprior}$, because the observed $\Zo_i$ are direct measurements from $F_{\gprior}$. In contrast, estimation of $\gprior$ is a difficult inverse problem. In some cases, the empirical Bayes estimand of interest may be expressed directly in terms of $F_{\gprior}$: for example, \citet{robbins1956empirical} proves that the posterior mean in the Poisson model (\smash{$p(\cdot \cond \mu) = \Poisson{\mu}$} in~\eqref{eq:EB}) is equal to,
\begin{equation}
\label{eq:robbins_poisson}
\theta_{\gprior}(\zo) = \EE[G]{\mu \cond \Zo = \zo} = (\zo+1)\frac{f_{\gprior}(\zo+1)}{f_{\gprior}(\zo)},\;f_{\gprior}(\zo) = \PP[G]{\Zo = \zo},\;  \zo \in \mathbb N_{\geq 0}.
\end{equation}
When a formula as~\eqref{eq:robbins_poisson} is available, it is convenient to proceed by $F$-modeling, i.e., to estimate \smash{$\hat{F}_G$} and evaluate the corresponding $F$-formula by the plugin principle. In the Poisson model, letting \smash{$\hat{F}_G$} be the empirical distribution of the $Z_i$, the plugin principle leads to the estimate $\htheta_{\text{Robbins}}(\zo) = (\zo+1)\#\cb{\Zo_i =\zo+1}/\#\cb{\Zo_i =\zo}$ for \eqref{eq:robbins_poisson}. A caveat of $F$-modeling, however, is that, natural constraints on the empirical Bayes estimand are not enforced. For instance, the posterior mean \eqref{eq:robbins_poisson} in the Poisson model is non-decreasing in $\zo$, but $\htheta_{\text{Robbins}}(\cdot)$ does not enforce such monotonicity. In contrast,  natural constraints such as monotonicity are automatically enforced under $\gprior$-modeling, that is, if one first estimates \smash{$\hG$} and then lets \smash{$\htheta_{\gprior}(\zo) = \theta_{\hG}(\zo)$}.

For inference using $F$-localization, the two perspectives are complementary. The data analyst constructs~\eqref{eq:dbn_nbhood}, an $F$-modeling task, and the Bayes structure of the problem is enforced through~\eqref{eq:nbhood_worst_case}. For example, in the Poisson posterior mean problem, the lower bounds of the confidence intervals, \smash{$\htheta^-_\alpha(\zo)$}, are monotonic in $\zo$, and similarly for the upper bounds \smash{$\htheta^+_\alpha(\zo)$}.

\subsubsection{\Amari~(Affine Minimax Anderson-Rubin Intervals)}
\label{subsubsec:amari}
The $F$-localization approach is generic, streamlined to implement and enables simultaneous inference for all empirical Bayes estimands of interest. The $F$-localization intervals for a specific estimand $\theta_{\gprior}(\zo)$, however, can be overly wide. Our second construction, \Amari, seeks to do better than $F$-localization, i.e., to provide shorter confidence intervals, by focusing on a specific estimand (compare e.g., columns 3 and 4 of Table~\ref{tab:bichsel_intro}). The starting point for \Amari~is the observation that we can write the empirical Bayes estimand $\theta_{\gprior}(\zo)$, as a ratio of two linear functionals of $\gprior$, i.e.,
\begin{equation}
\label{eq:ratio}
\theta_{\gprior}(\zo) = \frac{\int h(\mu) p(\zo \cond \mu) \ d\gprior(\mu)}{\int p(\zo \cond \mu) \ d\gprior(\mu)} = \frac{a_{\gprior}(\zo)}{f_{\gprior}(\zo)},
\end{equation}
where $f_{\gprior}(\cdot)$ is the marginal density of $\Zo$ and $a_{\gprior}(\zo)$ is used to denote the numerator. Hence, by a construction that goes back to at least~\citet{fieller1940biological}, the following two hypothesis tests are equivalent for $c \in \RR$,
\begin{equation}
\label{eq:hypothesis_test_fractional}
H_0:\; \theta_{\gprior}(\zo) = c        \; \Longleftrightarrow \; H_0: \theta_{\gprior}^{\text{lin}}(\zo; c) = 0, \;\text{where } \theta_{\gprior}^{\text{lin}}(z; c) = a_{\gprior}(\zo) - c f_{\gprior}(\zo).
\end{equation}
By inverting the test for \smash{$H_0:\; \theta_G(\zo) = c$} we can form confidence intervals for $\theta_{\gprior}(\zo)$. The upshot of ~\eqref{eq:hypothesis_test_fractional} then is that it suffices to construct confidence intervals for linear functionals $L(\gprior)$ of $\gprior$, say $L(\gprior) = \theta_{\gprior}^{\text{lin}}(\zo; c)$. We provide the details of this reduction in Section~\ref{sec:amari},
and proceed to explain our approach to inference for linear functionals of $\gprior$. Our core proposal is to estimate $L(\gprior)$  as an affine estimator, i.e., one of the form
\begin{equation}
\label{eq:Q}
\begin{split}
\hL = \hL(\gprior) = \frac{1}{n} \sum_{i = 1}^n Q(\Zo_i),
\end{split}
\end{equation}
where \smash{$Q(\cdot)$} is chosen to optimize a worst-case bias-variance tradeoff depending on the prior class $\gcal$.  To form confidence intervals, we first estimate the variance and worst-case bias of~\eqref{eq:Q} as
\begin{align}
&\hV = \frac{1}{n(n - 1)}\sqb{\sum_{i = 1}^n Q^2(\Zo_i) - \p{\sum_{i=1}^n Q(\Zo_i)}^2\bigg/n}, \label{eq:sample_var_est}\\
&\hB^2 = \sup_{\gprior \in \gcal(\ff_n)}\cb{\Bias[\gprior]{Q,L}^2]}, \;\; \Bias[\gprior]{Q,L} = \int Q(\zo)f_{\gprior}(\zo)d\lambda(\zo) - L(\gprior). \label{eq:bias_est}
\end{align}
Here, the worst case bias is computed with respect to \smash{$\gcal(\ff_n)$}~\eqref{eq:localized_gcal}, where \smash{$\ff_n = \ff_n(\alpha_n)$} is an $F$-localization at level $\alpha_n \to 0$ as $n \to \infty$. With \smash{$\hV,\hB$} in hand, we build bias-aware confidence intervals \smash{$\ii_\alpha$} for \smash{$L(\gprior)$}
\citep[e.g.,][]{armstrong2018optimal,imbens2004confidence, imbens2017optimized}
\begin{equation}
\label{eq:im_iw_ci}
\ii_\alpha = \hL \pm t_\alpha(\hB, \hV), \ \
  t_\alpha(B, V)= \inf\cb{t : \PP{\abs{b + V^{1/2} W} > t} < \alpha \text{ for all } \abs{b} \leq B},
\end{equation}
where $W \sim \nn\p{0, \, 1}$ is a standard Gaussian random variable. Sections~\ref{sec:linear_functionals} and \ref{sec:amari} have formal results establishing asymptotic coverage
properties for these intervals.

Conceptually, our AMARI intervals build on recent work by~\citet{noack2019bias}, who consider inference of average treatment effects in the fuzzy regression discontinuity design. There, the estimand also takes the form of a ratio of two linear functionals as in~\eqref{eq:ratio}.~\citet{noack2019bias} name their approach after~\citet{anderson1949estimation}, who develop confidence intervals in the linear instrumental variable model, and so similarly we acknowledge~\citet{anderson1949estimation} as part of the acronym \Amari.

\subsection{Related Work}
\label{subsec:related_work}

As discussed briefly above, the empirical Bayes principle has spurred considerable interest over
several decades. One of the most successful applications of this idea involves compound estimation of a
high-dimensional Gaussian mean: we observe \smash{$\mathbf{Z} \sim \nn\p{\bm{\mu}, \, I}$}, and want to recover $\bm{\mu}$ under squared
error loss. If we assume that the individual $\mu_i$ are drawn from a prior $\gprior$, then empirical Bayes estimation
provides a principled shrinkage rule \citep{efron1973stein,efron2011tweedie}, whose theoretical properties are well-understood~\citep{brown2009nonparametric,jiang2009general}. The compound estimation problem when the individual $Z_i$ are Poisson, is  also reasonably well understood~\citep{brown2013poisson}.

The more general empirical Bayes problem~\eqref{eq:EB} has raised interest in applications
\citep{efron2012large,efron2016empirical,stephens2016false, koenker2017rebayes}; however, the accompanying formal results are less
comprehensive. \citet{muralidharan2012high} considered compound estimation in~\eqref{eq:EB}, when $p(\cdot \cond \mu)$ is a one-dimensional exponential family. For $p(\cdot \cond \mu) = p(\cdot - \mu)$, a smooth location family, some authors, including~\citet{butucea2009adaptive,pensky2017minimax}, have considered rate-optimal estimation of linear functionals of $\gprior$; and their setup covers, for example, the numerator $a_{\gprior}(\zo)$ in \eqref{eq:ratio}. The main message of these papers, however, is rather pessimistic: for example, \citet{pensky2017minimax} shows that for many linear functionals, the minimax rate for estimation in mean squared error over certain Sobolev classes $\gcal$ is logarithmic (to some negative power) in the sample size.

In this paper, we study a closely related problem but take a different point of view. Even if minimax rates of
optimal point estimates $\htheta(\zo)$ may be extremely slow (or even if estimands are only partially identified), we seek confidence intervals for $\theta_{\gprior}(\zo)$ that still achieve accurate coverage in reasonable sample sizes and explicitly account for bias. The results of
\citet{butucea2009adaptive} and \citet{pensky2017minimax} imply that the length of our confidence intervals
must go to zero very slowly in general; but this does not mean that our intervals cannot be useful in finite samples
(and, in fact, our real data applications in Section~\ref{sec:applications} and numerical experiments in Section~\ref{sec:simulations} suggest that they can be).

To the best of our knowledge, with the exception of a handful of special cases, the problem of nonparametric inference in empirical Bayes problems has been left unexplored. Furthermore, practitioners using empirical Bayes ideas typically do not conduct inference and instead only consider point estimates of functionals of the unknown prior $\gprior$. Among recent empirical Bayes works, \citet{efron2014two, efron2016empirical,efron2019bayes} has advocated estimating (and reporting) the variance of empirical Bayes estimates \smash{$\htheta(\zo)$}, and then using these variance estimates for uncertainty quantification. Such intervals, however, do not account for bias and so could only achieve valid coverage via undersmoothing; and it is unclear how to achieve valid undersmoothing in practice, noting the very slow rates of convergence in empirical Bayes problems.
\citet{efron2014two, efron2016empirical,efron2019bayes} himself does not suggest his intervals be combined with undersmoothing, and rather uses them as pure uncertainty quantification tools.

Two notable existing results for inference as in~\eqref{eq:CI} concern the posterior mean in the Binomial~\citep{lord1975empirical, lord1976interval} and Poisson problems~\citep{robbins1980empirical, karlis2018confidence}. The $F$-localization approach we propose, generalizes the approach of~\citet{lord1975empirical} and~\citet{lord1976interval} for inference of the posterior mean in the Binomial problem to the general empirical Bayes problem~\eqref{eq:EB}. We provide more details regarding this connection at the end of Section~\ref{subsec:floc_emb}, and in Section~\ref{subsec:lord_dataset} we revisit the data application of~\citet{lord1975empirical}. The Poisson posterior mean problem is special, and particularly amenable to the task of forming confidence intervals, because of the existence of Robbins' formula~\eqref{eq:robbins_poisson}, as we elaborate in Section~\ref{subsec:power_poisson}.

From a methodological perspective, our work relies upon advances in convex programming and is inspired by~\citet{koenker2014convex}, who demonstrated that it is fruitful to revisit traditional ideas in empirical Bayes estimation through the lens of modern convex optimization. The $F$-localization approach requires solving two linear programs (or more generally, quasi-convex programs, cf. Section~\ref{sec:Flocalization}). \Amari, our second approach, builds heavily on the literature on affine minimax estimation of linear functionals in Gaussian problems. \citet{donoho1994statistical} and related papers \citep{armstrong2018optimal,cai2003note,donoho1991geometrizing,low1995bias, johnstone2011gaussian} show that there exist affine estimators that achieve quasi-minimax performance and can be efficiently derived via convex programming. In turn, such affine estimators have recently proven useful for statistical inference in a number
of settings, such as semiparametrics \citep{hirshberg2017balancing,kallus2016generalized} and regression discontinuity designs \citep{armstrong2018optimal,imbens2017optimized, eckles2020noise}.

\section{Simultaneous confidence intervals through $F$-localization}
\label{sec:Flocalization}
In this section we discuss our first approach, namely $F$-localization confidence intervals. The key idea is to `localize' the marginal distribution $F_{\gprior}$ with high probability, i.e., to construct a set $\ff_n(\alpha)$~\eqref{eq:dbn_nbhood} such that $F_{\gprior} \in \ff_n(\alpha)$ with (asymptotic) probability at least $1-\alpha$. $\ff_n(\alpha)$ then implies a confidence set $\cb{\gprior \in \gcal: F_{\gprior} \in  \ff_n(\alpha)}$ for $\gprior$, which we project to form confidence intervals for $\theta_{\gprior}(\zo)$ as in~\eqref{eq:nbhood_worst_case}.

A convenient and universal $F$-localization proceeds by restricting $F$ to be in a Kolmogorov-Smirnov ball around the empirical distribution function $\widehat{F}_n(t) = \frac{1}{n}\sum_{i=1}^n \ind\p{Z_i \leq t}$,
\begin{equation}
\label{eq:DKW}
\ff^{\text{DKW}}_n(\alpha) = \cb{F \text{ distribution }:  \sup_{t \in \mathbb R}\abs{F(t) - \widehat{F}_n(t)}  \leq  \sqrt{\log\p{2/\alpha}\big/(2n)}}.
\end{equation}
By Massart's tight constant for the Dvoretzky–Kiefer–Wolfowitz (DKW) inequality~\citep{massart1990tight}, the above is a finite-sample $F$-localization, i.e., \smash{$\PP[\gprior]{F_{\gprior} \in \ff^{\text{DKW}}_n(\alpha)} \geq 1-\alpha$} for all $n$ and for any choice of $p(\cdot \cond \mu)$ in~\eqref{eq:EB}.

This construction is rooted in empirical Bayes tradition. In his discussion, \citet{robbins1956empirical}, suggests that one could achieve asymptotically optimal empirical Bayes regret\footnote{That is, to learn a denoiser \smash{$\htheta(\zo)$} such that $\mathbb E[(\mu - \htheta(\Zo))^2]$ converges to the mean squared error Bayes risk for estimating $\mu$ in model~\eqref{eq:EB}.} by choosing \smash{$\hG$} such that \smash{$\sup_{t}|F_{\hG}(t) -  \widehat{F}_n(t)| \leq c_n$}, with $c_n \to 0$ as $n\to \infty$ and then using a plug-in estimate of the posterior mean \smash{$\htheta(\zo) = \theta_{\hG}(\zo)$}; see also~\citet{donoho2013achieving} for a modern refinement and implementation to achieve optimal empirical Bayes regret in the Gaussian problem.\footnote{\citet{anderson1969confidence} suggested to use the DKW band to form confidence intervals for the mean of a $[0,1]$-valued random variable, as follows: one takes the minimum, resp. maximum of \smash{$\int \zo dF(\zo)$} subject to \smash{$F \in \ff^{\text{DKW}}_n(\alpha)$} and $F$ supported on $[0,1]$; cf.~\citet{romano2000finite}.}

The optimization problem~\eqref{eq:nbhood_worst_case} defining \smash{$\htheta^+_\alpha(\zo)$} (and similarly for \smash{$\htheta^-_\alpha(\zo)$}) can be readily solved using modern convex optimization solvers, as long as $\gcal$ can be efficiently discretized (cf. Supplement~\ref{sec:gcal_discretization}). The simplest case occurs when $\gcal$ and $\ff_n(\alpha)$ may be represented by linear constraints. In that case, we may use the \citet{charnes1962programming} transformation for linear-fractional programming and compute \smash{$\htheta^+_\alpha(\zo)$} by solving a linear program. As a concrete example (see~\eqref{eq:charnescooper} below for the general case), consider the prior class $\gcal = \pp(\Ksupport)$ from \eqref{eq:all_dbns} with \smash{$\Ksupport = \cb{\mu_1,\dotsc,\mu_{\Bdisc}}$} a finite set. Then, we can compute \smash{$\htheta^+_\alpha(\zo)$} by solving the linear program:
\begin{equation}
\label{eq:dkw_loc_opt}
\begin{aligned}
&\underset{\zeta,\; (g_j)_{j=1}^{\Bdisc}}{\text{maximize}} 
& & \sum_{j=1}^{\Bdisc} h(\mu_j)p(\zo \cond \mu_j)g_j\\
& \text{subject to}
& & \sum_{j=1}^{\Bdisc}g_j =\zeta,\;\,\sum_{j=1}^{\Bdisc} p(\zo \cond \mu_j)g_j = 1,\,\;g_j \geq 0,\, j=1,\dotsc,\Bdisc,\;\, \zeta \geq 0,\;\;\\
&&& \sup_{t \in \mathbb R}\abs{ \sum_{j=1}^{\Bdisc} g_j \int_{(-\infty,t]} p(\tilde{\zo} \cond \mu_j)d\lambda(\tilde{\zo}) -\zeta \widehat{F}_n(t)}  \leq  \zeta \cdot \sqrt{\log\p{2/\alpha}\big/(2n)}.
\end{aligned} 
\end{equation}
The optimization variables $\zeta,\, g_j$ in~\eqref{eq:dkw_loc_opt} have the interpretation \smash{$\zeta=1/f_{\gprior}(\zo)$}, \smash{$g_j = \zeta \cdot \PP[\gprior]{\cb{\mu_j}}$} for $j=1,\dotsc,\Bdisc$, where $\gprior \in \gcal$.

\subsection{Refined $F$-localization}
\label{subsec:floc_emb}
The $F$-localization~\eqref{eq:DKW} is universal, and so works for any choice of likelihood $p(\cdot \mid \mu)$ in~\eqref{eq:EB}. 
In some settings, however, it is also possible to construct $F$-localizations that are tailored towards properties of a specific choice of likelihood $p(\cdot \mid \mu)$. We provide two such constructions in this section, the Gauss and $\chi^2-F$-localizations. Empirically we observe that our tailored $F$-localizations outperform~\eqref{eq:DKW} in terms of the length of~\eqref{eq:nbhood_worst_case} for the empirical Bayes estimands we consider in our data examples (Section~\ref{sec:applications}) and simulations (Section~\ref{sec:simulations}). It is an interesting future theoretical question to determine how one should choose the $F$-localization to direct power towards specific empirical Bayes estimands and likelihoods. However such considerations are outside the scope of this work.

\paragraph{Gauss-$F$-localization:} Our first tailored $F$-localization is applicable in the Gaussian empirical Bayes problem, i.e.,~\eqref{eq:EB} with \smash{$\Zo \cond \mu \; \sim \nn(\mu, \sigma^2)$} with known noise variance $\sigma^2>0$. In this case, the marginal density $f_{\gprior}$ is the convolution of $\gprior$ with the Gaussian density, and so is extremely smooth and can be estimated at quasi-parametric rates~\citep{kim2014minimax}. Here we build on this observation and seek to construct confidence intervals in terms of \smash{$\Norm{f}_{\infty,M} := \sup_{\zo \in [-M,M]} \abs{f(\zo)}$}, the supremum norm for the Lebesgue density of $F$ on the compact set $[-M,M]$, $M>0$. To this end, we first form point estimates of $f_{\gprior}(\cdot)$ using the kernel density estimator (KDE)
\begin{equation}
\label{eq:kde}
\hat{f}^{\text{K}}_n(\zo) = \frac{1}{n h_n} \sum_{i=1}^n K\p{\frac{\Zo_i-\zo}{h_n}},\;\;K(\zo) = \frac{\sin^2(1.1\zo/2)-\sin^2(\zo/2)}{\pi \zo^2/ 20},\;\;h_n = \frac{\sigma}{\sqrt{\log(n)}}.
\end{equation}
$K(\cdot)$ is a smoothing kernel of infinite order\footnote{$K(\cdot)$ is a superkernel~\citep{devroye1992note}, i.e., it is absolutely integrable, integrates to $1$ and its characteristic function is equal to $1$ on $[-1,1]$.} that was studied by~\citet{politis1993family}. We form confidence bands using Efron's multinomial bootstrap~\citep{efron1979bootstrap}. In the $b$-th bootstrap resample, we draw $(W_1^b, \dotsc, W_n^b) \sim \text{Multinomial}\p{n, (1/n,\dotsc,1/n)}$ and compute
\begin{equation}
\label{eq:kde_bootstrap} 
\hat{c}_n^b = \Norm{\hat{f}^{\text{K}}_n - \hat{f}^{\text{K},b}_n}_{\infty,M},  \;\; \hat{f}^{\text{K},b}_n(\zo) = \frac{1}{n h_n}\sum_{i=1}^n W_i^b K\p{\frac{\Zo_i-\zo}{h_n}}.
\end{equation}
The proposition below defines the Gauss-$F$-localization and proves its asymptotic validity.
\begin{prop}[Coverage of $\Norm{\cdot}_{\infty,M}$ localization in the Gaussian empirical Bayes problem]
\label{prop:floc_kde}
Assume~\eqref{eq:EB} holds with $p(\cdot \mid \mu) = \mathcal{N}(\mu, \sigma^2)$. Let $\hat{c}_n(\alpha)$ be the $(1-\alpha)$-quantile of $\hat{c}_n^b$ with respect to the Bootstrap distribution in~\eqref{eq:kde_bootstrap}. Then,
\begin{equation}
\label{eq:floc_kde}
\ff_n^{\text{Gauss}}(\alpha) = \cb{F \text{ distribution with Lebesgue density }f: \; \Norm{f - \hat{f}^{\text{K}}_n}_{\infty,M} \leq  \hat{c}_n(\alpha)}
\end{equation}
asymptotically covers the true distribution at level $1-\alpha$, in the sense of~\eqref{eq:dbn_nbhood}.
\end{prop}
Similarly to~\eqref{eq:DKW},~\eqref{eq:floc_kde} also enforces linear constraints on \smash{$F_{\gprior},\, \gprior \in \gcal$}. Hence, if \smash{$\gcal$} may be represented using linear constraints, then an analogous linear program to~\eqref{eq:dkw_loc_opt} can be used to compute $\htheta^+_\alpha(\zo)$.  More generally, whenever \smash{$\gcal$} and \smash{$\ff_n(\alpha)$} may be represented by linear constraints, then, 
by the \citet{charnes1962programming} transformation, \smash{$\htheta^+_\alpha(\zo)$}~\eqref{eq:nbhood_worst_case} may be computed by the following linear program,
\begin{align}
\label{eq:charnescooper}
&\htheta^+_\alpha(\zo) = \sup\cb{a_{\widetilde{\gprior}}(\zo) \mid f_{\widetilde{\gprior}}(\zo)= 1,\widetilde{G} \in \zeta \gcal(\ff_n(\alpha)),\; \zeta \geq 0},\; \text{where } \\
&\zeta \gcal(\ff_n(\alpha)) = \cb{\zeta\cdot \gprior \cond   \gprior \in \gcal(\ff_n(\alpha))},\;  f_{\widetilde{\gprior}}(\zo) = \int p(\zo \cond \mu)d\widetilde{\gprior}(\mu),\;a_{\widetilde{\gprior}}(\zo) = \int h(\mu)p(\zo \cond \mu)d\widetilde{\gprior}(\mu).    \notag
\end{align}

\paragraph{$\chi^2$-$F$-localization:} Our second construction pertains to categorical likelihoods, i.e., when $Z_i \in \zz$ and $\zz$ is a finite set with $\#\zz = N+1$, $N \in \NN$. It is based on Pearson's $\chi^2$ distance,
\begin{equation}
\label{eq:floc_chisq}
\ff_n^{\chi^2}(\alpha) = \cb{F \in \pp(\zz) \text{ with pmf }f:\; \sum_{\zo=0}^N \frac{(n \hat{f}_n(\zo) - n f(\zo))^2}{n f(\zo)} \leq \chi^2_{N,1-\alpha}},
\end{equation}
where \smash{$\hat{f}_n(\zo) = \#\cb{\Zo_i = \zo}/n$} is the empirical probability of $\zo$ and  \smash{$\chi^2_{N,1-\alpha}$} is the $1-\alpha$ quantile of the \smash{$\chi^2$} distribution with $N$ degrees of freedom. The validity of~\eqref{eq:floc_chisq} in the sense of~\eqref{eq:dbn_nbhood} follows from standard asymptotics in categorical data analysis~\citep{agresti2013categorical} and coverage will (approximately) hold in finite samples as long as $n \cdot f_{\gprior}(\zo)$ is sufficiently large for all $\zo$. 

The $\chi^2$-$F$-localization approach to inference in the empirical Bayes problem is not new. \citet{lord1975empirical} and \citet{lord1976interval} considered the Binomial problem with \smash{$Z_i \cond \mu_i \; \sim \text{Binom}(N, \mu_i)$} for $N \in \NN$ and \smash{$\gcal = \pp([0,1])$}. They suggested to form confidence intervals for the posterior mean \smash{$\theta_{\gprior}(\zo) = \EE[\gprior]{\mu \cond \Zo=\zo}$} by the $F$-localization approach~\eqref{eq:nbhood_worst_case} with \smash{$\ff_n(\alpha)$} as in~\eqref{eq:floc_chisq}.\footnote{It may seem surprising that~\citet{lord1975empirical} consider only the case of a Binomial likelihood, $\gcal = \pp([0,1])$ and posterior mean estimands. One reason is that they devise a numerical scheme for computing~\eqref{eq:nbhood_worst_case} that relies on these choices, cf. Section~\ref{subsec:lord_dataset}.}
For the $F$-localization intervals in the introductory Poisson example (Table~\ref{tab:bichsel_intro}), we used the $\chi^2-F$-localization with categories $0,\dotsc,4$ and grouping all observations $Z_i \geq 5$ as a sixth category. 

For the $\chi^2$-$F$-localization, the Charnes-Cooper transformation~\eqref{eq:charnescooper} is not directly applicable, yet the resulting optimization problem is quasi-convex and tractable; see Supplement~\ref{sec:worst_case_convex} for details.

\section{Inference for linear functionals of $\gprior$}
\label{sec:linear_functionals}

Our next goal is to develop the AMARI approach for targeted inference about $\theta_{\gprior}(\zo)$. However, as a preliminary for this task, we need to develop some general results on inference for general linear functionals $L(\gprior)$ in the empirical Bayes problem; and this will be the focus of this Section. Formally, $L(\cdot)$ is a map from $\gcal \to \RR$, that is linear in $\gprior$, i.e.,
\begin{equation}
\label{eq:lin_functional}
L\p{   \lambda \gprior + (1-\lambda) \widetilde{\gprior}} =\lambda L(\gprior) + (1-\lambda) L(\widetilde{\gprior}) \text{ for all } \gprior, \widetilde{\gprior} \in \gcal,\,\,\lambda \in [0,1].
\end{equation} 
The main reason we are interested in confidence intervals for $L(\gprior)$ is that we will use these as building blocks of the \Amari~intervals for \smash{$\theta_{\gprior}(\zo)$} in Section~\ref{sec:amari}. Nevertheless, the class~\eqref{eq:lin_functional} includes functionals that are interesting in their own right.  Some examples of linear functionals of interest include \smash{$L(\gprior) = \PP[\gprior]{\mu = 0}$} (the proportion of null effects), \smash{$L(\gprior) = \PP[\gprior]{\mu \geq 0}$} (the proportion of non-negative effects), and \smash{$L(\gprior) = \EE[\gprior]{\mu^2}$} (the second moment of the prior). Inference for \smash{$\PP[\gprior]{\mu \geq 0}$} has been considered for example by~\citet{es2005asymptotic, dattner2011deconvolution, efron2016empirical}.  \citet{greenshtein2018application} and \citet{brennan2020estimating} form confidence intervals for \smash{$L(\gprior)$} using constructions that are analogous to the $F$-localization intervals developed in this work.\footnote{Solving~\eqref{eq:nbhood_worst_case} for linear functionals is typically more straightforward compared to ratio functionals~\eqref{eq:ratio}. For example, the \citet{charnes1962programming} transformation is not required.} Such $F$-localization intervals have simultaneous coverage over all possible choices of (linear) functionals, but can be overly wide for a specific linear functional \smash{$L(\gprior)$}. Instead, the intervals we develop in this section are targeted towards a specific linear functional and so can be shorter.

\subsection{Affine minimax inference for linear functionals}
\label{subsec:linear_inference}

Our key idea is to estimate the linear functional \smash{$L(\gprior)$} of $\gprior$ with an \emph{affine} estimator, i.e., an estimator of the form \smash{$\hL = \sum Q(\Zo_i)/n$}~\eqref{eq:Q}. The class of affine estimators is convenient because it enables explicit control of the worst case bias~\eqref{eq:bias_est} and it is broad enough to include kernel density estimators as in~\eqref{eq:kde} and the Fourier estimators of~\citet{butucea2009adaptive, pensky2017minimax}. The latter provably attain minimax optimal rates for estimation of linear functionals of $\gprior$ in the empirical Bayes problem, when $p(\cdot \cond \mu)$ is a smooth location family.

 We choose $Q(\cdot)$ in a purely computational and data-driven way. To do so, we first construct a pilot $F$-localization \smash{$\ff_n = \ff_n(\alpha_n)$} with $\alpha_n \to 0$ and a pilot estimate \smash{$\barf_n(\cdot)$} of the marginal density \smash{$f_{\gprior}(\cdot)$}~\eqref{eq:marginal_density}. $\ff_n$ could be, for example, any of the $F$-localizations described in Section~\ref{sec:Flocalization}. For \smash{$\barf_n(\cdot)$} we use the Kolmogorov-Smirnov minimum distance estimator, that was studied in the empirical Bayes problem by~\citet{deely1968construction} and~\citet{heinrich2018strong}:\footnote{Case-by-case constructions would be possible here too or one could use the NPMLE.}
\begin{equation}
\label{eq:minimum_ks_fhat}
\barf_n(\zo) = f_{\hG_n}(\zo),\;\; \hG_n \in \argmin_{\gprior \in \gcal}\cb{\sup_{t \in \mathbb R}\abs{F_{\gprior}(t) - \widehat{F}_n(t)}}.
\end{equation}
$\ff_n, \barf_n$ enable us to navigate a bias-variance trade-off in choosing $Q(\cdot)$. $\barf_n$ facilitates estimating the variance of any fixed $Q(\cdot)$, through the quadratic form (in $Q(\cdot)$),
\begin{equation}
\hVar[\barf_ n]{Q} = \int Q^2(\zo) \barf_n(\zo) d\lambda(\zo) - \p{\int Q(\zo) \barf_n(\zo) d\lambda(\zo)}^2.
\end{equation}
$\ff_n$ facilitates computation of the $F$-localized worst-case bias~\eqref{eq:bias_est} of $Q(\cdot)$ among all priors $\gprior_n \in \gcal_n := \gcal(\ff_n)$~\eqref{eq:localized_gcal}. With these two ingredients, we choose $Q(\cdot) = Q_n(\cdot)$ in a data-driven way, by minimizing the localized worst-case bias, subject to controlling the estimated variance of $Q(\cdot)$:
\begin{equation}
\label{eq:minimax_problem_tractable1}
\underset{Q(\cdot) \in \RR}{\text{minimize}}\;\; \sup_{\gprior \in \gcal_n}\cb{\Bias[\gprior]{Q,L}^2} \text{ s.t. } \frac{1}{n}\hVar[\barf_ n]{Q} \leq \Gamma_n, \;  Q(\cdot)  \text{ constant in } \RR\setminus [-M,\,M].
\end{equation}
Here, \smash{$M>0$} is a (large) constant, and we restrict attention to functions \smash{$Q(\cdot)$} that are constant outside the interval \smash{$[-M, \, M]$} to avoid regularity issues at infinity and so that our inference is not unduly sensitive to outliers. $\Gamma_n >0$ is a hyperparameter that controls the bias-variance trade-off. For smaller values of $\Gamma_n$, we enforce that the $Q(\cdot)$ solving~\eqref{eq:minimax_problem_tractable1} takes on smaller values of \smash{$\hVar[\barf_ n]{Q}$}, at the cost of potentially increasing the $F$-localized worst-case bias. We explain how we choose $\Gamma_n$ below, after first outlining how we solve~\eqref{eq:minimax_problem_tractable1}.

First, to (formally) enforce $Q(\cdot)$ to be constant outside $[-M,M]$ as in~\eqref{eq:minimax_problem_tractable1}, we (formally) censor $\Zo$ in~\eqref{eq:EB} and define
\begin{equation}
\label{eq:tilde_Z}
\Zo^M = \Zo \text{ if } \Zo \in [-M,M],\; \Zo^M =  \triangleleft \text{ if } \Zo < -M,\; \Zo^M = \triangleright \text{ if } \Zo > M.
\end{equation}
In view of~\eqref{eq:tilde_Z}, we only need to define $Q(\cdot)$ on the set $\cb{\triangleleft, \triangleright}\cup [-M,M]$. $\Zo^M$ has conditional density \smash{$p^M(z \cond \mu) = p(z \cond \mu)$} for $z \in [-M,M]$, \smash{$p^M(\triangleleft \cond \mu) = \int_{(-\infty,M)} p(z \cond \mu) d\lambda(z)$} and \smash{$p^M(\triangleright \cond \mu)=\int_{(M,\infty)} p(z \cond \mu) d\lambda(z)$} with respect to the measure \smash{$\lambda^M =\delta_{\triangleleft} + \delta_{\triangleright} + \lambda$}, where $\delta_x$ is a point mass at $x$. The true marginal density \smash{$f_{\gprior}^M$} (and estimated density \smash{$\barf_n^M$}) of \smash{$Z^M$} is supported on \smash{$\cb{\triangleleft, \triangleright}\cup [-M,M]$} with \smash{$f_{\gprior}^M(\triangleleft) = \int_{(-\infty,M)}f(\zo)d\zo$}, \smash{$\barf_{n}^M(\triangleleft) = \int_{(-\infty,M)}\barf_n(\zo)d\zo$}, and similarly for $\triangleright$ and $\zo \in [-M,M]$.

To solve~\eqref{eq:minimax_problem_tractable1} we build upon a construction of~\citet{donoho1994statistical}, who formalizes a powerful heuristic due to Charles Stein on hardest one-dimensional subproblems. We refer the interested reader to Supplement~\ref{subsec:stein_heuristic} for details and proofs in the context of our application and also to~\citet{donoho1989hardest,donoho1994statistical,low1995bias,armstrong2018optimal} and references therein. The consequence of interest here is that to solve~\eqref{eq:minimax_problem_tractable1}, it suffices to solve the following surrogate optimization problem: 
\begin{equation}
\label{eq:continuous_modulus_problem}
\begin{split}
\sup\bigg\{& L(\gprior_1) - L(\gprior_{-1}) \, \cond \, \gpriorL, \gpriorR \in \gcal_n,\, \int  \frac{ \big(f^M_{\gprior_1}(\zo) - f^M_{\gprior_{-1}}(\zo)\big)^2}{\barf_n^M(\zo)}\ d\lambda^M(\zo)  \leq  \frac{\delta^2}{n} \bigg\}.
\end{split}
\end{equation}
The surrogate optimization problem is parameterized by another hyperparameter $\delta$\footnote{$\delta$ maps to the hyperparameter $\Gamma_n$ of \eqref{eq:minimax_problem_tractable1}, see below.} and it is a second order conic program (SOCP)~\citep{boyd2004convex} that is tractable by modern conic optimizers, such as MOSEK~\citep{mosek}.\footnote{See Supplement~\ref{sec:amari_computation} for implementation details including discretization considerations.} The value of the supremum in~\eqref{eq:continuous_modulus_problem} is called the modulus of continuity $\omega_n(\delta)$ at $\delta>0$. 
We say that the modulus problem~\eqref{eq:continuous_modulus_problem} is solvable at $\delta>0$ if there exist
\smash{$\gprior_1^{\delta},\gprior_{-1}^{\delta} \in \gcal_n$} such that \smash{$|L(\gprior_1^{\delta})|, |L(\gprior_{-1}^{\delta})| < \infty$} and
\begin{equation}
\label{eq:modulus_solutions}
L(\gprior_1^{\delta}) - L(\gprior_{-1}^{\delta}) \,= \, \omega_n(\delta),\;\; n\cdot \int  \p{f^M_{\gprior^{\delta}_1}(\zo) - f^M_{\gprior^{\delta}_{-1}}(\zo)}^2\Big/\barf_n^M(\zo)\ d\lambda^M(\zo) \, =\, \delta^2, 
\end{equation}
and we call \smash{$\gprior_1^{\delta},\gprior_{-1}^{\delta}$} solutions of $\omega_n(\delta)$. \smash{$\gprior_1^{\delta},\gprior_{-1}^{\delta}$} are close observationally, that is their marginal distributions have distance at most $\delta/\sqrt{n}$  in terms of the `pseudo'-$\chi^2$-distance in~\eqref{eq:continuous_modulus_problem}, and they exhibit the largest separation of the linear functional $L(\gprior)$. These two priors determine the worst-case optimal $Q(\cdot)$ in~\eqref{eq:minimax_problem_tractable1} at a specific value of $\Gamma_n$ that depends on $\delta$, i.e., $\Gamma_n=\Gamma_n(\delta)$. Let \smash{$\gprior_0^{\delta}  = (\gprior_1^{\delta} + \gprior_{-1}^{\delta})/2$} and define \smash{$Q(\cdot) = Q(\cdot\,\,; \delta) = Q(\cdot\,\,; \delta, \omega_n'(\delta), \gprior_1^{\delta}, \gprior_{-1}^{\delta})$} as,
\begin{equation}
\label{eq:optimal_Q}
\frac{n\cdot \omega'_n(\delta)}{\delta}\cb{ \frac{ f^M_{\gprior_1^{\delta}}(\cdot) - f^M_{\gprior_{-1}^{\delta}}(\cdot)}{\barf^M(\cdot)} \, - \, \displaystyle{\int} \frac{\p{f^M_{\gprior_1^{\delta}}(\zo) - f^M_{\gprior_{-1}^{\delta}}(\zo)}f^M_{\gprior_{0}^{\delta}}(\zo)}{\barf^M(\zo)}d\lambda^M(\zo)} \, + \, L(\gprior_0^{\delta}).
\end{equation}
Here $\omega'_n(\delta)$ is the derivative of $\omega_n(\cdot)$ at $\delta$, in case it exists, or an element of the superdifferential of $\omega_n(\cdot)$ at $\delta$ otherwise.\footnote{
\label{footnote:superdifferential}
That is, $\omega_n'(\delta)$  satisfies, $\omega_n(\tilde{\delta}) \leq \omega_n(\delta) \, + \, \omega_n'(\delta) (\tilde{\delta} - \delta) \;\text{ for all }\; \tilde{\delta} >0$. Such an element exists, because $\omega_n(\delta)$ is concave in $\delta>0$~\citep{rockafellar1970convex}. We provide details in Supplement~\ref{sec:modulus_properties}.}  $Q(\cdot, \delta)$ from~\eqref{eq:optimal_Q} is optimal for the min-max problem~\eqref{eq:minimax_problem_tractable1} at \smash{$\Gamma_n = \omega_n'(\delta)^2$}.  The worst-case squared-bias \smash{$\hB^2$}~\eqref{eq:bias_est} of $Q(\cdot)$ over $\gcal_n$ is equal to \smash{$(\omega_n(\delta)-\delta \omega_n'(\delta))^2/4$}. To navigate the bias-variance trade-off, we allow $\delta=\delta_n$ to vary with $n$, and we choose $\delta_n$ as the minimizer of the worst-case mean squared error, 
\begin{equation}
\label{eq:delta_MSE}
 \p{\omega_n(\delta)-\delta \omega_n'(\delta)}^2/4 \,+\, \omega_n'(\delta)^2 \;= \; \sup \cb{ \Bias[\gprior]{Q(\cdot,\delta),L}^2 : \gprior \in \gcal_n} \, + \, \Gamma_n(\delta),
\end{equation} 
among all $\delta \in \Delta\subset (0,\infty)$, for a set $\Delta$ bounded away from $0$ and $\infty$,\footnote{In our implementation, we use the concrete choice $\Delta=\Delta_{\text{grid}} := \cb{0.2, 0.4, \dotsc, 6.5, 6.7}$.} for which the modulus problem is solvable. We then construct $Q(\cdot) = Q(\cdot, \delta_n)$, which is optimal for~\eqref{eq:minimax_problem_tractable1} with \smash{$\Gamma_n = \Gamma_n(\delta_n) = \omega_n'(\delta_n)^2$} and finally we form the confidence interval for $L(\gprior)$ as described in Section~\ref{subsubsec:amari}. 

Algorithm~\ref{alg:linear_ci} summarizes our proposal for inference of linear functionals $L(\gprior)$. We prove the asymptotic coverage of the proposed confidence intervals by leveraging the representation of $Q(\cdot)$ in~\eqref{eq:optimal_Q} and by verifying a central limit theorem for \smash{$\hL$}~\eqref{eq:Q}.

\begin{algorithm}[t]
  \caption{Affine minimax confidence intervals for linear functionals $L(\gprior)$. \label{alg:linear_ci}}
  Form a pilot estimate $\barf_n(\cdot)$ of the marginal density $f_{\gprior}(\cdot)$ as in~\eqref{eq:minimum_ks_fhat} and a pilot $F$-localization~\eqref{eq:dbn_nbhood} \smash{$\ff_n = \ff_n(\alpha_n)$}.\;
  Choose $\delta_n$ as in~\eqref{eq:delta_MSE}.\;
  Solve the modulus problem~\eqref{eq:continuous_modulus_problem} at $\delta_n$ and use the solution to compute \smash{$Q(\cdot)$} as in~\eqref{eq:optimal_Q}, as well as its worst-case bias \smash{$\hB$}.\;
  Form the estimate \smash{$\hL$} of $L(\gprior)$ as in~\eqref{eq:Q} and its estimated variance \smash{$\hV$} as in~\eqref{eq:sample_var_est}.\;
  Form bias-aware confidence intervals as in~\eqref{eq:im_iw_ci}.\;
\end{algorithm}

\begin{theo}[Central limit theorem for affine minimax estimator] \label{theo:lin_functional_clt} 
Assume that for all \smash{$\widetilde{\gprior} \in \gcal$}, the linear functional \smash{$L(\widetilde{\gprior})$} is well-defined with \smash{$\sup_{\widetilde{\gprior}} |L(\widetilde{\gprior})| < \infty$} and that \smash{$f_{\widetilde{\gprior}}^M(\cdot) \in \mathcal{L}^2(\lambda^M)$}. Furthermore, assume that,
\begin{enumerate}[label=\Alph*.,leftmargin=*]
\item For each $n$, the modulus problem~\eqref{eq:continuous_modulus_problem} has solutions $\gpriorL^{\delta_n},\gpriorR^{\delta_n}$ at $\delta_n \in [\delta^{\ell}, \delta^u]$, where  $\delta^{\ell}, \delta^u \in (0, \infty)$ are fixed (i.e., do not change with $n$).
\item  $\gprior$ lies in the convex set of distributions $\gcal$ ($\gprior \in \gcal$). 
\item There exists $\eta>0$ s.t. $\inf_{\zo \in \cb{\triangleleft, \triangleright}\cup[-M,M]} \cb{ f_G^M(\zo)} > \eta$. 
\item  $\barf_n^M$ is a density ($\int \barf_n^M(\zo)d\lambda^M(\zo)=1$ and $\barf_n^M \geq 0$). It holds that $\PP[G]{A_n} \to 1$, where $A_n$ is the event on which,
\begin{equation} 
\label{eq:pilot_and_localization_quality}  
\NormInline{\barf^M_n(\cdot)-f^M_{\gprior}(\cdot)}_{\infty} \leq c_n,\;\; \NormInline{f^M_{\gpriorL^{\delta_n}}(\cdot)-f^M_{\gpriorR^{\delta_n}}(\cdot)}_{\infty} \leq c_n, \;\; F_{\gprior} \in \ff_n,
\end{equation}
for a sequence of constants $c_n \to 0$ as $n \to \infty$.
\item $\ff_n$ and $\barf_n^M$ are independent of $Z_1,\dotsc,Z_n$.
\end{enumerate}
Then, letting $Q(\cdot)=Q(\cdot;\,\, \delta_n)$~\eqref{eq:optimal_Q} and $\hL$~\eqref{eq:Q} the affine estimator of the linear functional $L(G)$, it holds that, 
\begin{equation*}
\label{eq:CLT}
\left(\hL - L(\gprior) - \Bias[\gprior]{Q,L}\right)\,\Big/\,\hV^{1/2} \; \xrightarrow[]{\mathcal{D}} \; \mathcal{N}(0,1),\;\; \mathbb P_{\gprior}[\lvert\Bias[\gprior]{Q,L}\rvert \leq \hB] \to 1 \text{ as } n\to \infty,
\end{equation*}
where \smash{$\Bias[\gprior]{Q,L}$} is defined in~\eqref{eq:bias_est}. It follows that the intervals~\eqref{eq:im_iw_ci} provide asymptotically correct coverage of the target $L(\gprior)$, i.e., $\liminf_{n \to \infty} \PP[\gprior]{ L(\gprior) \in \ii_\alpha} \geq 1-\alpha$.\footnote{Although the result stated only holds elementwise, we can obtain a uniform statement
in the sense of, e.g., \citet{robins2006adaptive} by adding slightly more constraints on the class $\gcal$.
Specifically, consider \smash{$\gcal^{\eta} = \cb{ \gprior \in \gcal \,:\, \inf_{\zo \in \cb{\triangleleft, \triangleright}\cup[-M,M]} \cb{ f_G^M(\zo)} > \eta}$}, for some $\eta >0$, i.e., qualitatively, prior distributions for which the induced marginal density cannot vanish anywhere.
The proof of Theorem \ref{theo:lin_functional_clt} implies that
then the above statements apply uniformly over $\gcal^{\eta}$, \smash{$\liminf_{n \to \infty} \inf\cb{\PP[\gprior]{ L(\gprior) \in \ii_\alpha} : \gprior \in \gcal^{\eta}} \geq 1-\alpha,$}
provided that $\PP[\gprior]{A_n} \to 1$ uniformly in $\gprior \in \gcal^{\eta}$, where $A_n$ has been defined in the statement of Theorem~\ref{theo:lin_functional_clt}.}

\end{theo}
We emphasize that $Q(\cdot)$ changes with $n$, and so, the central limit theorem above is that of a triangular array. The statistical assumption driving Theorem~\ref{theo:lin_functional_clt} is Assumption B, namely that model~\eqref{eq:EB} holds with $G \in \gcal$. 
Using a good choice of {\textprior} class $\gcal$ is critical, and we discuss this choice further in Section~\ref{sec:gcal}. The rest of the assumptions are under control of the analyst and may be verified before any data analysis is conducted. Assumption A guarantees that we can solve~\eqref{eq:minimax_problem_tractable1} by convex programming, cf. Supplement~\ref{subsec:stein_heuristic}, and Assumption C is an overlap condition. Assumption D concerns the quality of the pilot localization $\ff_n$ and pilot density estimator $\barf_n$, while Assumption E requires that both \smash{$\ff_n$} and \smash{$\barf_n$} are independent of $Z_1,\dotsc,Z_n$. Assumption E holds if we use sample-splitting. Following \citet{hajek1962asymptotically} and~\citet{bickel1982adaptive} we demonstrate that it suffices to retain an asymptotically vanishing fraction of the $Z_i$ to estimate $\ff_n,\barf_n$. 
\begin{prop}
\label{prop:applications}
Suppose we use $k=k_n$ samples from model~\eqref{eq:EB} to construct $\ff_n(\alpha_n)$, $\barf_n$ and the remaining $n-k_n$ samples for step 4 of Algorithm~\ref{alg:linear_ci}. Suppose further that $k/n \to 0$, $\alpha_n \to 0$ and $k \cdot \alpha_n \to \infty$ as $n\to \infty$. Then, Assumption D of Theorem~\ref{theo:lin_functional_clt}, is satisfied in the following cases:
\begin{enumerate}[leftmargin=*]
\item \textbf{Gaussian likelihood:} $Z \cond \mu  \sim \nn(\mu, \, \sigma^2)$, and we use the DKW~\eqref{eq:DKW}
 or the Gauss-$F$-localization~\eqref{eq:floc_kde} and $\barf_n$ is the minimum distance estimator~\eqref{eq:minimum_ks_fhat}.
\item \textbf{Binomial or Poisson likelihood:} $Z\cond \mu  \sim \text{Binom}(N, \, \mu)$ or $\sim \Poisson{\mu}$, $\gprior$ is not degenerate\footnote{That is, $\PP[\gprior]{\mu \in \cb{0,1}} < 1$ in the Binomial case, and $\PP[\gprior]{\mu = 0} < 1$ in the Poisson case.} and we use the DKW~\eqref{eq:DKW}
 or the $\chi^2$-$F$-localization~\eqref{eq:floc_chisq} (wherein, in the Poisson case, we treat all observations $> M$ as a single category) and $\barf_n$ is the minimum distance estimator~\eqref{eq:minimum_ks_fhat}.
\end{enumerate}
\end{prop}

In practice, we use the full data twice and do not sample-split; we have not observed any overfitting or loss of coverage thereby.

\section{Pointwise confidence intervals with \Amari}
\label{sec:amari}

In this section we return to our main task of forming confidence intervals for empirical Bayes estimands $\theta_{\gprior}(\zo)$ and discuss our second approach, \Amari~(Affine Minimax Anderson--Rubin Intervals). In contrast to the $F$-localization intervals, the \Amari~intervals are targeted towards a specific empirical Bayes estimand and have a pointwise (rather than simultaneous) coverage guarantee. The upshot is that \Amari~intervals can be substantially shorter.

The starting point for \Amari~is~\eqref{eq:ratio}, i.e., the fact that we can write the empirical Bayes estimand \smash{$\theta_{\gprior}(\zo)$} as a ratio of linear functionals of $G$, \smash{$a_{\gprior}(\zo)/f_{\gprior}(\zo)$}. Next, fix $c \in \RR$ and write,
\smash{$L(\gprior)\,:=\,\theta_{\gprior}^{\text{lin}}(\zo; c)\,:=\, a_{\gprior}(\zo) - c f_{\gprior}(\zo)$} as in~\eqref{eq:hypothesis_test_fractional}.
One may directly verify that $L(\cdot)$ is a linear functional of $\gprior$, as defined in~\eqref{eq:lin_functional}. 
The \emph{Affine Minimax} component of the \Amari~acronym refers to the fact that we will use the affine minimax approach of Section~\ref{subsec:linear_inference} to form confidence intervals \smash{$\ii_{\alpha}^{\text{lin}}(\zo; c)$} for \smash{$\theta_{\gprior}^{\text{lin}}(\zo; c)$},
treating the latter as a generic linear functional \smash{$L(\gprior)$}. The \emph{Anderson-Rubin} component of \Amari~enables us to construct confidence intervals for $\theta_{\gprior}(\zo)$ by lifting our intervals for $\theta_{\gprior}^{\text{lin}}(\zo; c)$ following the approach of~\citet{noack2019bias}, who in turn build upon \citet{anderson1949estimation} and \citet{fieller1940biological,fieller1954some}. 

The following Corollary captures the basic idea of our approach:
\begin{coro}
\label{coro:naive_ar}
Suppose that for each $c \in \RR$, we construct a confidence interval $\ii_{\alpha}^{\text{lin}}(\zo; c)$ for $\theta_{\gprior}^{\text{lin}}(\zo; c)$~\eqref{eq:hypothesis_test_fractional} following Algorithm~\ref{alg:linear_ci}. Suppose furthermore that the assumptions of Theorem~\ref{theo:lin_functional_clt} hold for the linear functional \smash{$L(\gprior)=\theta_{\gprior}^{\text{lin}}(\zo; c^*)$}, where \smash{$c^* = \theta_{\gprior}(\zo)$}. Then,
\begin{equation}
\label{eq:naive_ar_set}
\liminf_{n \rightarrow \infty} \PP[\gprior]{\theta_{\gprior}(\zo) \in \mathcal{S}_\alpha(\zo)} \geq 1 - \alpha, \;\;\; \mathcal{S}_{\alpha}(\zo) = \cb{c \in \RR \cond  0 \in \ii_{\alpha}^{\text{lin}}(\zo;c)}.
\end{equation}
\end{coro}
\begin{proof}
By definition of $c^*$ and $\mathcal{S}_\alpha(\zo)$, it holds that \smash{$\theta_{\gprior}(\zo) \in \mathcal{S}_\alpha(\zo) \; \Longleftrightarrow \; 0 \in \ii_{\alpha}^{\text{lin}}(\zo; c^*)$}.
On the other hand, \smash{$\theta_{\gprior}^{\text{lin}}(\zo; c^*) = a_{\gprior}(\zo) - c^* f_{\gprior}(\zo) =  a_{\gprior}(\zo) - (a_{\gprior}(\zo)/f_{\gprior}(\zo)) f_{\gprior}(\zo) = 0,$}
and therefore \smash{$\PP[\gprior]{\theta_{\gprior}(\zo) \in \mathcal{S}_\alpha(\zo)} = \PP[\gprior]{\theta_{\gprior}^{\text{lin}}(\zo; c^*) \in   \ii_{\alpha}^{\text{lin}}(\zo;c^*)}$}. We conclude by Theorem~\ref{theo:lin_functional_clt}.
\end{proof}
We provide Corollary~\ref{coro:naive_ar} for intuition. However, the confidence set $\mathcal{S}_\alpha(\zo)$ from~\eqref{eq:naive_ar_set} has some disadvantages. First, $\mathcal{S}_\alpha(\zo)$ will in general not be an interval. Second, computing the interval $\ii_{\alpha}^{\text{lin}}(\zo;c)$, even for a single $c$, is computationally demanding and requires the solution of~\eqref{eq:continuous_modulus_problem} along a grid of $\delta$ values~\eqref{eq:delta_MSE}, and so, computing $\ii_{\alpha}^{\text{lin}}(\zo;c)$ for `all' values of $c$ is not computationally tractable. Instead, in our actual implementation of \Amari, which we describe in Section~\ref{subsec:andersonrubin} below, we use an `accelerated' Anderson-Rubin procedure that is computationally streamlined (we only need to form $\ii_{\alpha}^{\text{lin}}(\zo;c)$ for two values of $c$), and leads to a confidence interval $\ii_{\alpha}(\zo)$ for $\theta_{\gprior}(\zo)$, rather than a confidence set.\footnote{The `accelerated' Anderson-Rubin  approach could be fruitful in other settings; for example it could allow replacing the local linear estimators in the fuzzy regression discontinuity approach of~\citet{noack2019bias} by the affine minimax estimators of~\citet{imbens2017optimized}.}

\subsection{Implementation of Anderson-Rubin inversion for \Amari}
\label{subsec:andersonrubin}
We now describe our actual implementation of \Amari. The key intuition is that we start with a preliminary interval such that $\theta_{\gprior}(\zo) \in [c^{\ell}, c^u]$ with high probability, and then find the affine minimax $Q^{\ell},Q^{u}$~\eqref{eq:minimax_problem_tractable1} for $\theta_{\gprior}^{\text{lin}}(\zo; c^{\ell})$, resp. $\theta_{\gprior}^{\text{lin}}(\zo; c^{u})$. Then, for any other $c$, say $c=\kappa c^{\ell} + (1-\kappa)c^u,\; \kappa \in (0,1)$, instead of resolving~\eqref{eq:minimax_problem_tractable1}, we use $Q^c = \kappa Q^{\ell} + (1-\kappa)Q^u$. The variance of $Q^c$ then can be directly computed from the covariance of $Q^{\ell},Q^u$, and their individual variances, while the worst-case bias of $Q^c$ can be upper bounded by the convex combination \smash{$\hB^{c} =  \kappa\hB^{\ell} + (1-\kappa)\hB^{u}$} of the worst-case biases of $Q^{\ell}$ and $Q^u$. Algorithm~\ref{alg:amari_ci} describes all steps of \Amari. Step 4 can be computed efficiently using grid search, since the evaluation of $\widetilde{\ii}_{\alpha}(\zo; c)$ for different values of $c \in [c^{\ell},c^u]$ is fast.

\begin{algorithm}[t]
  \caption{\Amari~confidence intervals for empirical Bayes estimands $\theta_{\gprior}(\zo)$ \label{alg:amari_ci}}
  Construct a pilot $F$-localization $\ff_n = \ff_n(\alpha_n)$ at level $\alpha_n$, as in Step 1 of Algorithm~\ref{alg:linear_ci}. Let \smash{$[c^{\ell}, c^u]$} be the $F$-localization interval~\eqref{eq:nbhood_worst_case} for \smash{$\theta_{\gprior}(\zo)$} based on $\ff_n$.\;
  Apply Algorithm~\ref{alg:linear_ci} to the linear functionals $\theta_{\gprior}^{\text{lin}}(\zo; c^{\ell})$  and $\theta_{\gprior}^{\text{lin}}(\zo; c^{u})$ defined in~\eqref{eq:hypothesis_test_fractional}. Let \smash{$Q^{\ell}, Q^{u}$} be the corresponding affine minimax kernels, \smash{$\hL^{\ell}, \hL^{u}$} the point estimates~\eqref{eq:Q}, \smash{$\hB^{\ell}, \hB^u$} the worst case biases~\eqref{eq:bias_est}, \smash{$\hV^{\ell},\hV^u$} the variances and $\hCov{\hL^{\ell}, \hL^{u}} = \frac{1}{n(n - 1)}\sqb{\sum_{i = 1}^n Q^{\ell}(\Zo_i)Q^{u}(\Zo_i) - \p{\sum_{i=1}^n Q^{\ell}(\Zo_i)}\p{\sum_{i=1}^n Q^u(\Zo_i)}\big/n}.$\;
  For $c=\kappa c^{\ell} + (1-\kappa)c^u, \; \kappa \in [0,1]$, let \smash{$\widehat{L}^{c} = \kappa \widehat{L}^{\ell} + (1-\kappa)\widehat{L}^{u}$}, \smash{$\hB^{c} =  \kappa\hB^{\ell} + (1-\kappa)\hB^{u}$}, \smash{$\hV^c = \kappa^2 \hV^{\ell} + (1-\kappa)^2\hV^{u} + 2\kappa(1-\kappa)\hCov{\hL^{\ell}, \hL^{u}}$} and $\widetilde{\ii}_{\alpha}(\zo; c) = \widehat{L}^{c} \pm t_\alpha(\hB^c, \hV^c)$ as in~\eqref{eq:im_iw_ci}.\;
  Report the interval $\ii_\alpha(\zo) = \sqb{\inf \mathcal{C},\, \sup \mathcal{C}} \cap [c^{\ell},\, c^u]$, $\,\mathcal{C} = \{c \in [c^{\ell}, c^u]:\, 0 \in \widetilde{\ii}_{\alpha}(\zo; c)\}$.
\end{algorithm}

 As a consequence of Theorem~\ref{theo:lin_functional_clt}, we can now prove, that the \Amari~confidence intervals asymptotically cover the empirical Bayes estimand \smash{$\theta_{\gprior}(\zo)$}.
\begin{theo}[Coverage of \Amari~intervals]
\label{theo:amari}
 Consider the confidence intervals constructed in Algorithm~\ref{alg:amari_ci}. Suppose the pilot $F$-localization interval endpoints $c^{\ell}, c^{u}$ are finite for all $n$ and let the assumptions of Theorem~\ref{theo:lin_functional_clt} hold for $L(\gprior)=\theta_{\gprior}^{\text{lin}}(\zo; c^{\ell})$ and $L(\gprior)=\theta_{\gprior}^{\text{lin}}(\zo; c^{u})$.\footnote{The proof of Theorem~\ref{theo:lin_functional_clt} uses triangular array asymptotics, and so, the linear functional $L(\gprior)$ may depend on $n$.} Furthermore, assume that $Q^{\ell}$ and $Q^{u}$ are not perfectly anticorrelated, i.e., there exists $\varepsilon > 0$, such that, \smash{$\mathbb P_{\gprior}[ \widehat{\operatorname{Cov}}[\hL^{\ell}, \hL^{u}] \big /  (\hV^{u} \hV^{\ell})^{1/2} \geq -1+\varepsilon] \to 1 \text{ as } n \to \infty$}. Then, \smash{$\liminf_{n \rightarrow \infty} \PP[\gprior]{\theta_{\gprior}(\zo) \in \ii_\alpha(\zo)} \geq 1 - \alpha.$}
\end{theo}
The additional assumption of Theorem~\ref{theo:amari} on the correlation between $Q^{\ell}$ and $Q^{u}$ is mild and can be verified from the data at hand; in applications we typically find a positive correlation.

\section{Empirical applications}
\label{sec:applications}
In this section we apply the $F$-localization and \Amari~intervals developed above in the context of two different applications; one in education and one in genomics. 

\subsection{Predicting student ability in psychometric tests}
\label{subsec:lord_dataset}

\begin{figure}
\centering
\begin{tabular}{@{}ll@{}}
a) & b)\\
\begin{adjustbox}{width=0.47\linewidth}\input{tikz_figures/lord_cressie_pdf.tikz}\end{adjustbox} & 
\begin{adjustbox}{width=0.48\linewidth}\input{tikz_figures/lord_cressie_posterior_mean.tikz}\end{adjustbox}
\end{tabular}
\caption{\textbf{Empirical Bayes confidence intervals in a psychometric test} with 20 questions taken by $12,990$ students \citep{lord1975empirical}. \textbf{a)} \smash{$\hat{f}_n(z)^{1/2}$} vs. $\zo$, where \smash{$\hat{f}_n(z)$} is the proportion of students that answered $z$ out of 20 questions correctly. \textbf{b)} $95\%$ confidence intervals for the posterior mean  $\theta_{\gprior}(\zo) = \EE[\gprior]{\mu \cond \Zo=z}$.}
\label{fig:lord_cressie}
\end{figure}
\citet{lord1975empirical} studied a dataset of scores by $n=12,990$ students on a psychological test with $N=20$ multiple choice questions (5 choices per question) and posited that $Z_i \cond \mu_i \sim \text{Binom}(20, \mu_i)$, where $\mu_i$ is the `true-score' of student $i$~\citep{lord1969estimating}. \citet{lord1976interval} were working for the Educational Testing Service (ETS) at the time and the following motivation for confidence intervals of the posterior mean $\theta_{\gprior}(\zo) = \EE[\gprior]{\mu \cond \Zo=\zo}$ with the property~\eqref{eq:CI} seems plausible: if the ETS were to use an estimate \smash{$\htheta(\zo)$} for student assessment, then it would be important to account for the uncertainty in estimating the regression function \smash{$\EE[\gprior]{\mu \cond \Zo}$} due to both variability (which can be large even for large sample sizes, e.g., $n=12,990$ in this example) and partial identification.\footnote{In the Binomial empirical Bayes problem (\smash{$p(\cdot \cond \mu) = \text{Binom}(N, \mu)$}, \smash{$N \in \mathbb N_{>0}$}) and without further restrictions on \smash{$\gcal$}, the posterior mean \smash{$\theta_{\gprior}(\zo) = \EE[\gprior]{\mu \cond \Zo=\zo}$} is only \emph{partially identified} and cannot be consistently estimated, even as $n \to \infty$. We discuss this issue further in Section~\ref{subsec:bernoulli_partial}.}

The empirical frequencies \smash{$\hat{f}_n(\zo) = \#\cb{\Zo_i = \zo}/n$} of test scores are shown in Figure~\ref{fig:lord_cressie}a). Panel b) shows three 95\% confidence intervals for $\theta_{\gprior}(\zo)$ that make no assumptions on $\gprior$, i.e., $\gprior \in \gcal = \pp([0,1])$~\eqref{eq:all_dbns}. The $\chi^2$-$F$-localization intervals~\eqref{eq:floc_chisq} were developed by \citet{lord1975empirical, lord1976interval}. We computed these intervals by grouping the lowest scores $\zo=0$ and $\zo=1$ together (to ensure the $\chi^2$-interval~\eqref{eq:floc_chisq} has the right coverage)
and then used the parametric convex programming approach from Supplement~\ref{sec:worst_case_convex} with the discretization $\pp([0,1]) \approx \pp(\Ksupport)$ and $\Ksupport$ an equidistant grid on $[0,1]$ with 300 points.
The intervals agree with the ones reported in~\citet{lord1975empirical}.
The latter used a numerical optimization routine due to Martha Stocking that leveraged the fact that in the Binomial empirical Bayes problem ($N=20$) and when $\theta_{\gprior}(z)$ is the posterior mean, then the worst case $\gprior$ in~\eqref{eq:nbhood_worst_case} must be discrete and supported on at most $11$ points. We also report intervals based on the DKW-$F$-localization, as well as the \Amari~intervals (with pilot $\chi^2$-$F$-localization using $\alpha_n = 0.01$).

We see that, for low scores $\zo$, there is substantial uncertainty. For students with score $\zo=0$, one could predict their true score as being almost $0$, or one could predict their true score as better than random guessing ($1/5$); and both predictions would be consistent with the data. On the other hand, for intermediate values of $\zo$ the intervals become substantially shorter. In this example we also observe that the DKW-$F$-Localization intervals are overly wide. The \Amari~intervals are shorter than the $\chi^2$-intervals for most $\zo$; at the cost of no longer being simultaneous.

\subsection{Identifying genes associated with prostate cancer}
\label{subsec:prostate}

Our next dataset is the `Prostate' dataset~\citep{efron2012large,singh2002gene}, by now a classic dataset used to illustrate empirical Bayes principles.
The dataset consists of Microarray expression levels measurements for $m = 6033$ genes of $52$ healthy men and $50$ men with prostate cancer.
For each gene, a t-statistic $T_i$ is calculated (based on a two-sample equal variance t-test) and $z$-scores are calculated as $\Zo_i = \Phi^{-1}(F_{100}(T_i))$, where $\Phi$ is the standard normal CDF and $F_{100}(\cdot)$ is the CDF of the t-distribution with $100$ degrees of freedom.
We posit \smash{$Z_i \cond \mu_i \approxdot \nn(\mu, 1)$}, where $\mu_i$ is the standardized effect size and is expected to be close to null for most genes~\citep{efron2001empirical}.
We seek to form confidence intervals for two empirical Bayes estimands.

Our first estimand of interest is the posterior mean, \smash{$\theta_{\gprior}(\zo) = \EE[\gprior]{\mu \cond \Zo = \zo}$}, which could be used to denoise the noisy measurements $\Zo_i$ by $\theta_{\gprior}(\Zo_i)$ and to provide estimates of $\mu_i$ that are (nearly) immune to selection bias~\citep{efron2011tweedie}. The standard empirical Bayes approach provides point estimates
of these oracle quantities by sharing information across genes, but the empirical Bayes estimation
error may be rather opaque and so it is not clear to what extent the estimates $\hEE{\mu_i \mid \Zo_i}$
eliminate selection bias. Our confidence intervals attach a measure of uncertainty to the estimation of $\theta_{\gprior}(\zo)$.

Second, we consider the local false sign rate  \smash{$\theta_{\gprior}(\zo) = \PP[\gprior]{\mu \zo \leq 0 \cond \Zo = \zo}$},
which measures the posterior probability that the sign of 
an observed signal $\Zo_i$ disagrees with the sign of the true effect $\mu_i$.
Local false sign rates provide a principled approach to multiple testing without assuming that the distribution
$\gprior$ of the effect sizes $\mu_i$ is spiked at 0, and form an attractive alternative to the local false
discovery rate, $\text{lfdr}(\zo) = \PP[\gprior]{\mu_i = 0 \cond \Zo_i = \zo}$, without requiring a sharp null hypothesis \citep{stephens2016false, zhu2018heavy}. Inferential emphasis is thus placed on whether we can reliably detect the direction of an effect. Below, for ease of visualization,  we report results on the (substantively equivalent)
quantity $\theta_{\gprior}(\zo)=\PP[\gprior]{\mu_i \geq 0 \cond \Zo_i=\zo}$ (instead of $\PP[\gprior]{\mu \zo \leq 0 \cond \Zo = \zo}$) so that the resulting confidence bands are monotonic in $\zo$.

While both the posterior mean and the local false sign rate are routinely reported in the analysis of genomics datasets~\citep{stephens2016false, zhu2018heavy}, the statistical difficulty of estimating them in the Gaussian empirical Bayes model is vastly different. The posterior mean can be estimated at the \emph{quasi-parametric} rate \smash{$\log(n)^{3/4}/\sqrt{n}$} over the class $\gcal$ of priors with Lebesgue density and finite first moment~\citep{matias2004minimax}. Meanwhile, minimax point estimates for the local false sign rate over Sobolev classes of priors converge at \emph{extremely slow rates}, e.g., polynomial in \smash{$1/\log(n)$}~\citep{butucea2009adaptive,pensky2017minimax}. Consequently, we expect that our confidence intervals for the posterior mean will be substantially shorter than the ones for the local false sign rate.

We form $95\%$ confidence intervals using $3 \times 2$ methods, namely the DKW~\eqref{eq:DKW} and Gauss~\eqref{eq:floc_kde}
(with $M=3$) $F$-localization methods, 
as well as \Amari~(with Gauss-$F$-localization pilot, $\alpha_n=0.01$), each applied based on two specifications for $\gcal$. 
First, we consider the Gaussian location mixture
\begin{equation}
 \label{eq:normal_mixing_class}
\law\nn(\tau^2, \Ksupport)  := \cb{G \text{ distribution}:\; \frac{dG\hfill(\mu)}{d\lambda^{\text{Leb}}} =  \int \frac{1}{\tau}\varphi\p{\frac{\mu-u}{\tau}} d\Pi(u),\;\Pi \in \pp(\Ksupport)},
\end{equation}
where $\tau>0$, $\Ksupport \subset \RR$, $\lambda^{\text{Leb}}$ is the Lebesgue measure and $\varphi$ is the standard Gaussian density. This is a natural choice of smooth priors in the Gaussian empirical Bayes problem~\citep{magder1996smooth,cordy1997deconvolution} and the noise level $\tau$  provides an interpretable way of specifying the smoothness of the priors. Here we make the concrete choice $\gcal=\law\nn(0.25^2, [-3.3])$.\footnote{We discretize it as $\law\nn(0.25^2, \Ksupport)$ with $\Ksupport$ an equidistant grid on $[-3,3]$ of step size equal to $0.05$.}

Second, we consider Gaussian scale mixtures with mode at zero, discretized as suggested by~\citet{stephens2016false}, i.e., for $0<\tau_{\ell}<\tau_{u}$ and $\eta>1$:
\begin{equation}
 \label{eq:normal_scale_class}
\mathcal{S}\nn(\tau_{\ell}, \tau_{u}, \eta) := \cb{G \text{ dbn}:\, \frac{dG\hfill(\mu)}{d\lambda^{\text{Leb}}} =  \int\varphi\p{\frac{\mu}{\tau}}  \frac{d\Pi(\tau)}{\tau},\, \Pi \in \pp(\cb{\tau_{\ell}, \eta \cdot \tau_{\ell},  \dotsc,  \tau_{u}})}.
\end{equation}
We take \smash{$\tau_{\ell}=0.1$},  \smash{$\tau_{u} = 10.7 > (\max_i \cb{ Z_i^2 - 1})^{1/2}$} and $\eta = 1.1$.\footnote{\citet{stephens2016false} uses a coarser grid with $\eta = \sqrt{2}$.} \citet{stephens2016false} argues that the unimodal Gaussian scale mixture leads to more accurate inference provided that it holds; our intervals allow a quantitative assessment of this claim for any analyzed dataset.

\begin{figure}
\centering
\begin{tabular}{@{}lll@{}}
a) & c) & e)  \\
\begin{adjustbox}{width=0.32\linewidth}\input{tikz_figures/prostate/prostate_dkw_band.tikz}\end{adjustbox} & 
\begin{adjustbox}{width=0.32\linewidth}\input{tikz_figures/prostate/prostate_locmix_postmean.tikz}\end{adjustbox} & 
\begin{adjustbox}{width=0.32\linewidth}\input{tikz_figures/prostate/prostate_locmix_lfsr.tikz}\end{adjustbox} \\
b) & d) & f) \\
\begin{adjustbox}{width=0.32\linewidth}\input{tikz_figures/prostate/prostate_kde_band.tikz}\end{adjustbox} & 
\begin{adjustbox}{width=0.32\linewidth}\input{tikz_figures/prostate/prostate_scalemix_postmean.tikz}\end{adjustbox} & 
\begin{adjustbox}{width=0.32\linewidth}\input{tikz_figures/prostate/prostate_scalemix_lfsr.tikz}\end{adjustbox} 
\end{tabular}
\caption{\textbf{Empirical Bayes inference for the Prostate dataset}~\citep{efron2012large,singh2002gene}\textbf{. a)} Empirical distribution of the $Z_i$ and DKW-$F$-localization band. \textbf{b)} Histogram of the $Z_i$ and Gauss-$F$-localization. \textbf{c)} $95\%$ confidence intervals for the posterior mean \smash{$\EE[\gprior]{\mu \cond \Zo=\zo}$} assuming $\gprior$ is a Gaussian location, resp. \textbf{d)} scale mixture. \textbf{e)} $95\%$ confidence intervals for the local false sign rate \smash{$\EE[\gprior]{\mu \geq 0 \cond \Zo=\zo}$} assuming $\gprior$ is a Gaussian location, resp. \textbf{f)} scale mixture. }
\label{fig:prostate}
\end{figure}

Figure~\ref{fig:prostate} shows the results of the analysis. For the posterior mean, we observe that all intervals suggest that many effects are close to null and so there is substantial shrinkage towards zero.
\smash{$\EE[\gprior]{\mu \cond \Zo=\zo}$} is almost flat in the interval $[-1.5,1.5]$; and we can say so with confidence.
All four $F$-localization bands are quite similar, while \Amari~leads to substantially shorter intervals, and the improvement is more noticeable for \smash{$\gcal=\law\nn(0.25^2, [-3.3])$}.
For the local false sign rate, the intervals are long when assuming $\gprior \in \law\nn(0.25^2, [-3.3])$. 
The DKW-$F$-localization intervals perform worst, while the \Amari~and Gauss-$F$-localization intervals perform comparably (with \Amari~leading to shorter intervals only for more extreme values of $\zo$).
As explained above, long confidence intervals are expected in this case.
On the other hand, if we are willing to assume that $\gprior$ is a Gaussian scale mixture with mode at $0$, then inference for the local false sign rate is much more precise, exactly as argued by~\citet{stephens2016false}.
However, the assumption is strong, and for example it implies that the local false sign rate at $0$ is equal to $1/2$; all intervals proposed here have vanishing length in that case.

\section{Simulations}
\label{sec:simulations}
The setting of our simulations is similar to the Prostate data example in Section~\ref{subsec:prostate}. We consider model~\eqref{eq:EB} with \smash{$Z \cond \mu \; \sim \nn(\mu,1)$}, $n=5000$, and two different data-generating priors,
\begin{equation}
\begin{aligned}
\label{eq:sim_priors}
&\gprior^{\text{Spiky}} \; = \; 0.4 \nn(0, 0.25^2)  + 0.2 \nn(0, 0.5^2) + 0.2 \nn(0, 1) +  0.2 \nn(0, 2^2), \\
&\gprior^{\text{NegSpiky}} \; = \; 0.8 \nn(-0.25, 0.25^2)  + 0.2 \nn(0, 1).
\end{aligned}
\end{equation}
\smash{$\gprior^{\text{Spiky}}$} was used in the simulations of~\citet{stephens2016false} and is a unimodal symmetric prior centered at $0$, 
while \smash{$\gprior^{\text{NegSpiky}}$} is a prior with strong peak just to the left of zero, reflecting many slightly negative effects. 
Figure~\ref{fig:simulation_priors} shows the Lebesgue densities of the priors and the induced marginal densities $f_{\gprior}(\zo)$.

\begin{figure}
\centering
\begin{tabular}{@{}ll@{}}
a) & b)  \\
\begin{adjustbox}{width=0.4\linewidth}\input{tikz_figures/simulations/prior_densities.tikz}\end{adjustbox} & 
\begin{adjustbox}{width=0.4\linewidth}\input{tikz_figures/simulations/marginal_densities.tikz}\end{adjustbox}
\end{tabular}
\caption{\textbf{Priors used in simulations. a)} Lebesgue density of priors and \textbf{b)} marginal density $f_{\gprior}(\zo)$.}
\label{fig:simulation_priors}
\end{figure}

We seek to form $95\%$ confidence intervals for the posterior mean and the local false sign rate using the DKW-$F$-localization, Gauss-$F$-localization ($M=4$) and \Amari~approaches  (with pilot Gauss-$F$-localization, $\alpha_n=0.01$) and $\gcal$ equal to the Gaussian location mixture class $\law\nn(0.25^2,  [-4,4])$.

We also consider an additional plug-in baseline~\citep{efron2016empirical, narasimhan2016g}. We estimate \smash{$\hG$} by (penalized) maximum likelihood over
a flexible exponential family with a natural spline (5 degrees of freedoms) as the sufficient statistic and base measure $U[-4,4]$.
We then obtain \smash{$\htheta(\zo)$}
by applying Bayes rule with prior \smash{$\hG$}. As is standard in the literature, this baseline constructs
confidence intervals for $\theta_{\gprior}(\zo)$ using the delta method, which captures the variance
of \smash{$\htheta(\zo)$} but not its bias. Such confidence intervals are only guaranteed to cover
$\theta_{\gprior}(\zo)$ in the presence of undersmoothing or if the parametric specification is correct. See Supplement~\ref{sec:exp_family} for implementation details of this plug-in baseline.

\begin{figure}
\centering
\begin{tabular}{@{}ll@{}}
  a)   Spiky $G$ & b) Spiky $G$ \\
\begin{adjustbox}{width=0.4\linewidth}\input{tikz_figures/simulations/spiky_postmean_locmix_intervals.tikz}\end{adjustbox} & 
\begin{adjustbox}{width=0.4\linewidth}\begin{tikzpicture}[/tikz/background rectangle/.style={fill={rgb,1:red,1.0;green,1.0;blue,1.0}, draw opacity={1.0}}, show background rectangle]
\begin{axis}[point meta max={nan}, point meta min={nan}, title={}, title style={at={{(0.5,1)}}, anchor={south}, font={{\fontsize{18.2 pt}{23.66 pt}\selectfont}}, color={rgb,1:red,0.0;green,0.0;blue,0.0}, draw opacity={1.0}, rotate={0.0}}, legend style={color={rgb,1:red,0.0;green,0.0;blue,0.0}, draw opacity={0.0}, line width={1.3}, solid, fill={rgb,1:red,0.0;green,0.0;blue,0.0}, fill opacity={0.0}, text opacity={1.0}, font={{\fontsize{10.4 pt}{13.520000000000001 pt}\selectfont}}, text={rgb,1:red,0.0;green,0.0;blue,0.0}, cells={anchor={west}}, at={(0.5, 0.02)}, anchor={south}}, axis background/.style={fill={rgb,1:red,1.0;green,1.0;blue,1.0}, opacity={1.0}}, anchor={north west}, xshift={1.0mm}, yshift={-1.0mm}, width={104.68mm}, height={81.82000000000001mm}, scaled x ticks={false}, xlabel={$z$}, x tick style={color={rgb,1:red,0.0;green,0.0;blue,0.0}, opacity={1.0}}, x tick label style={color={rgb,1:red,0.0;green,0.0;blue,0.0}, opacity={1.0}, rotate={0}}, xlabel style={at={(ticklabel cs:0.5)}, anchor=near ticklabel, font={{\fontsize{14.3 pt}{18.59 pt}\selectfont}}, color={rgb,1:red,0.0;green,0.0;blue,0.0}, draw opacity={1.0}, rotate={0.0}}, xmajorgrids={false}, xmin={-3.18}, xmax={3.18}, xtick={{-3.0,-2.0,-1.0,0.0,1.0,2.0,3.0}}, xticklabels={{$-3$,$-2$,$-1$,$0$,$1$,$2$,$3$}}, xtick align={inside}, xticklabel style={font={{\fontsize{10.4 pt}{13.520000000000001 pt}\selectfont}}, color={rgb,1:red,0.0;green,0.0;blue,0.0}, draw opacity={1.0}, rotate={0.0}}, x grid style={color={rgb,1:red,0.0;green,0.0;blue,0.0}, draw opacity={0.1}, line width={0.65}, solid}, x axis line style={color={rgb,1:red,0.0;green,0.0;blue,0.0}, draw opacity={1.0}, line width={1.3}, solid}, scaled y ticks={false}, ylabel={Coverage}, y tick style={color={rgb,1:red,0.0;green,0.0;blue,0.0}, opacity={1.0}}, y tick label style={color={rgb,1:red,0.0;green,0.0;blue,0.0}, opacity={1.0}, rotate={0}}, ylabel style={at={(ticklabel cs:0.5)}, anchor=near ticklabel, font={{\fontsize{14.3 pt}{18.59 pt}\selectfont}}, color={rgb,1:red,0.0;green,0.0;blue,0.0}, draw opacity={1.0}, rotate={0.0}}, ymajorgrids={false}, ymin={0.0}, ymax={1.05}, ytick={{0.0,0.25,0.5,0.75,1.0}}, yticklabels={{$0.00$,$0.25$,$0.50$,$0.75$,$1.00$}}, ytick align={inside}, yticklabel style={font={{\fontsize{10.4 pt}{13.520000000000001 pt}\selectfont}}, color={rgb,1:red,0.0;green,0.0;blue,0.0}, draw opacity={1.0}, rotate={0.0}}, y grid style={color={rgb,1:red,0.0;green,0.0;blue,0.0}, draw opacity={0.1}, line width={0.65}, solid}, y axis line style={color={rgb,1:red,0.0;green,0.0;blue,0.0}, draw opacity={1.0}, line width={1.3}, solid}, colorbar={false}]
    \addplot[color={rgb,1:red,1.0;green,0.549;blue,0.0}, name path={867846e0-ab65-4a4d-a7a6-e380ef361d79}, draw opacity={0.7}, line width={1.3}, solid]
        table[row sep={\\}]
        {
            \\
            -3.0  1.0  \\
            -2.8  1.0  \\
            -2.6  1.0  \\
            -2.4  1.0  \\
            -2.2  1.0  \\
            -2.0  1.0  \\
            -1.8  1.0  \\
            -1.6  1.0  \\
            -1.4  1.0  \\
            -1.2  0.9975  \\
            -1.0  0.9975  \\
            -0.8  0.9975  \\
            -0.6  0.9975  \\
            -0.4  1.0  \\
            -0.2  1.0  \\
            0.0  0.9975  \\
            0.2  0.9975  \\
            0.4  1.0  \\
            0.6  1.0  \\
            0.8  1.0  \\
            1.0  1.0  \\
            1.2  1.0  \\
            1.4  1.0  \\
            1.6  1.0  \\
            1.8  1.0  \\
            2.0  1.0  \\
            2.2  1.0  \\
            2.4  1.0  \\
            2.6  1.0  \\
            2.8  1.0  \\
            3.0  1.0  \\
        }
        ;
    \addlegendentry {$\textrm{Gauss-F-Loc } (\mathcal{L}\mathcal{N})$}
    \addplot[color={rgb,1:red,0.0;green,0.0;blue,0.0}, name path={eaae2609-44e6-4cd8-a7b2-3c65200b848b}, draw opacity={0.9}, line width={1.3}, dashed]
        table[row sep={\\}]
        {
            \\
            -3.0  1.0  \\
            -2.8  1.0  \\
            -2.6  1.0  \\
            -2.4  1.0  \\
            -2.2  1.0  \\
            -2.0  1.0  \\
            -1.8  1.0  \\
            -1.6  1.0  \\
            -1.4  1.0  \\
            -1.2  1.0  \\
            -1.0  1.0  \\
            -0.8  1.0  \\
            -0.6  1.0  \\
            -0.4  1.0  \\
            -0.2  1.0  \\
            0.0  1.0  \\
            0.2  1.0  \\
            0.4  1.0  \\
            0.6  1.0  \\
            0.8  1.0  \\
            1.0  1.0  \\
            1.2  1.0  \\
            1.4  1.0  \\
            1.6  1.0  \\
            1.8  1.0  \\
            2.0  1.0  \\
            2.2  1.0  \\
            2.4  1.0  \\
            2.6  1.0  \\
            2.8  1.0  \\
            3.0  1.0  \\
        }
        ;
    \addlegendentry {$\textrm{DKW-F-Loc } (\mathcal{L}\mathcal{N})$}
    \addplot[color={rgb,1:red,0.0;green,0.0;blue,1.0}, name path={a221ab65-ec01-4f2e-bccc-fb13391dca68}, draw opacity={0.4}, line width={1.3}, solid]
        table[row sep={\\}]
        {
            \\
            -3.0  0.9625  \\
            -2.8  0.95  \\
            -2.6  0.9525  \\
            -2.4  0.965  \\
            -2.2  0.9725  \\
            -2.0  0.97  \\
            -1.8  0.975  \\
            -1.6  0.9775  \\
            -1.4  0.97  \\
            -1.2  0.9775  \\
            -1.0  0.9725  \\
            -0.8  0.9725  \\
            -0.6  0.965  \\
            -0.4  0.97  \\
            -0.2  0.9725  \\
            0.0  0.9775  \\
            0.2  0.98  \\
            0.4  0.975  \\
            0.6  0.97  \\
            0.8  0.975  \\
            1.0  0.97  \\
            1.2  0.9775  \\
            1.4  0.9775  \\
            1.6  0.965  \\
            1.8  0.97  \\
            2.0  0.975  \\
            2.2  0.9725  \\
            2.4  0.97  \\
            2.6  0.9675  \\
            2.8  0.9725  \\
            3.0  0.98  \\
        }
        ;
    \addlegendentry {$\textrm{AMARI } (\mathcal{L}\mathcal{N})$}
    \addplot[color={rgb,1:red,0.5451;green,0.0;blue,0.0}, name path={0c7c8e5f-a88b-4c78-913b-be2fee6a37b3}, draw opacity={0.8}, line width={1.3}, dotted]
        table[row sep={\\}]
        {
            \\
            -3.0  0.96  \\
            -2.8  0.8525  \\
            -2.6  0.67  \\
            -2.4  0.6125  \\
            -2.2  0.645  \\
            -2.0  0.7725  \\
            -1.8  0.8775  \\
            -1.6  0.97  \\
            -1.4  0.9075  \\
            -1.2  0.71  \\
            -1.0  0.5425  \\
            -0.8  0.505  \\
            -0.6  0.5975  \\
            -0.4  0.8025  \\
            -0.2  0.905  \\
            0.0  0.94  \\
            0.2  0.9  \\
            0.4  0.7775  \\
            0.6  0.6275  \\
            0.8  0.535  \\
            1.0  0.5525  \\
            1.2  0.7275  \\
            1.4  0.9125  \\
            1.6  0.96  \\
            1.8  0.8825  \\
            2.0  0.74  \\
            2.2  0.6325  \\
            2.4  0.585  \\
            2.6  0.6575  \\
            2.8  0.805  \\
            3.0  0.95  \\
        }
        ;
    \addlegendentry {$\textrm{Logspline}$}
    \addplot[color={rgb,1:red,0.8275;green,0.8275;blue,0.8275}, name path={41f7cb7e-d7c4-47c4-b832-5500d7db4611}, draw opacity={1.0}, line width={1.3}, dotted, forget plot]
        table[row sep={\\}]
        {
            \\
            -9.540000000000001  0.95  \\
            9.540000000000001  0.95  \\
        }
        ;
\end{axis}
\end{tikzpicture}\end{adjustbox} \\
c) NegSpiky $G$ & d) NegSpiky $G$ \\
\begin{adjustbox}{width=0.4\linewidth}\input{tikz_figures/simulations/negspiky_postmean_locmix_intervals.tikz}\end{adjustbox} &
\begin{adjustbox}{width=0.4\linewidth}\begin{tikzpicture}[/tikz/background rectangle/.style={fill={rgb,1:red,1.0;green,1.0;blue,1.0}, draw opacity={1.0}}, show background rectangle]
\begin{axis}[point meta max={nan}, point meta min={nan}, title={}, title style={at={{(0.5,1)}}, anchor={south}, font={{\fontsize{18.2 pt}{23.66 pt}\selectfont}}, color={rgb,1:red,0.0;green,0.0;blue,0.0}, draw opacity={1.0}, rotate={0.0}}, legend style={color={rgb,1:red,0.0;green,0.0;blue,0.0}, draw opacity={0.0}, line width={1.3}, solid, fill={rgb,1:red,0.0;green,0.0;blue,0.0}, fill opacity={0.0}, text opacity={1.0}, font={{\fontsize{10.4 pt}{13.520000000000001 pt}\selectfont}}, text={rgb,1:red,0.0;green,0.0;blue,0.0}, cells={anchor={west}}, at={(0.5, 0.02)}, anchor={south}}, axis background/.style={fill={rgb,1:red,1.0;green,1.0;blue,1.0}, opacity={1.0}}, anchor={north west}, xshift={1.0mm}, yshift={-1.0mm}, width={104.68mm}, height={81.82000000000001mm}, scaled x ticks={false}, xlabel={$z$}, x tick style={color={rgb,1:red,0.0;green,0.0;blue,0.0}, opacity={1.0}}, x tick label style={color={rgb,1:red,0.0;green,0.0;blue,0.0}, opacity={1.0}, rotate={0}}, xlabel style={at={(ticklabel cs:0.5)}, anchor=near ticklabel, font={{\fontsize{14.3 pt}{18.59 pt}\selectfont}}, color={rgb,1:red,0.0;green,0.0;blue,0.0}, draw opacity={1.0}, rotate={0.0}}, xmajorgrids={false}, xmin={-3.18}, xmax={3.18}, xtick={{-3.0,-2.0,-1.0,0.0,1.0,2.0,3.0}}, xticklabels={{$-3$,$-2$,$-1$,$0$,$1$,$2$,$3$}}, xtick align={inside}, xticklabel style={font={{\fontsize{10.4 pt}{13.520000000000001 pt}\selectfont}}, color={rgb,1:red,0.0;green,0.0;blue,0.0}, draw opacity={1.0}, rotate={0.0}}, x grid style={color={rgb,1:red,0.0;green,0.0;blue,0.0}, draw opacity={0.1}, line width={0.65}, solid}, x axis line style={color={rgb,1:red,0.0;green,0.0;blue,0.0}, draw opacity={1.0}, line width={1.3}, solid}, scaled y ticks={false}, ylabel={Coverage}, y tick style={color={rgb,1:red,0.0;green,0.0;blue,0.0}, opacity={1.0}}, y tick label style={color={rgb,1:red,0.0;green,0.0;blue,0.0}, opacity={1.0}, rotate={0}}, ylabel style={at={(ticklabel cs:0.5)}, anchor=near ticklabel, font={{\fontsize{14.3 pt}{18.59 pt}\selectfont}}, color={rgb,1:red,0.0;green,0.0;blue,0.0}, draw opacity={1.0}, rotate={0.0}}, ymajorgrids={false}, ymin={0.0}, ymax={1.05}, ytick={{0.0,0.25,0.5,0.75,1.0}}, yticklabels={{$0.00$,$0.25$,$0.50$,$0.75$,$1.00$}}, ytick align={inside}, yticklabel style={font={{\fontsize{10.4 pt}{13.520000000000001 pt}\selectfont}}, color={rgb,1:red,0.0;green,0.0;blue,0.0}, draw opacity={1.0}, rotate={0.0}}, y grid style={color={rgb,1:red,0.0;green,0.0;blue,0.0}, draw opacity={0.1}, line width={0.65}, solid}, y axis line style={color={rgb,1:red,0.0;green,0.0;blue,0.0}, draw opacity={1.0}, line width={1.3}, solid}, colorbar={false}]
    \addplot[color={rgb,1:red,1.0;green,0.549;blue,0.0}, name path={04386c4f-27d2-4628-b4ad-fb172cf3a949}, draw opacity={0.7}, line width={1.3}, solid]
        table[row sep={\\}]
        {
            \\
            -3.0  1.0  \\
            -2.8  1.0  \\
            -2.6  1.0  \\
            -2.4  1.0  \\
            -2.2  1.0  \\
            -2.0  1.0  \\
            -1.8  1.0  \\
            -1.6  1.0  \\
            -1.4  0.995  \\
            -1.2  0.995  \\
            -1.0  0.995  \\
            -0.8  0.995  \\
            -0.6  0.9975  \\
            -0.4  1.0  \\
            -0.2  1.0  \\
            0.0  1.0  \\
            0.2  1.0  \\
            0.4  1.0  \\
            0.6  1.0  \\
            0.8  1.0  \\
            1.0  1.0  \\
            1.2  1.0  \\
            1.4  1.0  \\
            1.6  1.0  \\
            1.8  1.0  \\
            2.0  1.0  \\
            2.2  1.0  \\
            2.4  1.0  \\
            2.6  1.0  \\
            2.8  1.0  \\
            3.0  1.0  \\
        }
        ;
    \addlegendentry {$\textrm{Gauss-F-Loc } (\mathcal{L}\mathcal{N})$}
    \addplot[color={rgb,1:red,0.0;green,0.0;blue,0.0}, name path={95d8649d-817b-4174-b144-5bed6a79788b}, draw opacity={0.9}, line width={1.3}, dashed]
        table[row sep={\\}]
        {
            \\
            -3.0  1.0  \\
            -2.8  1.0  \\
            -2.6  1.0  \\
            -2.4  1.0  \\
            -2.2  1.0  \\
            -2.0  1.0  \\
            -1.8  1.0  \\
            -1.6  1.0  \\
            -1.4  1.0  \\
            -1.2  1.0  \\
            -1.0  1.0  \\
            -0.8  1.0  \\
            -0.6  1.0  \\
            -0.4  1.0  \\
            -0.2  1.0  \\
            0.0  1.0  \\
            0.2  1.0  \\
            0.4  1.0  \\
            0.6  1.0  \\
            0.8  1.0  \\
            1.0  1.0  \\
            1.2  1.0  \\
            1.4  1.0  \\
            1.6  1.0  \\
            1.8  1.0  \\
            2.0  1.0  \\
            2.2  1.0  \\
            2.4  1.0  \\
            2.6  1.0  \\
            2.8  1.0  \\
            3.0  1.0  \\
        }
        ;
    \addlegendentry {$\textrm{DKW-F-Loc } (\mathcal{L}\mathcal{N})$}
    \addplot[color={rgb,1:red,0.0;green,0.0;blue,1.0}, name path={d72f3a76-c62e-4cb3-be56-d67fb5e6e88a}, draw opacity={0.4}, line width={1.3}, solid]
        table[row sep={\\}]
        {
            \\
            -3.0  0.94  \\
            -2.8  0.9675  \\
            -2.6  0.9775  \\
            -2.4  0.98  \\
            -2.2  0.9725  \\
            -2.0  0.9575  \\
            -1.8  0.96  \\
            -1.6  0.9525  \\
            -1.4  0.955  \\
            -1.2  0.9625  \\
            -1.0  0.9675  \\
            -0.8  0.9725  \\
            -0.6  0.9725  \\
            -0.4  0.97  \\
            -0.2  0.97  \\
            0.0  0.9625  \\
            0.2  0.9575  \\
            0.4  0.955  \\
            0.6  0.96  \\
            0.8  0.9525  \\
            1.0  0.96  \\
            1.2  0.9625  \\
            1.4  0.9675  \\
            1.6  0.9775  \\
            1.8  0.965  \\
            2.0  0.965  \\
            2.2  0.9725  \\
            2.4  0.975  \\
            2.6  0.9775  \\
            2.8  0.9725  \\
            3.0  0.9625  \\
        }
        ;
    \addlegendentry {$\textrm{AMARI } (\mathcal{L}\mathcal{N})$}
    \addplot[color={rgb,1:red,0.5451;green,0.0;blue,0.0}, name path={1103e0c1-9bdf-438b-ba03-157b54347ddc}, draw opacity={0.8}, line width={1.3}, dotted]
        table[row sep={\\}]
        {
            \\
            -3.0  0.3125  \\
            -2.8  0.3875  \\
            -2.6  0.5225  \\
            -2.4  0.7125  \\
            -2.2  0.845  \\
            -2.0  0.9075  \\
            -1.8  0.9125  \\
            -1.6  0.895  \\
            -1.4  0.8425  \\
            -1.2  0.7825  \\
            -1.0  0.7225  \\
            -0.8  0.72  \\
            -0.6  0.775  \\
            -0.4  0.86  \\
            -0.2  0.92  \\
            0.0  0.95  \\
            0.2  0.9325  \\
            0.4  0.9  \\
            0.6  0.88  \\
            0.8  0.85  \\
            1.0  0.845  \\
            1.2  0.865  \\
            1.4  0.91  \\
            1.6  0.9275  \\
            1.8  0.9475  \\
            2.0  0.9475  \\
            2.2  0.93  \\
            2.4  0.915  \\
            2.6  0.9025  \\
            2.8  0.8975  \\
            3.0  0.905  \\
        }
        ;
    \addlegendentry {$\textrm{Logspline}$}
    \addplot[color={rgb,1:red,0.8275;green,0.8275;blue,0.8275}, name path={06619500-0385-47a2-9f64-2da7f2b010a0}, draw opacity={1.0}, line width={1.3}, dotted, forget plot]
        table[row sep={\\}]
        {
            \\
            -9.540000000000001  0.95  \\
            9.540000000000001  0.95  \\
        }
        ;
\end{axis}
\end{tikzpicture}\end{adjustbox} \\
\end{tabular}
\caption{\textbf{Simulation results: Inference for the posterior mean in the Gaussian empirical Bayes problem. a)} Expected confidence intervals in the simulation with the prior \smash{$\gprior^{\text{Spiky}}$}~\eqref{eq:sim_priors}. 4 different inference methods are shown, as well as the ground truth as a function of $z$. \textbf{b)} Coverage of the above confidence intervals as a function of $z$. \textbf{c, d)} Inference results in the simulation with the prior \smash{$\gprior^{\text{NegSpiky}}$~\eqref{eq:sim_priors}}.}
\label{fig:postmean_simulations}
\end{figure}

Figure~\ref{fig:postmean_simulations} shows the results of the simulations for the posterior mean, averaged over 400 Monte Carlo replicates. 
The length of the different confidence intervals is qualitatively similar to what we observed in Figure~\ref{fig:prostate}. 
The log-spline intervals are shortest; however they do not achieve nominal coverage, while all other methods do. 
The pointwise coverage of the $F$-localization intervals is close to $100\%$, while the coverage of \Amari~is closer to the nominal $95\%$. The simultaneous coverage of \Amari~for $\EE[\gprior]{\mu \mid \Zo_i=\zo}$ as $\zo$ varies in Figure~\ref{fig:postmean_simulations} is $75\%$ for $\gprior^{\text{Spiky}}$ and $70\%$ for $\gprior^{\text{NegSpiky}}$. The $F$-localization methods have simultaneous coverage above $95\%$.

\begin{figure}
\centering
\begin{tabular}{@{}ll@{}}
a)  Spiky $G$ & b) Spiky $G$ \\
\begin{adjustbox}{width=0.4\linewidth}\input{tikz_figures/simulations/spiky_lfsr_locmix_intervals.tikz}\end{adjustbox} & 
\begin{adjustbox}{width=0.4\linewidth}\begin{tikzpicture}[/tikz/background rectangle/.style={fill={rgb,1:red,1.0;green,1.0;blue,1.0}, draw opacity={1.0}}, show background rectangle]
\begin{axis}[point meta max={nan}, point meta min={nan}, title={}, title style={at={{(0.5,1)}}, anchor={south}, font={{\fontsize{18.2 pt}{23.66 pt}\selectfont}}, color={rgb,1:red,0.0;green,0.0;blue,0.0}, draw opacity={1.0}, rotate={0.0}}, legend style={color={rgb,1:red,0.0;green,0.0;blue,0.0}, draw opacity={0.0}, line width={1.3}, solid, fill={rgb,1:red,0.0;green,0.0;blue,0.0}, fill opacity={0.0}, text opacity={1.0}, font={{\fontsize{10.4 pt}{13.520000000000001 pt}\selectfont}}, text={rgb,1:red,0.0;green,0.0;blue,0.0}, cells={anchor={west}}, at={(0.5, 0.02)}, anchor={south}}, axis background/.style={fill={rgb,1:red,1.0;green,1.0;blue,1.0}, opacity={1.0}}, anchor={north west}, xshift={1.0mm}, yshift={-1.0mm}, width={104.68mm}, height={81.82000000000001mm}, scaled x ticks={false}, xlabel={$z$}, x tick style={color={rgb,1:red,0.0;green,0.0;blue,0.0}, opacity={1.0}}, x tick label style={color={rgb,1:red,0.0;green,0.0;blue,0.0}, opacity={1.0}, rotate={0}}, xlabel style={at={(ticklabel cs:0.5)}, anchor=near ticklabel, font={{\fontsize{14.3 pt}{18.59 pt}\selectfont}}, color={rgb,1:red,0.0;green,0.0;blue,0.0}, draw opacity={1.0}, rotate={0.0}}, xmajorgrids={false}, xmin={-3.18}, xmax={3.18}, xtick={{-3.0,-2.0,-1.0,0.0,1.0,2.0,3.0}}, xticklabels={{$-3$,$-2$,$-1$,$0$,$1$,$2$,$3$}}, xtick align={inside}, xticklabel style={font={{\fontsize{10.4 pt}{13.520000000000001 pt}\selectfont}}, color={rgb,1:red,0.0;green,0.0;blue,0.0}, draw opacity={1.0}, rotate={0.0}}, x grid style={color={rgb,1:red,0.0;green,0.0;blue,0.0}, draw opacity={0.1}, line width={0.65}, solid}, x axis line style={color={rgb,1:red,0.0;green,0.0;blue,0.0}, draw opacity={1.0}, line width={1.3}, solid}, scaled y ticks={false}, ylabel={Coverage}, y tick style={color={rgb,1:red,0.0;green,0.0;blue,0.0}, opacity={1.0}}, y tick label style={color={rgb,1:red,0.0;green,0.0;blue,0.0}, opacity={1.0}, rotate={0}}, ylabel style={at={(ticklabel cs:0.5)}, anchor=near ticklabel, font={{\fontsize{14.3 pt}{18.59 pt}\selectfont}}, color={rgb,1:red,0.0;green,0.0;blue,0.0}, draw opacity={1.0}, rotate={0.0}}, ymajorgrids={false}, ymin={0.0}, ymax={1.05}, ytick={{0.0,0.25,0.5,0.75,1.0}}, yticklabels={{$0.00$,$0.25$,$0.50$,$0.75$,$1.00$}}, ytick align={inside}, yticklabel style={font={{\fontsize{10.4 pt}{13.520000000000001 pt}\selectfont}}, color={rgb,1:red,0.0;green,0.0;blue,0.0}, draw opacity={1.0}, rotate={0.0}}, y grid style={color={rgb,1:red,0.0;green,0.0;blue,0.0}, draw opacity={0.1}, line width={0.65}, solid}, y axis line style={color={rgb,1:red,0.0;green,0.0;blue,0.0}, draw opacity={1.0}, line width={1.3}, solid}, colorbar={false}]
    \addplot[color={rgb,1:red,1.0;green,0.549;blue,0.0}, name path={c6e6a6b8-a164-4dc6-84f4-86b72c12fa06}, draw opacity={0.7}, line width={1.3}, solid]
        table[row sep={\\}]
        {
            \\
            -3.0  1.0  \\
            -2.8  1.0  \\
            -2.6  1.0  \\
            -2.4  1.0  \\
            -2.2  1.0  \\
            -2.0  1.0  \\
            -1.8  1.0  \\
            -1.6  1.0  \\
            -1.4  1.0  \\
            -1.2  1.0  \\
            -1.0  1.0  \\
            -0.8  1.0  \\
            -0.6  1.0  \\
            -0.4  1.0  \\
            -0.2  1.0  \\
            0.0  1.0  \\
            0.2  1.0  \\
            0.4  1.0  \\
            0.6  1.0  \\
            0.8  1.0  \\
            1.0  1.0  \\
            1.2  1.0  \\
            1.4  1.0  \\
            1.6  1.0  \\
            1.8  1.0  \\
            2.0  1.0  \\
            2.2  1.0  \\
            2.4  1.0  \\
            2.6  1.0  \\
            2.8  1.0  \\
            3.0  1.0  \\
        }
        ;
    \addlegendentry {$\textrm{Gauss-F-Loc } (\mathcal{L}\mathcal{N})$}
    \addplot[color={rgb,1:red,0.0;green,0.0;blue,0.0}, name path={6e0b0e6d-0384-43be-973c-8ce56eea9508}, draw opacity={0.9}, line width={1.3}, dashed]
        table[row sep={\\}]
        {
            \\
            -3.0  1.0  \\
            -2.8  1.0  \\
            -2.6  1.0  \\
            -2.4  1.0  \\
            -2.2  1.0  \\
            -2.0  1.0  \\
            -1.8  1.0  \\
            -1.6  1.0  \\
            -1.4  1.0  \\
            -1.2  1.0  \\
            -1.0  1.0  \\
            -0.8  1.0  \\
            -0.6  1.0  \\
            -0.4  1.0  \\
            -0.2  1.0  \\
            0.0  1.0  \\
            0.2  1.0  \\
            0.4  1.0  \\
            0.6  1.0  \\
            0.8  1.0  \\
            1.0  1.0  \\
            1.2  1.0  \\
            1.4  1.0  \\
            1.6  1.0  \\
            1.8  1.0  \\
            2.0  1.0  \\
            2.2  1.0  \\
            2.4  1.0  \\
            2.6  1.0  \\
            2.8  1.0  \\
            3.0  1.0  \\
        }
        ;
    \addlegendentry {$\textrm{DKW-F-Loc } (\mathcal{L}\mathcal{N})$}
    \addplot[color={rgb,1:red,0.0;green,0.0;blue,1.0}, name path={40948b1b-79c8-495f-9437-f5fd95dd0c44}, draw opacity={0.4}, line width={1.3}, solid]
        table[row sep={\\}]
        {
            \\
            -3.0  0.9975  \\
            -2.8  0.9975  \\
            -2.6  1.0  \\
            -2.4  1.0  \\
            -2.2  1.0  \\
            -2.0  1.0  \\
            -1.8  1.0  \\
            -1.6  1.0  \\
            -1.4  1.0  \\
            -1.2  1.0  \\
            -1.0  1.0  \\
            -0.8  1.0  \\
            -0.6  1.0  \\
            -0.4  1.0  \\
            -0.2  1.0  \\
            0.0  1.0  \\
            0.2  1.0  \\
            0.4  1.0  \\
            0.6  1.0  \\
            0.8  1.0  \\
            1.0  1.0  \\
            1.2  1.0  \\
            1.4  1.0  \\
            1.6  1.0  \\
            1.8  1.0  \\
            2.0  1.0  \\
            2.2  1.0  \\
            2.4  1.0  \\
            2.6  1.0  \\
            2.8  1.0  \\
            3.0  1.0  \\
        }
        ;
    \addlegendentry {$\textrm{AMARI } (\mathcal{L}\mathcal{N})$}
    \addplot[color={rgb,1:red,0.5451;green,0.0;blue,0.0}, name path={343f8695-81d8-42fc-99da-8332b53e1769}, draw opacity={0.8}, line width={1.3}, dotted]
        table[row sep={\\}]
        {
            \\
            -3.0  0.4175  \\
            -2.8  0.3925  \\
            -2.6  0.3275  \\
            -2.4  0.2575  \\
            -2.2  0.2225  \\
            -2.0  0.1875  \\
            -1.8  0.2  \\
            -1.6  0.2225  \\
            -1.4  0.295  \\
            -1.2  0.3925  \\
            -1.0  0.525  \\
            -0.8  0.6375  \\
            -0.6  0.78  \\
            -0.4  0.8725  \\
            -0.2  0.9325  \\
            0.0  0.93  \\
            0.2  0.93  \\
            0.4  0.87  \\
            0.6  0.7775  \\
            0.8  0.6225  \\
            1.0  0.49  \\
            1.2  0.36  \\
            1.4  0.285  \\
            1.6  0.2075  \\
            1.8  0.175  \\
            2.0  0.165  \\
            2.2  0.1925  \\
            2.4  0.27  \\
            2.6  0.3525  \\
            2.8  0.425  \\
            3.0  0.4625  \\
        }
        ;
    \addlegendentry {$\textrm{Logspline}$}
    \addplot[color={rgb,1:red,0.8275;green,0.8275;blue,0.8275}, name path={625b423f-3a24-4ce2-a495-70e29de4894a}, draw opacity={1.0}, line width={1.3}, dotted, forget plot]
        table[row sep={\\}]
        {
            \\
            -9.540000000000001  0.95  \\
            9.540000000000001  0.95  \\
        }
        ;
\end{axis}
\end{tikzpicture}\end{adjustbox}
 \\
c) NegSpiky $G$ & d) NegSpiky $G$ \\
\begin{adjustbox}{width=0.4\linewidth}\input{tikz_figures/simulations/negspiky_lfsr_locmix_intervals.tikz}\end{adjustbox} & 
\begin{adjustbox}{width=0.4\linewidth}\begin{tikzpicture}[/tikz/background rectangle/.style={fill={rgb,1:red,1.0;green,1.0;blue,1.0}, draw opacity={1.0}}, show background rectangle]
\begin{axis}[point meta max={nan}, point meta min={nan}, title={}, title style={at={{(0.5,1)}}, anchor={south}, font={{\fontsize{18.2 pt}{23.66 pt}\selectfont}}, color={rgb,1:red,0.0;green,0.0;blue,0.0}, draw opacity={1.0}, rotate={0.0}}, legend style={color={rgb,1:red,0.0;green,0.0;blue,0.0}, draw opacity={0.0}, line width={1.3}, solid, fill={rgb,1:red,0.0;green,0.0;blue,0.0}, fill opacity={0.0}, text opacity={1.0}, font={{\fontsize{10.4 pt}{13.520000000000001 pt}\selectfont}}, text={rgb,1:red,0.0;green,0.0;blue,0.0}, cells={anchor={west}}, at={(0.5, 0.02)}, anchor={south}}, axis background/.style={fill={rgb,1:red,1.0;green,1.0;blue,1.0}, opacity={1.0}}, anchor={north west}, xshift={1.0mm}, yshift={-1.0mm}, width={104.68mm}, height={81.82000000000001mm}, scaled x ticks={false}, xlabel={$z$}, x tick style={color={rgb,1:red,0.0;green,0.0;blue,0.0}, opacity={1.0}}, x tick label style={color={rgb,1:red,0.0;green,0.0;blue,0.0}, opacity={1.0}, rotate={0}}, xlabel style={at={(ticklabel cs:0.5)}, anchor=near ticklabel, font={{\fontsize{14.3 pt}{18.59 pt}\selectfont}}, color={rgb,1:red,0.0;green,0.0;blue,0.0}, draw opacity={1.0}, rotate={0.0}}, xmajorgrids={false}, xmin={-3.18}, xmax={3.18}, xtick={{-3.0,-2.0,-1.0,0.0,1.0,2.0,3.0}}, xticklabels={{$-3$,$-2$,$-1$,$0$,$1$,$2$,$3$}}, xtick align={inside}, xticklabel style={font={{\fontsize{10.4 pt}{13.520000000000001 pt}\selectfont}}, color={rgb,1:red,0.0;green,0.0;blue,0.0}, draw opacity={1.0}, rotate={0.0}}, x grid style={color={rgb,1:red,0.0;green,0.0;blue,0.0}, draw opacity={0.1}, line width={0.65}, solid}, x axis line style={color={rgb,1:red,0.0;green,0.0;blue,0.0}, draw opacity={1.0}, line width={1.3}, solid}, scaled y ticks={false}, ylabel={Coverage}, y tick style={color={rgb,1:red,0.0;green,0.0;blue,0.0}, opacity={1.0}}, y tick label style={color={rgb,1:red,0.0;green,0.0;blue,0.0}, opacity={1.0}, rotate={0}}, ylabel style={at={(ticklabel cs:0.5)}, anchor=near ticklabel, font={{\fontsize{14.3 pt}{18.59 pt}\selectfont}}, color={rgb,1:red,0.0;green,0.0;blue,0.0}, draw opacity={1.0}, rotate={0.0}}, ymajorgrids={false}, ymin={0.0}, ymax={1.05}, ytick={{0.0,0.25,0.5,0.75,1.0}}, yticklabels={{$0.00$,$0.25$,$0.50$,$0.75$,$1.00$}}, ytick align={inside}, yticklabel style={font={{\fontsize{10.4 pt}{13.520000000000001 pt}\selectfont}}, color={rgb,1:red,0.0;green,0.0;blue,0.0}, draw opacity={1.0}, rotate={0.0}}, y grid style={color={rgb,1:red,0.0;green,0.0;blue,0.0}, draw opacity={0.1}, line width={0.65}, solid}, y axis line style={color={rgb,1:red,0.0;green,0.0;blue,0.0}, draw opacity={1.0}, line width={1.3}, solid}, colorbar={false}]
    \addplot[color={rgb,1:red,1.0;green,0.549;blue,0.0}, name path={d36560d9-1dc0-4da4-9281-1fc5d2c695c1}, draw opacity={0.7}, line width={1.3}, solid]
        table[row sep={\\}]
        {
            \\
            -3.0  1.0  \\
            -2.8  1.0  \\
            -2.6  1.0  \\
            -2.4  1.0  \\
            -2.2  1.0  \\
            -2.0  1.0  \\
            -1.8  1.0  \\
            -1.6  1.0  \\
            -1.4  1.0  \\
            -1.2  1.0  \\
            -1.0  1.0  \\
            -0.8  1.0  \\
            -0.6  1.0  \\
            -0.4  1.0  \\
            -0.2  1.0  \\
            0.0  1.0  \\
            0.2  1.0  \\
            0.4  1.0  \\
            0.6  1.0  \\
            0.8  1.0  \\
            1.0  1.0  \\
            1.2  1.0  \\
            1.4  1.0  \\
            1.6  1.0  \\
            1.8  1.0  \\
            2.0  1.0  \\
            2.2  1.0  \\
            2.4  1.0  \\
            2.6  1.0  \\
            2.8  1.0  \\
            3.0  1.0  \\
        }
        ;
    \addlegendentry {$\textrm{Gauss-F-Loc } (\mathcal{L}\mathcal{N})$}
    \addplot[color={rgb,1:red,0.0;green,0.0;blue,0.0}, name path={805b3f42-4061-477c-b156-9c695b922693}, draw opacity={0.9}, line width={1.3}, dashed]
        table[row sep={\\}]
        {
            \\
            -3.0  1.0  \\
            -2.8  1.0  \\
            -2.6  1.0  \\
            -2.4  1.0  \\
            -2.2  1.0  \\
            -2.0  1.0  \\
            -1.8  1.0  \\
            -1.6  1.0  \\
            -1.4  1.0  \\
            -1.2  1.0  \\
            -1.0  1.0  \\
            -0.8  1.0  \\
            -0.6  1.0  \\
            -0.4  1.0  \\
            -0.2  1.0  \\
            0.0  1.0  \\
            0.2  1.0  \\
            0.4  1.0  \\
            0.6  1.0  \\
            0.8  1.0  \\
            1.0  1.0  \\
            1.2  1.0  \\
            1.4  1.0  \\
            1.6  1.0  \\
            1.8  1.0  \\
            2.0  1.0  \\
            2.2  1.0  \\
            2.4  1.0  \\
            2.6  1.0  \\
            2.8  1.0  \\
            3.0  1.0  \\
        }
        ;
    \addlegendentry {$\textrm{DKW-F-Loc } (\mathcal{L}\mathcal{N})$}
    \addplot[color={rgb,1:red,0.0;green,0.0;blue,1.0}, name path={eadd7a07-1e12-43ac-aa06-ff0e280467e9}, draw opacity={0.4}, line width={1.3}, solid]
        table[row sep={\\}]
        {
            \\
            -3.0  0.9975  \\
            -2.8  1.0  \\
            -2.6  1.0  \\
            -2.4  1.0  \\
            -2.2  1.0  \\
            -2.0  1.0  \\
            -1.8  1.0  \\
            -1.6  1.0  \\
            -1.4  0.9975  \\
            -1.2  0.995  \\
            -1.0  0.995  \\
            -0.8  0.995  \\
            -0.6  0.995  \\
            -0.4  0.9925  \\
            -0.2  0.99  \\
            0.0  0.9875  \\
            0.2  0.9875  \\
            0.4  0.985  \\
            0.6  0.98  \\
            0.8  0.9775  \\
            1.0  0.9775  \\
            1.2  0.9725  \\
            1.4  0.9725  \\
            1.6  0.97  \\
            1.8  0.9775  \\
            2.0  0.98  \\
            2.2  0.9825  \\
            2.4  0.98  \\
            2.6  0.98  \\
            2.8  0.985  \\
            3.0  0.9875  \\
        }
        ;
    \addlegendentry {$\textrm{AMARI } (\mathcal{L}\mathcal{N})$}
    \addplot[color={rgb,1:red,0.5451;green,0.0;blue,0.0}, name path={2334b91a-a3b1-4913-ad81-5f2e90efec1c}, draw opacity={0.8}, line width={1.3}, dotted]
        table[row sep={\\}]
        {
            \\
            -3.0  0.63  \\
            -2.8  0.6875  \\
            -2.6  0.74  \\
            -2.4  0.78  \\
            -2.2  0.7925  \\
            -2.0  0.795  \\
            -1.8  0.7725  \\
            -1.6  0.76  \\
            -1.4  0.7475  \\
            -1.2  0.72  \\
            -1.0  0.685  \\
            -0.8  0.6475  \\
            -0.6  0.61  \\
            -0.4  0.5725  \\
            -0.2  0.55  \\
            0.0  0.5375  \\
            0.2  0.5325  \\
            0.4  0.5125  \\
            0.6  0.495  \\
            0.8  0.465  \\
            1.0  0.44  \\
            1.2  0.42  \\
            1.4  0.4175  \\
            1.6  0.395  \\
            1.8  0.39  \\
            2.0  0.3875  \\
            2.2  0.4025  \\
            2.4  0.4425  \\
            2.6  0.495  \\
            2.8  0.575  \\
            3.0  0.66  \\
        }
        ;
    \addlegendentry {$\textrm{Logspline}$}
    \addplot[color={rgb,1:red,0.8275;green,0.8275;blue,0.8275}, name path={b3c70bd5-0f34-4f5f-843e-5eef8a8c0f50}, draw opacity={1.0}, line width={1.3}, dotted, forget plot]
        table[row sep={\\}]
        {
            \\
            -9.540000000000001  0.95  \\
            9.540000000000001  0.95  \\
        }
        ;
\end{axis}
\end{tikzpicture}\end{adjustbox} 
\end{tabular}
\caption{\textbf{Simulation results: Inference for the local false sign rate in the Gaussian empirical Bayes problem.} The panels are analogous to the ones of Figure~\ref{fig:postmean_simulations}.}
\label{fig:lfsr_simulations}
\end{figure}

Figure~\ref{fig:lfsr_simulations} shows the simulation results for the local false sign rate. Most conclusions are similar to the ones we made for the posterior mean. However, here the Gauss-$F$-localization leads to shorter intervals compared to the DKW-$F$-localization.  Furthermore, in this case, both $F$-localization intervals and \Amari~have pointwise coverage close to $100\%$; the reason is that the worst case bias is substantial, and so bias-aware intervals lead to conservative inference for most $\gprior \in \gcal$.  In fact, \Amari~has simultaneous coverage above $95\%$ for $\PP{\mu \geq 0 \mid \Zo=\zo}$ as $\zo$ varies in Figure~\ref{fig:lfsr_simulations}.

\paragraph{Gaussian scale mixture $\gcal$:} We next repeat our simulations with the same settings, but using a different choice of \smash{$\gcal$}, namely the Gaussian scale mixture class~\eqref{eq:normal_scale_class}
\smash{$\set\nn(0.1, 15.6, 1.1)$}.
The scale mixture class \smash{$\set\nn(0.1, 15.6, 1.1)$} is strongly misspecified for \smash{$\gprior^{\text{NegSpiky}}$}. This was detected by our proposed methods, as the intersection of \smash{$\cb{ F_{\gprior}: \gprior \in \set\nn(0.1, 15.6, 1.1)}$} and $F$-localizations $\ff_n$ was empty. 
Thus, in Figure~\ref{fig:scalemix_simulations} we report the results of our simulations only for \smash{$\gprior^{\text{Spiky}}$}. 
We observe that the assumption that $\gprior$ is a scale mixture centered at $0$, instead of a location mixture, leads to substantially more precise inference, and especially so for the local false sign rate.

\begin{figure}
\centering
\begin{tabular}{@{}ll@{}}
a)  Posterior mean & b) Posterior mean \\
\begin{adjustbox}{width=0.4\linewidth}\input{tikz_figures/simulations/spiky_postmean_scalemix_intervals.tikz}\end{adjustbox} & 
\begin{adjustbox}{width=0.4\linewidth}\begin{tikzpicture}[/tikz/background rectangle/.style={fill={rgb,1:red,1.0;green,1.0;blue,1.0}, draw opacity={1.0}}, show background rectangle]
\begin{axis}[point meta max={nan}, point meta min={nan}, title={}, title style={at={{(0.5,1)}}, anchor={south}, font={{\fontsize{18.2 pt}{23.66 pt}\selectfont}}, color={rgb,1:red,0.0;green,0.0;blue,0.0}, draw opacity={1.0}, rotate={0.0}}, legend style={color={rgb,1:red,0.0;green,0.0;blue,0.0}, draw opacity={0.0}, line width={1.3}, solid, fill={rgb,1:red,0.0;green,0.0;blue,0.0}, fill opacity={0.0}, text opacity={1.0}, font={{\fontsize{10.4 pt}{13.520000000000001 pt}\selectfont}}, text={rgb,1:red,0.0;green,0.0;blue,0.0}, cells={anchor={west}}, at={(0.5, 0.02)}, anchor={south}}, axis background/.style={fill={rgb,1:red,1.0;green,1.0;blue,1.0}, opacity={1.0}}, anchor={north west}, xshift={1.0mm}, yshift={-1.0mm}, width={104.68mm}, height={81.82000000000001mm}, scaled x ticks={false}, xlabel={$z$}, x tick style={color={rgb,1:red,0.0;green,0.0;blue,0.0}, opacity={1.0}}, x tick label style={color={rgb,1:red,0.0;green,0.0;blue,0.0}, opacity={1.0}, rotate={0}}, xlabel style={at={(ticklabel cs:0.5)}, anchor=near ticklabel, font={{\fontsize{14.3 pt}{18.59 pt}\selectfont}}, color={rgb,1:red,0.0;green,0.0;blue,0.0}, draw opacity={1.0}, rotate={0.0}}, xmajorgrids={false}, xmin={-3.18}, xmax={3.18}, xtick={{-3.0,-2.0,-1.0,0.0,1.0,2.0,3.0}}, xticklabels={{$-3$,$-2$,$-1$,$0$,$1$,$2$,$3$}}, xtick align={inside}, xticklabel style={font={{\fontsize{10.4 pt}{13.520000000000001 pt}\selectfont}}, color={rgb,1:red,0.0;green,0.0;blue,0.0}, draw opacity={1.0}, rotate={0.0}}, x grid style={color={rgb,1:red,0.0;green,0.0;blue,0.0}, draw opacity={0.1}, line width={0.65}, solid}, x axis line style={color={rgb,1:red,0.0;green,0.0;blue,0.0}, draw opacity={1.0}, line width={1.3}, solid}, scaled y ticks={false}, ylabel={Coverage}, y tick style={color={rgb,1:red,0.0;green,0.0;blue,0.0}, opacity={1.0}}, y tick label style={color={rgb,1:red,0.0;green,0.0;blue,0.0}, opacity={1.0}, rotate={0}}, ylabel style={at={(ticklabel cs:0.5)}, anchor=near ticklabel, font={{\fontsize{14.3 pt}{18.59 pt}\selectfont}}, color={rgb,1:red,0.0;green,0.0;blue,0.0}, draw opacity={1.0}, rotate={0.0}}, ymajorgrids={false}, ymin={0.0}, ymax={1.05}, ytick={{0.0,0.25,0.5,0.75,1.0}}, yticklabels={{$0.00$,$0.25$,$0.50$,$0.75$,$1.00$}}, ytick align={inside}, yticklabel style={font={{\fontsize{10.4 pt}{13.520000000000001 pt}\selectfont}}, color={rgb,1:red,0.0;green,0.0;blue,0.0}, draw opacity={1.0}, rotate={0.0}}, y grid style={color={rgb,1:red,0.0;green,0.0;blue,0.0}, draw opacity={0.1}, line width={0.65}, solid}, y axis line style={color={rgb,1:red,0.0;green,0.0;blue,0.0}, draw opacity={1.0}, line width={1.3}, solid}, colorbar={false}]
    \addplot[color={rgb,1:red,1.0;green,0.549;blue,0.0}, name path={90904bb6-6e61-48d9-baeb-3e85e5af326d}, draw opacity={0.7}, line width={1.3}, solid]
        table[row sep={\\}]
        {
            \\
            -3.0  0.9974937343358395  \\
            -2.8  0.9974937343358395  \\
            -2.6  1.0  \\
            -2.4  1.0  \\
            -2.2  1.0  \\
            -2.0  1.0  \\
            -1.8  1.0  \\
            -1.6  0.9974937343358395  \\
            -1.4  0.9974937343358395  \\
            -1.2  0.9974937343358395  \\
            -1.0  0.9974937343358395  \\
            -0.8  0.9974937343358395  \\
            -0.6  0.9974937343358395  \\
            -0.4  0.9974937343358395  \\
            -0.2  0.9974937343358395  \\
            0.0  1.0  \\
            0.2  0.9974937343358395  \\
            0.4  0.9974937343358395  \\
            0.6  0.9974937343358395  \\
            0.8  0.9974937343358395  \\
            1.0  0.9974937343358395  \\
            1.2  0.9974937343358395  \\
            1.4  0.9974937343358395  \\
            1.6  0.9974937343358395  \\
            1.8  1.0  \\
            2.0  1.0  \\
            2.2  1.0  \\
            2.4  1.0  \\
            2.6  1.0  \\
            2.8  0.9974937343358395  \\
            3.0  0.9974937343358395  \\
        }
        ;
    \addlegendentry {$\textrm{Gauss-F-Loc } (\mathcal{S}\mathcal{N})$}
    \addplot[color={rgb,1:red,0.0;green,0.0;blue,0.0}, name path={eece17bf-fc0e-484e-9731-0f229c036530}, draw opacity={0.9}, line width={1.3}, dashed]
        table[row sep={\\}]
        {
            \\
            -3.0  1.0  \\
            -2.8  0.9974619289340102  \\
            -2.6  0.9974619289340102  \\
            -2.4  1.0  \\
            -2.2  1.0  \\
            -2.0  1.0  \\
            -1.8  0.9974619289340102  \\
            -1.6  1.0  \\
            -1.4  1.0  \\
            -1.2  1.0  \\
            -1.0  1.0  \\
            -0.8  0.9974619289340102  \\
            -0.6  0.9974619289340102  \\
            -0.4  0.9974619289340102  \\
            -0.2  0.9974619289340102  \\
            0.0  1.0  \\
            0.2  0.9974619289340102  \\
            0.4  0.9974619289340102  \\
            0.6  0.9974619289340102  \\
            0.8  0.9974619289340102  \\
            1.0  1.0  \\
            1.2  1.0  \\
            1.4  1.0  \\
            1.6  1.0  \\
            1.8  0.9974619289340102  \\
            2.0  1.0  \\
            2.2  1.0  \\
            2.4  1.0  \\
            2.6  0.9974619289340102  \\
            2.8  0.9974619289340102  \\
            3.0  1.0  \\
        }
        ;
    \addlegendentry {$\textrm{DKW-F-Loc } (\mathcal{S}\mathcal{N})$}
    \addplot[color={rgb,1:red,0.0;green,0.0;blue,1.0}, name path={61a4598c-7270-4c7f-ab95-37cdbe5ad443}, draw opacity={0.4}, line width={1.3}, solid]
        table[row sep={\\}]
        {
            \\
            -3.0  0.9473684210526315  \\
            -2.8  0.9548872180451128  \\
            -2.6  0.9724310776942355  \\
            -2.4  0.9849624060150376  \\
            -2.2  0.9774436090225563  \\
            -2.0  0.9774436090225563  \\
            -1.8  0.9724310776942355  \\
            -1.6  0.9598997493734336  \\
            -1.4  0.9598997493734336  \\
            -1.2  0.974937343358396  \\
            -1.0  0.9799498746867168  \\
            -0.8  0.9849624060150376  \\
            -0.6  0.9849624060150376  \\
            -0.4  0.9874686716791979  \\
            -0.2  0.9874686716791979  \\
            0.0  1.0  \\
            0.2  0.9874686716791979  \\
            0.4  0.9874686716791979  \\
            0.6  0.9849624060150376  \\
            0.8  0.9849624060150376  \\
            1.0  0.9799498746867168  \\
            1.2  0.974937343358396  \\
            1.4  0.9598997493734336  \\
            1.6  0.9598997493734336  \\
            1.8  0.9724310776942355  \\
            2.0  0.9774436090225563  \\
            2.2  0.9774436090225563  \\
            2.4  0.9849624060150376  \\
            2.6  0.9724310776942355  \\
            2.8  0.9548872180451128  \\
            3.0  0.9473684210526315  \\
        }
        ;
    \addlegendentry {$\textrm{AMARI } (\mathcal{S}\mathcal{N})$}
    \addplot[color={rgb,1:red,0.8275;green,0.8275;blue,0.8275}, name path={317f3e52-5487-4632-a951-e61f017d7d96}, draw opacity={1.0}, line width={1.3}, dotted, forget plot]
        table[row sep={\\}]
        {
            \\
            -9.540000000000001  0.95  \\
            9.540000000000001  0.95  \\
        }
        ;
\end{axis}
\end{tikzpicture}\end{adjustbox}
 \\
c) Local false sign rate & d) Local false sign rate \\
\begin{adjustbox}{width=0.4\linewidth}\input{tikz_figures/simulations/spiky_lfsr_scalemix_intervals.tikz}\end{adjustbox} & 
\begin{adjustbox}{width=0.4\linewidth}\begin{tikzpicture}[/tikz/background rectangle/.style={fill={rgb,1:red,1.0;green,1.0;blue,1.0}, draw opacity={1.0}}, show background rectangle]
\begin{axis}[point meta max={nan}, point meta min={nan}, title={}, title style={at={{(0.5,1)}}, anchor={south}, font={{\fontsize{18.2 pt}{23.66 pt}\selectfont}}, color={rgb,1:red,0.0;green,0.0;blue,0.0}, draw opacity={1.0}, rotate={0.0}}, legend style={color={rgb,1:red,0.0;green,0.0;blue,0.0}, draw opacity={0.0}, line width={1.3}, solid, fill={rgb,1:red,0.0;green,0.0;blue,0.0}, fill opacity={0.0}, text opacity={1.0}, font={{\fontsize{10.4 pt}{13.520000000000001 pt}\selectfont}}, text={rgb,1:red,0.0;green,0.0;blue,0.0}, cells={anchor={west}}, at={(0.5, 0.02)}, anchor={south}}, axis background/.style={fill={rgb,1:red,1.0;green,1.0;blue,1.0}, opacity={1.0}}, anchor={north west}, xshift={1.0mm}, yshift={-1.0mm}, width={104.68mm}, height={81.82000000000001mm}, scaled x ticks={false}, xlabel={$z$}, x tick style={color={rgb,1:red,0.0;green,0.0;blue,0.0}, opacity={1.0}}, x tick label style={color={rgb,1:red,0.0;green,0.0;blue,0.0}, opacity={1.0}, rotate={0}}, xlabel style={at={(ticklabel cs:0.5)}, anchor=near ticklabel, font={{\fontsize{14.3 pt}{18.59 pt}\selectfont}}, color={rgb,1:red,0.0;green,0.0;blue,0.0}, draw opacity={1.0}, rotate={0.0}}, xmajorgrids={false}, xmin={-3.18}, xmax={3.18}, xtick={{-3.0,-2.0,-1.0,0.0,1.0,2.0,3.0}}, xticklabels={{$-3$,$-2$,$-1$,$0$,$1$,$2$,$3$}}, xtick align={inside}, xticklabel style={font={{\fontsize{10.4 pt}{13.520000000000001 pt}\selectfont}}, color={rgb,1:red,0.0;green,0.0;blue,0.0}, draw opacity={1.0}, rotate={0.0}}, x grid style={color={rgb,1:red,0.0;green,0.0;blue,0.0}, draw opacity={0.1}, line width={0.65}, solid}, x axis line style={color={rgb,1:red,0.0;green,0.0;blue,0.0}, draw opacity={1.0}, line width={1.3}, solid}, scaled y ticks={false}, ylabel={Coverage}, y tick style={color={rgb,1:red,0.0;green,0.0;blue,0.0}, opacity={1.0}}, y tick label style={color={rgb,1:red,0.0;green,0.0;blue,0.0}, opacity={1.0}, rotate={0}}, ylabel style={at={(ticklabel cs:0.5)}, anchor=near ticklabel, font={{\fontsize{14.3 pt}{18.59 pt}\selectfont}}, color={rgb,1:red,0.0;green,0.0;blue,0.0}, draw opacity={1.0}, rotate={0.0}}, ymajorgrids={false}, ymin={0.0}, ymax={1.05}, ytick={{0.0,0.25,0.5,0.75,1.0}}, yticklabels={{$0.00$,$0.25$,$0.50$,$0.75$,$1.00$}}, ytick align={inside}, yticklabel style={font={{\fontsize{10.4 pt}{13.520000000000001 pt}\selectfont}}, color={rgb,1:red,0.0;green,0.0;blue,0.0}, draw opacity={1.0}, rotate={0.0}}, y grid style={color={rgb,1:red,0.0;green,0.0;blue,0.0}, draw opacity={0.1}, line width={0.65}, solid}, y axis line style={color={rgb,1:red,0.0;green,0.0;blue,0.0}, draw opacity={1.0}, line width={1.3}, solid}, colorbar={false}]
    \addplot[color={rgb,1:red,1.0;green,0.549;blue,0.0}, name path={87b84207-27c2-4e28-9a3e-7a9a8e3659fa}, draw opacity={0.7}, line width={1.3}, solid]
        table[row sep={\\}]
        {
            \\
            -3.0  1.0  \\
            -2.8  0.9974937343358395  \\
            -2.6  0.9974937343358395  \\
            -2.4  0.9974937343358395  \\
            -2.2  0.9974937343358395  \\
            -2.0  0.9974937343358395  \\
            -1.8  0.9974937343358395  \\
            -1.6  0.9974937343358395  \\
            -1.4  0.9974937343358395  \\
            -1.2  0.9974937343358395  \\
            -1.0  0.9974937343358395  \\
            -0.8  0.9974937343358395  \\
            -0.6  0.9974937343358395  \\
            -0.4  0.9974937343358395  \\
            -0.2  0.9974937343358395  \\
            0.0  1.0  \\
            0.2  0.9974937343358395  \\
            0.4  0.9974937343358395  \\
            0.6  0.9974937343358395  \\
            0.8  0.9974937343358395  \\
            1.0  0.9974937343358395  \\
            1.2  0.9974937343358395  \\
            1.4  0.9974937343358395  \\
            1.6  0.9974937343358395  \\
            1.8  0.9974937343358395  \\
            2.0  0.9974937343358395  \\
            2.2  0.9974937343358395  \\
            2.4  0.9974937343358395  \\
            2.6  0.9974937343358395  \\
            2.8  0.9974937343358395  \\
            3.0  1.0  \\
        }
        ;
    \addlegendentry {$\textrm{Gauss-F-Loc } (\mathcal{S}\mathcal{N})$}
    \addplot[color={rgb,1:red,0.0;green,0.0;blue,0.0}, name path={a1473e22-386c-4bdb-8688-53166d323df3}, draw opacity={0.9}, line width={1.3}, dashed]
        table[row sep={\\}]
        {
            \\
            -3.0  0.9974619289340102  \\
            -2.8  0.9974619289340102  \\
            -2.6  1.0  \\
            -2.4  1.0  \\
            -2.2  1.0  \\
            -2.0  1.0  \\
            -1.8  1.0  \\
            -1.6  1.0  \\
            -1.4  1.0  \\
            -1.2  1.0  \\
            -1.0  1.0  \\
            -0.8  1.0  \\
            -0.6  1.0  \\
            -0.4  1.0  \\
            -0.2  1.0  \\
            0.0  1.0  \\
            0.2  1.0  \\
            0.4  1.0  \\
            0.6  1.0  \\
            0.8  1.0  \\
            1.0  1.0  \\
            1.2  1.0  \\
            1.4  1.0  \\
            1.6  1.0  \\
            1.8  1.0  \\
            2.0  1.0  \\
            2.2  1.0  \\
            2.4  1.0  \\
            2.6  1.0  \\
            2.8  0.9974619289340102  \\
            3.0  0.9974619289340102  \\
        }
        ;
    \addlegendentry {$\textrm{DKW-F-Loc } (\mathcal{S}\mathcal{N})$}
    \addplot[color={rgb,1:red,0.0;green,0.0;blue,1.0}, name path={5e36a4f3-8f33-4a2a-84af-aab30e923210}, draw opacity={0.4}, line width={1.3}, solid]
        table[row sep={\\}]
        {
            \\
            -3.0  0.9949874686716792  \\
            -2.8  0.9924812030075187  \\
            -2.6  0.9874686716791979  \\
            -2.4  0.9874686716791979  \\
            -2.2  0.9874686716791979  \\
            -2.0  0.9874686716791979  \\
            -1.8  0.9899749373433584  \\
            -1.6  0.9899749373433584  \\
            -1.4  0.9899749373433584  \\
            -1.2  0.9899749373433584  \\
            -1.0  0.9899749373433584  \\
            -0.8  0.9949874686716792  \\
            -0.6  0.9949874686716792  \\
            -0.4  0.9949874686716792  \\
            -0.2  0.9949874686716792  \\
            0.0  1.0  \\
            0.2  0.9949874686716792  \\
            0.4  0.9949874686716792  \\
            0.6  0.9949874686716792  \\
            0.8  0.9949874686716792  \\
            1.0  0.9899749373433584  \\
            1.2  0.9899749373433584  \\
            1.4  0.9899749373433584  \\
            1.6  0.9899749373433584  \\
            1.8  0.9899749373433584  \\
            2.0  0.9874686716791979  \\
            2.2  0.9874686716791979  \\
            2.4  0.9874686716791979  \\
            2.6  0.9874686716791979  \\
            2.8  0.9924812030075187  \\
            3.0  0.9949874686716792  \\
        }
        ;
    \addlegendentry {$\textrm{AMARI } (\mathcal{S}\mathcal{N})$}
    \addplot[color={rgb,1:red,0.8275;green,0.8275;blue,0.8275}, name path={0109d143-27fd-4ab9-a1bd-d3fcdfb85fb5}, draw opacity={1.0}, line width={1.3}, dotted, forget plot]
        table[row sep={\\}]
        {
            \\
            -9.540000000000001  0.95  \\
            9.540000000000001  0.95  \\
        }
        ;
\end{axis}
\end{tikzpicture}\end{adjustbox} 
\end{tabular}
\caption{\textbf{Simulation results: Inference in the Gaussian empirical Bayes problem with $\gcal= \set\nn$ a Gaussian scale mixture.} \textbf{a,b)} Inference for the posterior mean (same simulation setting as Figure~\ref{fig:postmean_simulations}a,b). \textbf{c,d)} Inference for the local false sign rate (same simulation setting as Figure~\ref{fig:lfsr_simulations}a,b).}
\label{fig:scalemix_simulations}
\end{figure}

\paragraph{Degrees of freedom for the logspline approach:} One might at this point wonder whether one can reduce the bias of the plug-in logspline approach and achieve nominal coverage by increasing the degrees of freedom of the spline; we explore this in Figure~\ref{fig:gmodel_varying_dof} for the above simulation with the prior $G^{\text{NegSpiky}}$. In general, coverage indeed improves as the degrees of freedom increase; however, 
with many degrees of freedom, the variance can be so large that the resulting confidence intervals
are longer than the intervals proposed in this work. More importantly it is not clear a-priori, i.e., without knowing the ground truth, how to properly undersmooth the plug-in estimation and choose a number of degrees of freedom that provides good coverage. \citet{efron2016empirical} does not suggest undersmoothing, and instead, acknowledges that using a low-dimensional parametric family induces `definitional bias' in point estimates, `the pay-off being reduced variability'. On the other hand, as this example highlights, if we want confidence intervals that cover the true local false sign rate, it is important to explicitly account for bias.

\begin{figure}
\centering
\begin{adjustbox}{width=0.6\linewidth}\input{tikz_figures/simulations/varying_dof_plot.tikz}\end{adjustbox}
\caption{\textbf{Coverage versus expected length of confidence intervals:} Here the simulation setting is the same as that of Figure~\ref{fig:lfsr_simulations}, panels c,d) but we only consider inference at $\zo=2$, i.e., for \smash{$\theta_{\gprior}(2) = \PP[\gprior]{\mu \geq 0 \cond \Zo=2}$}. We apply the exponential family plug-in estimator for a range of degrees of freedom (from $2$ to $12$ shown by the number as well as progressively darker blue color), while in Figure~\ref{fig:lfsr_simulations} only the estimator with $5$ degrees is considered.}
\label{fig:gmodel_varying_dof}
\end{figure}

\section{On the asymptotic power of $F$-Localization and \Amari}
\label{sec:power}

Our goal in this work is to provide a unified approach for constructing intervals with the coverage property~\eqref{eq:CI} that lead to useful confidence statements in applied situations (cf. Sections~\ref{sec:applications} and~\ref{sec:simulations}). Given the generality of~\eqref{eq:EB}, we suspect it may be difficult to develop a unified theory of optimality. Nevertheless in this section we consider the issue of optimality and asymptotic relative efficiency in two concrete settings to provide the following conceptual insights. First, we describe a situation in which \Amari~is asymptotically efficient and outperforms the $F$-localization approach. Second, we illustrate the form of the $Q(\cdot)$ that solves the worst-case bias-variance problem~\eqref{eq:minimax_problem_tractable1} in a familiar context. Third, we elaborate on the issue of partial identification.

\subsection{Asymptotic relative efficiency in the Poisson model}
\label{subsec:power_poisson}

Consider the Poisson model in which we seek to conduct inference for the posterior mean \smash{$\theta_{\gprior}(\zo) = \EE[\gprior]{\mu \mid \Zo=z}$}. In view of Robbins' formula~\eqref{eq:robbins_poisson}, $\theta_{\gprior}(\zo)$ can be estimated at the \emph{parametric} \smash{$1/\sqrt{n}$} rate and so we can compare confidence intervals and estimators in terms of their asymptotic relative efficiency. Before studying \smash{$\theta_{\gprior}(\zo)$}, we first discuss inference for the linear functional \smash{$L(\gprior) = f_{\gprior}(\zo) = \PP[\gprior]{\Zo = \zo}$} for a fixed $\zo$. A consistent estimator in this case is given by the sample proportion \smash{$\hat{f}(\zo) = \#\cb{\Zo_i = \zo}/n$}, which has the limiting distribution,
\begin{equation}
\label{eq:sample_freq_clt}
\sqrt{n}\p{ \hat{f}(\zo) - f_{\gprior}(\zo)}  \xrightarrow[]{\mathcal{D}} \nn(0, f_{\gprior}(\zo)(1-f_{\gprior}(\zo)),
\end{equation}
and so we can build an asymptotic $1-\alpha$ confidence interval with asymptotic length equal to \smash{$2q_{1-\alpha/2}\sqrt{f_{\gprior}(\zo)(1-f_{\gprior}(\zo))}/\sqrt{n}$}, where \smash{$q_{1-\alpha/2}$} is the $1-\alpha/2$ quantile of the standard normal distribution. \citet{tierney1984asymptotic} prove that among a class of regular estimators, the estimator in~\eqref{eq:sample_freq_clt} is asymptotically efficient and so, asymptotically, the information that $F_{\gprior}$ is a Poisson mixture, is not helpful for inference of $L(\gprior)$. In this setting, it turns out that \Amari\footnote{We slightly abuse terminology in this section, and use the term \Amari~also for Algorithm~\ref{alg:linear_ci}, i.e., for our inference approach for linear functionals. The pilot estimates for \Amari~are chosen as in Proposition~\ref{prop:applications}.} matches the efficient confidence interval based on~\eqref{eq:sample_freq_clt} and is shorter than the DKW-$F$-Localization interval.
\begin{prop}
\label{prop:asymptotics_poisson_marginal_prob}
We consider inference for \smash{$L(\gprior) = f_{\gprior}(\zo)$}, in the Poisson model at level $\alpha \in (0,1)$ for a fixed \smash{$\zo \in \NN_{\geq 1}$}. We treat the observations as right-censored for $\Zo_i > M$ as in~\eqref{eq:tilde_Z}, for fixed $M \geq \zo+1$, and consider the prior class \smash{$\gcal = \pp([a, b])$}, $0 \leq a < b< \infty$. If \smash{$\gprior \in \gcal$} and $\gprior$ is supported on at least $M+2$ points, then as \smash{$n\to \infty$} it holds that: \Amari~has asymptotically the same length as the confidence interval constructed using~\eqref{eq:sample_freq_clt}, namely \smash{$2q_{1-\alpha/2}\sqrt{f_{\gprior}(\zo)(1-f_{\gprior}(\zo))/n}(1+o_{\mathbb P_G}(1))$}. The DKW-$F$-Localization interval has asymptotic length \smash{$2\sqrt{2 \log(2/\alpha)/n}(1+o_{\mathbb P_G}(1))$}.
\end{prop}
The key argument in the proof of the above proposition is that, in this setting, 
the optimal \smash{$Q(\cdot)$} solving~\eqref{eq:minimax_problem_tractable1} is with high probability equal to \smash{$\ind(\cdot = \zo)$} for $n$ large enough, in which case, \smash{$\hL = \#\cb{\Zo_i = \zo}/n$}. In finite samples, however, the optimal $Q(\cdot)$ takes the form of a kernel smoother that upweights \smash{$\Zo_i$} in a neighborhood of $\zo$. To illustrate, we simulate from the Poisson empirical Bayes model~\eqref{eq:EB} with \smash{$\gprior = U[0,\,2]$} and let $n$ vary. We specify \smash{$\gcal = \mathcal{P}([0,4])$}. The optimal \smash{$Q(\cdot)$} of \Amari~for different values of $n$ is shown in Figure~\ref{fig:poisson_rel_efficiency}a). Figure~\ref{fig:poisson_rel_efficiency}b) shows the expected length of the \Amari~and DKW-$F$-Localization confidence intervals, as well as the asymptotic lengths from Proposition~\ref{prop:asymptotics_poisson_marginal_prob}. As expected, \Amari~has shorter length than the DKW-$F$-localization intervals. The information that \smash{$F_{\gprior}$} is a Poisson mixture is not helpful to both approaches for $n$ large, however, both \Amari~and DKW-$F$-Localization can use this information for smaller $n$ to provide sharper inference.

\begin{figure}
\centering
\begin{tabular}{@{}lll@{}}
a) & b) & c)\\
\begin{adjustbox}{width=0.308\linewidth}\input{tikz_figures/relative_efficiency/Q_plot.tikz}\end{adjustbox} &
\begin{adjustbox}{width=0.33\linewidth}\begin{tikzpicture}[/tikz/background rectangle/.style={fill={rgb,1:red,1.0;green,1.0;blue,1.0}, draw opacity={1.0}}, show background rectangle]
\begin{axis}[point meta max={nan}, point meta min={nan}, title={}, title style={at={{(0.5,1)}}, anchor={south}, font={{\fontsize{18.2 pt}{23.66 pt}\selectfont}}, color={rgb,1:red,0.0;green,0.0;blue,0.0}, draw opacity={1.0}, rotate={0.0}}, legend style={color={rgb,1:red,0.0;green,0.0;blue,0.0}, draw opacity={0.0}, line width={1.3}, solid, fill={rgb,1:red,0.0;green,0.0;blue,0.0}, fill opacity={0.0}, text opacity={1.0}, font={{\fontsize{10.4 pt}{13.520000000000001 pt}\selectfont}}, text={rgb,1:red,0.0;green,0.0;blue,0.0}, cells={anchor={west}}, at={(0.98, 0.98)}, anchor={north east}}, axis background/.style={fill={rgb,1:red,1.0;green,1.0;blue,1.0}, opacity={1.0}}, anchor={north west}, xshift={1.0mm}, yshift={-1.0mm}, width={104.68mm}, height={74.2mm}, scaled x ticks={false}, xlabel={$n$}, x tick style={color={rgb,1:red,0.0;green,0.0;blue,0.0}, opacity={1.0}}, x tick label style={color={rgb,1:red,0.0;green,0.0;blue,0.0}, opacity={1.0}, rotate={0}}, xlabel style={at={(ticklabel cs:0.5)}, anchor=near ticklabel, font={{\fontsize{14.3 pt}{18.59 pt}\selectfont}}, color={rgb,1:red,0.0;green,0.0;blue,0.0}, draw opacity={1.0}, rotate={0.0}}, xmode={log}, log basis x={10}, xmajorgrids={false}, xmin={66.06934480075961}, xmax={1.5135612484362072e8}, xtick={{100.0,1000.0,10000.0,100000.0,1.0e6,1.0e7,1.0e8}}, xticklabels={{$10^{2}$,$10^{3}$,$10^{4}$,$10^{5}$,$10^{6}$,$10^{7}$,$10^{8}$}}, xtick align={inside}, xticklabel style={font={{\fontsize{10.4 pt}{13.520000000000001 pt}\selectfont}}, color={rgb,1:red,0.0;green,0.0;blue,0.0}, draw opacity={1.0}, rotate={0.0}}, x grid style={color={rgb,1:red,0.0;green,0.0;blue,0.0}, draw opacity={0.1}, line width={0.65}, solid}, x axis line style={color={rgb,1:red,0.0;green,0.0;blue,0.0}, draw opacity={1.0}, line width={1.3}, solid}, scaled y ticks={false}, ylabel={Expected CI length}, y tick style={color={rgb,1:red,0.0;green,0.0;blue,0.0}, opacity={1.0}}, y tick label style={color={rgb,1:red,0.0;green,0.0;blue,0.0}, opacity={1.0}, rotate={0}}, ylabel style={at={(ticklabel cs:0.5)}, anchor=near ticklabel, font={{\fontsize{14.3 pt}{18.59 pt}\selectfont}}, color={rgb,1:red,0.0;green,0.0;blue,0.0}, draw opacity={1.0}, rotate={0.0}}, ymode={log}, log basis y={10}, ymajorgrids={false}, ymin={7.802184304047143e-5}, ymax={0.7029391637557397}, ytick={{0.0001,0.001,0.01,0.1}}, yticklabels={{$10^{-4}$,$10^{-3}$,$10^{-2}$,$10^{-1}$}}, ytick align={inside}, yticklabel style={font={{\fontsize{10.4 pt}{13.520000000000001 pt}\selectfont}}, color={rgb,1:red,0.0;green,0.0;blue,0.0}, draw opacity={1.0}, rotate={0.0}}, y grid style={color={rgb,1:red,0.0;green,0.0;blue,0.0}, draw opacity={0.1}, line width={0.65}, solid}, y axis line style={color={rgb,1:red,0.0;green,0.0;blue,0.0}, draw opacity={1.0}, line width={1.3}, solid}, colorbar={false}]
    \addplot[color={rgb,1:red,0.0;green,0.0;blue,1.0}, name path={a9648438-6e45-43d2-9229-87234d54182a}, draw opacity={1.0}, line width={1.3}, solid]
        table[row sep={\\}]
        {
            \\
            100.0  0.058242255436050244  \\
            316.0  0.03455333731347423  \\
            1000.0  0.020858018287577062  \\
            3162.0  0.012514681428636232  \\
            10000.0  0.007675150662951855  \\
            31622.0  0.004631472776171761  \\
            100000.0  0.002739992332364611  \\
            316227.0  0.0016081449545843817  \\
            1.0e6  0.0009429678561132582  \\
            3.162277e6  0.0005591069810104321  \\
            1.0e7  0.000319293075911421  \\
            3.1622776e7  0.0001795304522215613  \\
            1.0e8  0.00010095822820660405  \\
        }
        ;
    \addlegendentry {$\textrm{AMARI}$}
    \addplot[color={rgb,1:red,0.0;green,0.0;blue,1.0}, name path={7c7aa1ed-00f6-447d-a304-e35b0aac8807}, draw opacity={1.0}, line width={1.3}, dotted]
        table[row sep={\\}]
        {
            \\
            100.0  0.10095990672870495  \\
            316.0  0.056794384769542905  \\
            1000.0  0.03192632576208669  \\
            3162.0  0.017954280582336  \\
            10000.0  0.010095990672870495  \\
            31622.0  0.005677462488937356  \\
            100000.0  0.0031926325762086693  \\
            316227.0  0.0017953514082442417  \\
            1.0e6  0.0010095990672870495  \\
            3.162277e6  0.0005677393366707888  \\
            1.0e7  0.0003192632576208669  \\
            3.1622776e7  0.00017953492508310373  \\
            1.0e8  0.00010095990672870496  \\
        }
        ;
    \addlegendentry {$\textrm{AMARI asymp.}$}
    \addplot[color={rgb,1:red,1.0;green,0.549;blue,0.0}, name path={e67c4039-5e1a-4d0f-839c-22d06d89b244}, draw opacity={1.0}, line width={1.3}, solid]
        table[row sep={\\}]
        {
            \\
            100.0  0.10072548400512861  \\
            316.0  0.06833405494673572  \\
            1000.0  0.04876458481930885  \\
            3162.0  0.034409778118529605  \\
            10000.0  0.023510215249875552  \\
            31622.0  0.015804584265707602  \\
            100000.0  0.010210546606749627  \\
            316227.0  0.006751374149450681  \\
            1.0e6  0.00458485210787465  \\
            3.162277e6  0.0029762951430524864  \\
            1.0e7  0.0017178776333870232  \\
            3.1622776e7  0.0009660335940639852  \\
            1.0e8  0.0005432406062961437  \\
        }
        ;
    \addlegendentry {$\textrm{DKW-F-Loc}$}
    \addplot[color={rgb,1:red,1.0;green,0.549;blue,0.0}, name path={4d5be5ff-d074-4a11-a3bd-71861b8c6526}, draw opacity={1.0}, line width={1.3}, dotted]
        table[row sep={\\}]
        {
            \\
            100.0  0.5432406062962478  \\
            316.0  0.305596716717813  \\
            1000.0  0.17178776333869503  \\
            3162.0  0.09660759984030413  \\
            10000.0  0.05432406062962478  \\
            31622.0  0.030549039362747596  \\
            100000.0  0.0171787763338695  \\
            316227.0  0.009660347549153634  \\
            1.0e6  0.005432406062962478  \\
            3.162277e6  0.0030548667432908697  \\
            1.0e7  0.0017178776333869501  \\
            3.1622776e7  0.0009660335940639965  \\
            1.0e8  0.0005432406062962478  \\
        }
        ;
    \addlegendentry {$\textrm{DKW-F-Loc asymp.}$}
\end{axis}
\end{tikzpicture}\end{adjustbox} & 
\begin{adjustbox}{width=0.33\linewidth}\begin{tikzpicture}[/tikz/background rectangle/.style={fill={rgb,1:red,1.0;green,1.0;blue,1.0}, draw opacity={1.0}}, show background rectangle]
\begin{axis}[point meta max={nan}, point meta min={nan}, title={}, title style={at={{(0.5,1)}}, anchor={south}, font={{\fontsize{18.2 pt}{23.66 pt}\selectfont}}, color={rgb,1:red,0.0;green,0.0;blue,0.0}, draw opacity={1.0}, rotate={0.0}}, legend style={color={rgb,1:red,0.0;green,0.0;blue,0.0}, draw opacity={0.0}, line width={1.3}, solid, fill={rgb,1:red,0.0;green,0.0;blue,0.0}, fill opacity={0.0}, text opacity={1.0}, font={{\fontsize{10.4 pt}{13.520000000000001 pt}\selectfont}}, text={rgb,1:red,0.0;green,0.0;blue,0.0}, cells={anchor={west}}, at={(0.98, 0.98)}, anchor={north east}}, axis background/.style={fill={rgb,1:red,1.0;green,1.0;blue,1.0}, opacity={1.0}}, anchor={north west}, xshift={1.0mm}, yshift={-1.0mm}, width={104.68mm}, height={74.2mm}, scaled x ticks={false}, xlabel={$n$}, x tick style={color={rgb,1:red,0.0;green,0.0;blue,0.0}, opacity={1.0}}, x tick label style={color={rgb,1:red,0.0;green,0.0;blue,0.0}, opacity={1.0}, rotate={0}}, xlabel style={at={(ticklabel cs:0.5)}, anchor=near ticklabel, font={{\fontsize{14.3 pt}{18.59 pt}\selectfont}}, color={rgb,1:red,0.0;green,0.0;blue,0.0}, draw opacity={1.0}, rotate={0.0}}, xmode={log}, log basis x={10}, xmajorgrids={false}, xmin={66.06934480075961}, xmax={1.5135612484362072e8}, xtick={{100.0,1000.0,10000.0,100000.0,1.0e6,1.0e7,1.0e8}}, xticklabels={{$10^{2}$,$10^{3}$,$10^{4}$,$10^{5}$,$10^{6}$,$10^{7}$,$10^{8}$}}, xtick align={inside}, xticklabel style={font={{\fontsize{10.4 pt}{13.520000000000001 pt}\selectfont}}, color={rgb,1:red,0.0;green,0.0;blue,0.0}, draw opacity={1.0}, rotate={0.0}}, x grid style={color={rgb,1:red,0.0;green,0.0;blue,0.0}, draw opacity={0.1}, line width={0.65}, solid}, x axis line style={color={rgb,1:red,0.0;green,0.0;blue,0.0}, draw opacity={1.0}, line width={1.3}, solid}, scaled y ticks={false}, ylabel={Expected CI length}, y tick style={color={rgb,1:red,0.0;green,0.0;blue,0.0}, opacity={1.0}}, y tick label style={color={rgb,1:red,0.0;green,0.0;blue,0.0}, opacity={1.0}, rotate={0}}, ylabel style={at={(ticklabel cs:0.5)}, anchor=near ticklabel, font={{\fontsize{14.3 pt}{18.59 pt}\selectfont}}, color={rgb,1:red,0.0;green,0.0;blue,0.0}, draw opacity={1.0}, rotate={0.0}}, ymode={log}, log basis y={10}, ymajorgrids={false}, ymin={0.0031603867609636227}, ymax={54.873426474522326}, ytick={{0.01,0.1,1.0,10.0}}, yticklabels={{$10^{-2}$,$10^{-1}$,$10^{0}$,$10^{1}$}}, ytick align={inside}, yticklabel style={font={{\fontsize{10.4 pt}{13.520000000000001 pt}\selectfont}}, color={rgb,1:red,0.0;green,0.0;blue,0.0}, draw opacity={1.0}, rotate={0.0}}, y grid style={color={rgb,1:red,0.0;green,0.0;blue,0.0}, draw opacity={0.1}, line width={0.65}, solid}, y axis line style={color={rgb,1:red,0.0;green,0.0;blue,0.0}, draw opacity={1.0}, line width={1.3}, solid}, colorbar={false}]
    \addplot[color={rgb,1:red,0.0;green,0.0;blue,1.0}, name path={b3863b10-2fbb-4923-84cf-e331cd9901ae}, draw opacity={1.0}, line width={1.3}, solid]
        table[row sep={\\}]
        {
            \\
            100.0  1.8459948244092508  \\
            316.0  1.005120690988898  \\
            1000.0  0.5653849197275671  \\
            3162.0  0.34989602894629657  \\
            10000.0  0.2064815494180251  \\
            31622.0  0.12433069732260586  \\
            100000.0  0.07674535269289942  \\
            316227.0  0.04924578772948537  \\
            1.0e6  0.03222068685690074  \\
            3.162277e6  0.02060934688814652  \\
            1.0e7  0.012538215511988731  \\
            3.1622776e7  0.007408849977324161  \\
            1.0e8  0.004166411676640141  \\
        }
        ;
    \addlegendentry {$\textrm{AMARI}$}
    \addplot[color={rgb,1:red,0.0;green,0.0;blue,1.0}, name path={e2d8d1eb-07d8-40c3-823a-1937f7373541}, draw opacity={1.0}, line width={1.3}, dotted]
        table[row sep={\\}]
        {
            \\
            100.0  4.166098476176573  \\
            316.0  2.3436134948063034  \\
            1000.0  1.3174360141274704  \\
            3162.0  0.7408812408663771  \\
            10000.0  0.41660984761765724  \\
            31622.0  0.2342798105714432  \\
            100000.0  0.13174360141274702  \\
            316227.0  0.07408496113399435  \\
            1.0e6  0.041660984761765725  \\
            3.162277e6  0.023427695825091136  \\
            1.0e7  0.013174360141274702  \\
            3.1622776e7  0.007408487210859648  \\
            1.0e8  0.004166098476176572  \\
        }
        ;
    \addlegendentry {$\textrm{AMARI asymp.}$}
    \addplot[color={rgb,1:red,1.0;green,0.549;blue,0.0}, name path={eb2e31ee-80c5-445b-be71-18606967e28c}, draw opacity={1.0}, line width={1.3}, solid]
        table[row sep={\\}]
        {
            \\
            100.0  3.2314837546207436  \\
            316.0  2.554972178101104  \\
            1000.0  1.6976303936205601  \\
            3162.0  1.1191965174618954  \\
            10000.0  0.7214689642633503  \\
            31622.0  0.47599249611354977  \\
            100000.0  0.3059030626490075  \\
            316227.0  0.20160311245545248  \\
            1.0e6  0.13697700845303373  \\
            3.162277e6  0.09501342471131992  \\
            1.0e7  0.06291605729992926  \\
            3.1622776e7  0.04331838845902233  \\
            1.0e8  0.029942914580171957  \\
        }
        ;
    \addlegendentry {$\textrm{DKW-F-Loc}$}
    \addplot[color={rgb,1:red,1.0;green,0.549;blue,0.0}, name path={a12ed922-db04-493d-8431-9db1fe94faa8}, draw opacity={1.0}, line width={1.3}, dotted]
        table[row sep={\\}]
        {
            \\
            100.0  41.62677659937315  \\
            316.0  23.416891353252634  \\
            1000.0  13.163542570501756  \\
            3162.0  7.402728974499625  \\
            10000.0  4.162677659937315  \\
            31622.0  2.340874415755791  \\
            100000.0  1.3163542570501754  \\
            316227.0  0.7402412938947817  \\
            1.0e6  0.41626776599373144  \\
            3.162277e6  0.23408459159711004  \\
            1.0e7  0.13163542570501754  \\
            3.1622776e7  0.07402404043717978  \\
            1.0e8  0.04162677659937314  \\
        }
        ;
    \addlegendentry {$\textrm{DKW-F-Loc asymp.}$}
\end{axis}
\end{tikzpicture}\end{adjustbox}
\end{tabular}
\caption{\textbf{Relative efficiency in the Poisson empirical Bayes problem.} We draw $n$ samples \smash{$\mu_i \sim U[0,2],\; \Zo_i \mid \mu_i \sim \Poisson{\mu_i}$}. \textbf{a)} Optimal \smash{$Q(\cdot)$}~\eqref{eq:minimax_problem_tractable1} of \Amari~for estimating \smash{$L(\gprior)= f_{\gprior}(3)$}, when $\gcal = \mathcal{P}([0,4])$. Each $Q(\cdot)$ corresponds to a single simulation and a different value of $n$. For example, for \smash{$n=10^7$}, \smash{$Q(\cdot)$} is the indicator \smash{$\ind(\cdot = 3)$}. All \smash{$Q(\cdot)$} are constant for \smash{$\zo \in \cb{6,7,\dotsc}$} by construction~\eqref{eq:minimax_problem_tractable1}. \textbf{b)} Confidence interval length of \Amari~and DKW-$F$-Localization in the same setting as panel a), averaged over 50 Monte Carlo replicates. The panel also shows the asymptotic interval length for the two methods as per Proposition~\ref{prop:asymptotics_poisson_marginal_prob}. \textbf{c)} Analogous to panel b) for inference of the posterior mean \smash{$\theta_{\gprior}(3) = \EE[\gprior]{\mu \cond \Zo=3}$}. Asymptotic lengths are derived in Proposition~\ref{prob:asymptotics_poisson_robbins}.} 
\label{fig:poisson_rel_efficiency}
\end{figure}

The next proposition and  Figure~\ref{fig:poisson_rel_efficiency}c) pertain to the posterior mean $\theta_{\gprior}(\zo) = \EE[\gprior]{\mu \mid \Zo=z}$ and are analogous to Proposition~\ref{prop:asymptotics_poisson_marginal_prob} and  Figure~\ref{fig:poisson_rel_efficiency}b). Our findings are similar; \Amari~outperforms the DKW-$F$-localization intervals and for small $n$, both methods perform better than predicted by the asymptotic limit.

\begin{prop}
\label{prob:asymptotics_poisson_robbins}
We consider inference for $\theta_{\gprior}(\zo) = \EE[\gprior]{\mu \mid \Zo=z}$ in the setting of Proposition~\ref{prop:asymptotics_poisson_marginal_prob}. Then:
$$
\begin{aligned}
\sqrt{n}\abs{\ii_{\alpha}^{\text{AMARI}}(\zo)} &\xrightarrow{\mathbb P_{\gprior}} 2(\zo+1) q_{1-\alpha/2}\sqrt{f_{\gprior}(\zo+1) \p{1 + f_{\gprior}(\zo+1)/ f_{\gprior}(\zo)}}\, \big/\, f_{\gprior}(\zo)\,,\\
\sqrt{n}\abs{\ii_{\alpha}^{\text{DKW-F-Loc}}(\zo)} &\xrightarrow{\mathbb P_{\gprior}} 2(\zo+1)\sqrt{2 \log(2/\alpha)} \cdot \p{1 + f_{\gprior}(\zo+1)  / f_{\gprior}(\zo)}\, \big / \, f_{\gprior}(\zo).
\end{aligned}
$$
\end{prop}
In this case, \Amari~is asymptotically equivalent to the intervals constructed by~\citet{robbins1980empirical} and~\citet{karlis2018confidence}. The confidence intervals of~\citet{robbins1980empirical} are based on the joint central limit theorem for $\#\cb{Z_i =z}/n$ and $\#\cb{Z_i =z+1}/n$, the delta method and his formula~\eqref{eq:robbins_poisson}.

\subsection{Sharp partial identification in the Bernoulli model}
\label{subsec:bernoulli_partial}

Above we compared methods in a problem in which parametric rates are attainable. Here we compare methods under partial identification, i.e., when confidence intervals will not shrink to a point mass, even as $n\to \infty$. We consider model~\eqref{eq:EB} with $\Zo_i \cond \mu_i \; \sim \; \text{Bernoulli}(\mu_i)$, i.e., the Binomial model with a single ($N=1$) trial.  Furthermore, we do not impose additional structure on $\gcal$, i.e., we assume that $\gprior \in \gcal = \pp([0,1])$~\eqref{eq:all_dbns}. $\Zo_i$ is supported on $\cb{0,1}$ and we can take $\lambda = \delta_0 + \delta_1$ to be the counting measure on $\cb{0,1}$ and $p(\zo \cond \mu) = \mu^{\zo}(1-\mu)^{1-\zo}$. The marginal distribution $F_{\gprior}$ is fully determined by $f_{\gprior}(1)$.

We first consider inference for the second moment $L(\gprior) = \int \mu^2\, d\gprior(\mu)$, which is a linear functional of $\gprior$. The distribution of $\Zo_i$, however, does not point identify $L(\gprior)$, unless \smash{$f_{\gprior}(1) \in \cb{0,1}$}, and the partial identification interval for $L(\gprior)$ is equal to \smash{$[f_{\gprior}(1)^2,\, f_{\gprior}(1)]$}.\footnote{ The right bound is attained by the prior on $[0,1]$ with $\PP[\gprior]{\cb{0}} = f_{\gprior}(0), \PP[\gprior]{\cb{1}} = f_{\gprior}(1)$, and the left bound by the point mass prior with $\PP[\gprior]{\cb{f_{\gprior}(1)}}=1$. Both of these priors induce the same marginal distribution $F_{\gprior}$.} The posterior mean, $\theta_{\gprior}(1) = \EE[\gprior]{\mu \mid \Zo=1} = \int \mu^2 \, dG(\mu) \big / f_{\gprior}(1)$ is also not identified and its partial identification interval is equal to $[f_{\gprior}(1), 1]$. The next proposition shows that, as $n\to \infty$, both \Amari~and $F$-localization intervals converge to the corresponding partial identification intervals, and so all proposed intervals have the same asymptotic length (which is the best possible).

\begin{prop}
\label{prop:bernoulli_partial}
Consider inference for $L(\gprior) = \int \mu^2\, dG(\mu)$ in the Bernoulli model with the DKW-$F$-localization~\eqref{eq:floc_chisq}, the $\chi^2$-$F$-localization~\eqref{eq:floc_chisq}, or \Amari~ (using any of the above as a pilot $F$-localization). We use $\gcal = \pp([0,1])$ and suppose that $0<f_{\gprior}(1) <1$. The length of all these confidence intervals asymptotically converges to the length of the partial identification interval, i.e., \smash{$\abs{\mathcal{I}_{\alpha}}\big / (f_{\gprior}(1)(1-f_{\gprior}(1))) \stackrel{\mathbb P_{\gprior}}{\to} 1$}. Similarly, all of the above confidence intervals for $\theta_{\gprior}(1) = \EE[\gprior]{\mu \mid \Zo=1}$ asymptotically match the partial identification intervals, i.e.,  \smash{$\abs{\mathcal{I}_{\alpha}(\zo)}\big /\p{1-f_{\gprior}(1)}\stackrel{\mathbb P_{\gprior}}{\to} 1$}.
\end{prop}

\section{On the choice of $\gcal$}
\label{sec:gcal}
Throughout this paper, we have taken the choice of $\gcal$ for granted and suggested some choices such as~\eqref{eq:all_dbns},~\eqref{eq:normal_mixing_class} and \eqref{eq:normal_scale_class} in our numerical examples. There are two difficulties regarding this choice; first, $\gcal$ needs to capture the true $\gprior$ and second, if $\gcal$ is infinite-dimensional, then it has to be suitably discretized to numerically solve optimization problems such as~\eqref{eq:nbhood_worst_case} or~\eqref{eq:continuous_modulus_problem}\footnote{We provide guidance for the numerical discretization of infinite-dimensional $\gcal$ in Supplement~\ref{sec:gcal_discretization}.}. Recent successful applications of empirical Bayes for point estimation use a discretized convex class $\gcal$. For example, \citet{koenker2014convex} and \citet{koenker2017rebayes} use the nonparametric maximum likelihood estimator (NPMLE) for a plethora of different likelihoods in~\eqref{eq:EB} with \smash{$\hG \in \pp(\Ksupport)$} and \smash{$\Ksupport$} a finite set (an equidistant discretization of a compact interval). The above classes \smash{$\gcal$} are typically discretized so densely that a nonparametric approach to forming confidence intervals, as pursued in this work, is warranted. In cases where a choice of \smash{$\gcal$} has been used for estimation, we suggest that the point estimates be accompanied by confidence intervals using the same choice of prior class. 

\subsection{Sensitivity analysis for $\gcal$}
\label{subsec:sensitivity}

The choice of $\gcal$, is not innocuous and the sensitivity of our intervals to the non-parametric specification of $\gcal$ is an important consideration for their practical adoption. One could hope to choose $\gcal$ on the basis of goodness-of-fit testing. However, goodness-of-fit tests are only able to rule out $\gcal$ that are inconsistent with the data, but there may be many choices of $\gcal$ that are consistent with the data, and for each of these, depending on the target of inference, the length of the confidence intervals may vary substantially or remain relatively stable. To illustrate these ideas further, we suppose that $\gprior$ is location mixture of Gaussians as in~\eqref{eq:normal_mixing_class}. These classes are nested as \smash{$\law\nn(\tilde{\tau}^2, \RR)  \supset \law\nn(\tau^2, \RR)$} for \smash{$\tilde{\tau} < \tau$}.\footnote{Suppose $\gprior \in \law\nn(\tau^2, \RR)$, i.e., $\gprior= \nn(0, \tau^2) \star H$ for $\tau >0 $ and a distribution $H$, where $\star$ denotes convolution. For any $0<\tilde{\tau}<\tau$, we can write $G = \nn(0, \tilde{\tau}^2) \star \tilde{H}$ with $\tilde{H} = \nn(0, \tau^2 - \tilde{\tau}^2) \star H$, and so $G \in \law \nn(\tilde{\tau}^2, \RR)$. } If we define,
$$ \tau^*(G) = \sup\cb{\tau >0 \mid G \in \law\nn(\tau^2, \RR)},$$
then inference using our methods will be valid with $\gcal =  \law\nn(\tau^2, \RR)$ for any $\tau \leq \tau^*(G)$, and will be more conservative, the smaller $\tau$ is. So ideally we would like to use $\tau=\tau^*(G)$, however, $\tau^*(G)$ is a one-sided discontinuous statistical functional in the sense of~\citet{donoho1988one}, so that it is impossible to derive a non-trivial lower bound on it in a data-driven way; see also~\citet{donoho2013achieving}. On the other hand, it is possible to derive upper bounds on $\tau^*(G)$.\footnote{For example, note that $\tau^*(G)^2 \leq \Var[G]{Z} - \sigma^2$.} Hence, a goodness-of-fit test may be able to reject values of $\tau$ that are too large, however it cannot disambiguate between choices of small $\tau$. 

Thus, the only way to obtain practically meaningful results is by the analyst choosing a range of $\tau$'s that appear to be plausible. The analyst can further interpret the results and evaluate how pessimistic a choice of $\tau$ (or more generally, of $\gcal$) may be by inspecting the worst-case priors that determine the confidence interval for a given estimand (the worst case priors in~\eqref{eq:nbhood_worst_case} for $F$-Localization, and the worst cases priors in~\eqref{eq:continuous_modulus_problem} for the affine minimax approach). In Supplement~\ref{sec:supplement_sensitivity}, we explore these issues in the context of the Prostate data analysis of Section~\ref{subsec:prostate} using the split-likelihood-ratio of \citet{wasserman2020universal} for goodness-of-fit testing.

While conducting the suggested sensitivity analysis, it is important to recall that the minimax estimation error for $\theta_{\gprior}(z)$ decays extremely slowly (often poly-logarithmically)
with sample size \citep{butucea2009adaptive,pensky2017minimax} for some of the problems we consider (e.g., local false sign rate in the Gaussian empirical Bayes problem). In this case, unlike in classical estimation
problems, we cannot expect to make our confidence intervals meaningfully shorter by, say, collecting
100 times more data than we have now. From this perspective, the amount of assumptions (smoothness, unimodality and so forth) we are willing to impose on $\gcal$ determines the accuracy with which we can ever hope to learn $\theta_{\gprior}(\zo)$, and
the sensitivity analysis discussed above is closely aligned with recommendations for applications
with partially identified parameters~\citep{armstrong2018optimal,imbens2017optimized,rosenbaum2002observational}.

\section{Discussion}
\label{sec:discussion}

We have presented two general approaches towards building confidence intervals for empirical Bayes estimands in model~\eqref{eq:EB} that work for any choice of $h(\cdot)$, convex class of {\textprior}s $\gcal$ and likelihood $p(\cdot \cond \mu_i)$. Our methods are computationally intensive and require repeatedly solving non-trivial convex optimization problems. Nevertheless, in light of an available software implementation, our confidence intervals are practical, and can accompany applied work using nonparametric empirical Bayes point estimates. As in ~\citet{koenker2014convex}, our implementation is facilitated by recent advances in convex optimization.

Here, we focused on inference for empirical Bayes estimands of the form $\EE[\gprior]{h(\mu) \mid \Zo = \zo}$ in model~\eqref{eq:EB}; our approach, however, can also handle other empirical Bayes estimands. As explained in Section~\ref{sec:linear_functionals}, simpler versions of our methods can be used to form confidence intervals for linear functionals of $\gprior$ such as $\PP[\gprior]{\mu \geq 0}$. Our methods are also directly applicable to tail (rather than local) empirical Bayes quantities, such as the tail (marginal) false sign rate $\PP[\gprior]{\mu \cdot \Zo \leq 0 \mid \Zo \geq \zo}$ as considered in, e.g., \citet{yu2019adaptive}.  Another important class of estimands consists of posterior quantiles
$\theta_{\gprior}^{\text{perc}}(\zo; p) = \inf\cb{t : \PP[\gprior]{\mu \leq t \cond \Zo = \zo} \geq p}$ for $p\in(0,1)$. $F$-Localization can be used to conduct inference for posterior quantiles by inverting simultaneous confidence intervals for $\PP[\gprior]{\mu \leq t \cond \Zo = \zo}$, $t\in \mathbb R$. However, it seems more challenging to generalize \Amari~to posterior quantiles.

A further important direction for future work is to handle generalizations of model~\eqref{eq:EB}. In some applications, it would be important to allow for unknown structural parameters or global parameters, such as the variance parameter $\sigma^2$ in the Gaussian model. \eqref{eq:EB} could also be extended to higher-dimensions, e.g., $\mu_i \in \RR^d$, and to heteroskedastic problems in which the likelihood $p_i(\cdot \cond \mu_i)$ can vary across $i$, for example, the Gaussian location model with per-observation noise standard deviation $\sigma_i$, so that $p_i(\cdot \cond \mu_i) = \nn\p{\mu_i, \sigma_i^2}$~\citep{gu2017unobserved, weinstein2018group}.

\subsection*{Software} We provide reproducible code for all numerical results in the following Github repository:
\url{https://github.com/nignatiadis/empirical-bayes-confidence-intervals-paper}.\\
A package implementing the method is available at \url{https://github.com/nignatiadis/Empirikos.jl}. The package has been implemented in the Julia programming language~\citep{bezanson2017julia} and depends, among others, on the packages JuMP.jl~\citep{DunningHuchetteLubin2017} and Distributions.jl~\citep{besanccon2019distributions}.

\subsection*{Acknowledgments}
This paper was first presented on May 24th, 2018 at a workshop in honor of Bradley Efron's 80th birthday. 
We are grateful to Timothy Armstrong, Bradley Efron, Jiaying Gu, Guido Imbens, Panagiotis Lolas, Michail Savvas, Paris Syminelakis, Han Wu and seminar participants at several venues for helpful feedback and discussions. We thank Jiaying Gu for suggesting the Anderson-Rubin construction. Some of the computing for this project was performed on the Sherlock cluster. We would like to thank Stanford University and the Stanford Research Computing Center for providing computational resources and support that contributed to these research results. We acknowledge support from a Ric Weiland Graduate Fellowship, a gift from Google, and
National Science Foundation grant DMS-1916163.

\bibliographystyle{plainnat}

\bibliography{lfsr}

\newpage    
\setcounter{page}{1}
\renewcommand{\thepage}{S\arabic{page}} 
\setcounter{table}{0}
\renewcommand{\thetable}{S\arabic{table}}

\setcounter{figure}{0}
\renewcommand{\thefigure}{S\arabic{figure}}

\renewcommand{\theequation}{S\arabic{equation}}

\renewcommand{\therema}{S\arabic{rema}}
\renewcommand{\theprop}{S\arabic{prop}}
\renewcommand{\thefootnote}{S\arabic{footnote}}
\begin{appendix}

\section{Gaussian $F$-localization: Proof of Proposition~\ref{prop:floc_kde}}
\label{sec:marginal_nbhood}

Throughout this section we assume that model~\eqref{eq:EB} holds with $p(\cdot \cond \mu) = \nn(\mu,\, \sigma^2)$. Furthermore, without loss of generality, we assume that $\sigma^2=1$. The key idea of the proof is the following: we first use the smoothness of $f_G(\zo)$ in the Gaussian empirical Bayes problem to verify that \smash{$\hat{f}(\zo)=\hat{f}^K_n(\zo)$}~\eqref{eq:kde} has bias of order $O(1/\sqrt{n \log(n)})$. Thus the dominant error in \smash{$\sup_{\zo \in [-M,M]}|\hat{f}(\zo) - f_G(\zo)|$} is stochastic and equal to \smash{$\sup_{\zo \in [-M,M]}|\hat{f}(\zo) - \mathbb E_{G}[\hat{f}(\zo)]|$}. We then study the variance of the stochastic term (pointwise) and use results of~\citetsupplement{chernozhukov2014gaussian, chernozhukov2016empirical} to verify the accuracy of the bootstrap approximation. 

Notation: We often omit the dependence on $n$, for example, we write $h$ for $h_n = 1/\sqrt{\log(n)}$. All integrals in this section are computed with respect to the Lebesgue measure.

\subsection{Bias of KDE}
For a function $\psi:\RR \to \RR$, we write $\psi^*(t)$ for its Fourier transform, i.e., $\psi^*(t) = \int \exp(itx)\psi(x)dx$, assuming it exists. The crucial property of the kernel $K(\cdot)$ in~\eqref{eq:kde} that we will use to control bias, is that $K^*$ is equal to $1$ on $[-1,1]$~\citep{politis1993family}:
\begin{equation}
K^*(t)= \begin{cases} 1, & \text{ if } t\in[-1,1] \\ 0, &\text{ if } |t| \geq 1.1 \\ 11-10\cdot|t|,& \text{ if } |t| \in [1,1.1] \end{cases}.
\end{equation}
We are ready to state our result on the bias.
\begin{prop}
\label{prop:kde_bias}
Consider estimating the marginal density $f_{\gprior}$ (for some effect size distribution $\gprior$) with the KDE~\eqref{eq:kde}. Then, for some constant $C$ it holds that,
$$\Bias[\gprior]{\hf(\zo),f_{\gprior}(\zo)}^2 = \p{  \EE[\gprior]{\hf(\zo)} - f_{\gprior}(\zo)}^2 \leq  C \frac{1}{n\log(n)} \; \text{for all } \zo \in \RR.$$
\end{prop}

\begin{proof}
$$
\begin{aligned}
\abs{\Bias[\gprior]{\hf(\zo), f_{\gprior}(\zo)}} &= \abs{f_{\gprior}(\zo) - \frac{1}{h}\EE[\gprior]{K\p{\frac{\Zo_i-\zo}{h}}}} \\
 &= \abs{f_{\gprior}(\zo) - \int\frac{1}{h}{K\p{\frac{u-\zo}{h}}}f_{\gprior}(u)du} \\
 &= \frac{1}{2\pi} \abs{\int\exp(-it\zo)f_{\gprior}^*(t)dt - \int f_{\gprior}^*(t)\widebar{\p{\frac{1}{h}{K\p{\frac{\cdot-\zo}{h}}}}^*(t)}dt}\\
 &= \frac{1}{2\pi} \abs{\int\exp(-it\zo)f_{\gprior}^*(t)dt - \int f_{\gprior}^*(t)\exp(-it\zo)K^*(th)dt} \\
 &\leq \frac{1}{2\pi} \int_{\cb{|t| \geq \frac{1}{h}}}|f_{\gprior}^*(t)dt| \\
 &\leq \frac{1}{2\pi} \int_{\cb{|t| \geq \frac{1}{h}}}\exp\p{-t^2/2} \\
 & \leq \frac{1}{\pi} h \exp\p{-\frac{1}{2 h^2}}\\
 & = \frac{1}{\pi} \frac{1}{\sqrt{n\log(n)}}
\end{aligned}
$$
In the 3rd line we used the Fourier inversion formula, as well as the Plancherel isometry~\citepsupplement[Theorem A.4]{meister2009deconvolution}. We also used the facts that $K$ is square integrable, $K$ is even, $K^*(t) = 1$ on $[-1,1]$, $|K^*(t)| \leq 1$ outside $[-1,1]$ and that $|f_{\gprior}^*(t)| = |\int \exp(i\mu t -t^2/2) \dgprior| \leq \exp(-t^2/2)$. Finally, we used the Gaussian tail inequality.
\end{proof}
We note that~\citetsupplement[Section 3.2]{taupin2001semi} also sketches the above argument.

\subsection{Variance of KDE}

We next study the variance term.
\begin{prop}
\label{prop:kde_var}
There exist constants $c, C > 0$ and $n_0 \in \mathbb N$, such that, 
$$\frac{c}{n h} \leq \Var[\gprior]{\hf(\zo)} \leq \frac{C}{n h} \text{ for all } \zo\in [-M,M],\, n \geq n_0$$
\end{prop}

\begin{proof}
A consequence of Proposition~\ref{prop:kde_bias} is that,
$$ \EE[\gprior]{K\p{\frac{\Zo_i-\zo}{h}}} =  h \p{ f_{\gprior}(\zo) \, + \, o(1)}.$$
For the second moment, we get:
$$
\begin{aligned}
\EE[\gprior]{K^2\p{\frac{\Zo_i-\zo}{h}}} &= \int K^2\p{\frac{u-\zo}{h}}f_{\gprior}(u)du \\
&= h  \int K^2(u)f_{\gprior}(uh + \zo)du \\
&= h\p{f_{\gprior}(\zo)\int K^2(u)du \,+\, o(1)}.
\end{aligned}
$$
Combining these two results, we find that,
$$\Var[\gprior]{\frac{1}{h}K\p{\frac{\Zo_i-\zo}{h}}} = \frac{f_{\gprior}(\zo)\int K^2(u)du \, + \, o(1)}{h} - \p{ f_{\gprior}(\zo) \, + \, o(1)}^2.$$
We conclude after noting that the $o(1)$ terms are uniform in $\zo \in [-M,M]$ and that 
$$0 < \inf_{\zo \in [-M,M]}f_{\gprior}(\zo) \leq  \sup_{\zo \in [-M,M]}f_{\gprior}(\zo) < \infty.$$
\end{proof}

\subsection{Validity of bootstrap approximation}

Let us define the following suprema,
\begin{align}
&\widehat{W} = \sqrt{nh} \sup_{\zo \in [-M,M]} \abs{ f(\zo) - \hf(\zo)} \label{eq:hatw} \\
&W =  \sqrt{nh} \sup_{\zo \in [-M,M]} \abs{ \EE{\hf(\zo)} - \hf(\zo)}  \label{eq:supw}\\
&\beta =  \sqrt{nh} \sup_{\zo \in [-M,M]}\abs{  \EE{\hf(\zo)} -  f(\zo)} \\
&W^* = \sqrt{nh} \sup_{\zo \in [-M,M]} \abs{ \hf^*(\zo) - \hf(\zo)} \label{eq:bootw},
\end{align} 
where $\hf^*(\zo)$ is a single bootstrap evaluation of the KDE as in~\eqref{eq:kde_bootstrap}. Our high-level strategy is to argue that the distribution of $W^*$ conditionally on the data $\mathbf{Z} = (Z_1,\dotsc,Z_n)$ is close to the unconditional distribution of $W$ and that \smash{$\widehat{W}$} and $W$ are essentially indistinguishable (compared to the fluctuations of the above suprema), because $\beta$, i.e. the worst case bias over $\zo$, is small (Proposition~\ref{prop:kde_bias}).

We first record the following fact: Define the class $\mathcal{H}$ of all functions that are dilations and translations of the kernel $K(\cdot)$, i.e.,
\begin{equation}
\mathcal{H} = \cb{ K(a\cdot + b) \cond a>0,\, b \in \RR}.
\end{equation}
This class has envelope function $\Norm{K(\cdot)}_{\infty}$, it is pointwise measurable and is of VC type, i.e., there exist constants $A \geq e, \; v \geq 1$ such that for any finitely discrete probability measures $Q$ and any \smash{$\varepsilon \in (0, \Norm{K(\cdot)}_{\infty}]$}, it holds that
\begin{equation}
\label{eq:vctype}
N\p{ \mathcal{H},\,\mathcal{L}^2(Q), \, \varepsilon} \leq \p{\frac{A}{\varepsilon}}^v,
\end{equation}
where \smash{$N\p{ \mathcal{H},\,\mathcal{L}^2(Q),\,  \varepsilon}$} is the $\varepsilon$-covering number of $\mathcal{H}$ with respect to the $\mathcal{L}^2(Q)$ norm. This follows directly from~\citetsupplement[Proposition 3.6.12]{gine2016mathematical}, since the kernel $K(\cdot)$ is of bounded variation.

The fact that $\mathcal{H}$ is VC type will allow us to construct a coupling $W$ and \smash{$\widetilde{W}$}, where \smash{$\widetilde{W}$} is the supremum of a Gaussian process. Concretely, let $\GG$ be a Gaussian process indexed by $[-M,M]$ with mean $0$ and covariance,
\begin{equation}
\label{eq:gp_covar}
\Cov{\GG(s), \GG(t)} =  \frac{1}{h} \cdot \Cov{ K\p{\frac{\Zo-t}{h}}, \, K\p{\frac{\Zo-s}{h}}}.
\end{equation}
 Then the following holds,
\begin{prop}
\label{prop:gaussian_approx}
$\GG$ is a tight Gaussian process on $\ell^{\infty}([-M,M])$. Furthermore, there exists a coupling $\tW,\; W$ (with $W$ defined in~\eqref{eq:supw}) and $r_1,\; r_2$ such that 
$$ \tW \, \stackrel{\mathcal{D}}{=} \, \sup_{\zo \in [-M,M]}  \GG(t)$$
and,
$$\PP{ \abs{ W - \tW} >   r_1 } \leq r_2,\;\; r_1 = O\p{ (nh)^{-1/6}\log(n)},\;\; r_2 = O(1/\log(n)). $$
\end{prop}
\noindent Similarly, we can construct a conditional coupling of the bootstrap statistic $W^*$ and $\widetilde{W}^*$, where $\tW^*$ conditionally on $\mathbf{Z}$ has the same distribution as (the unconditional law) of $\tW$.

\begin{prop}
\label{prop:bootstrap_coupling}
There exists a coupling $\tW^*,\; W^*$ (with $W^*$ defined in~\eqref{eq:bootw}), such that 
$$ (\tW^* \cond \mathbf{Z}) \, \, \stackrel{\mathcal{D}}{=} \, \sup_{\zo \in [-M,M]}  \abs{\GG(\zo)}$$
and such that there exists an event $\mathcal{E}$ with \smash{$\PP{  \mathcal{E}^c} = O(1/\sqrt{\log(n)})$} on which
$$\PP{ \abs{ W^* - \widetilde{W}^*} >  r_1^* \cond \mathbf{Z} } \leq  r_2^*,\;\, r_1^* = O\p{ (nh)^{-1/6}\log(n)},\;\, r_2^* = O(1/\sqrt{\log(n)}).$$
\end{prop}
Before proceeding with the proof of Proposition~\ref{prop:floc_kde}, we need one final ingredient. We define Levy's function for $\tW$ as
\begin{equation}
\label{eq:levy}
\levy{r} = \sup_{t \in \RR} \PP{ \abs{\tW - t} \leq r}.
\end{equation}
The following Proposition holds as a consequence of~\citetsupplement[Lemma A.1]{chernozhukov2014gaussian}.
\begin{prop} 
\label{prop:levy}
The Levy concentration function~\eqref{eq:levy} satisfies:
$$r \cdot \sqrt{\log(\log(n))} \to 0 \text{ as } n \to \infty  \Longrightarrow \levy{r} \to 0 \text{ as } n \to \infty.$$
\end{prop}
\noindent We postpone the proof of the above three Propositions to the end of this section and proceed with the main argument.

\begin{proof}[Proof of Proposition~\ref{prop:floc_kde}:]
Write $c^* = \widehat{c}_n(\alpha) \cdot \sqrt{nh}$ for the $1-\alpha$ conditional quantile of $W^*$, i.e., 
$$ c^* = \inf\cb{ t\,:\, \PP{  W^* \leq t \cond \mathbf{Z}}  \geq 1-\alpha}.$$
It holds that,
\begin{equation*}
\begin{aligned}
\PP{ F \in \ff_n } &= \PP{ \widehat{W} \leq c^* } \\
				   &\stackrel{(i)}{\geq} \PP{ W \leq c^* - \beta} \\
				   &\stackrel{(ii)}{\geq} \PP{ \tW \leq c^* - \beta} - \levy{r_1} - r_2 \\
				   &\stackrel{(iii)}{\geq} \PP{ \tW \leq c^* } - \levy{\beta} - \levy{r_1} - r_2 \\
				   &\stackrel{(iv)}{\geq} 1-\alpha - \PP{\mathcal{E}^c} - r_2^* - \levy{r_1^*} - \levy{r_1} - \levy{\beta}  - r_2\\
				   &\stackrel{(v)}{=} \, 1-\alpha - o(1) 
\end{aligned}
\end{equation*}

We justify the individual steps $(i)-(v)$. \\
\textbf{(i)} follows by the triangle inequality, since $\widehat{W} \leq W + \beta$.\\
\textbf{(ii)} follows from Proposition~\ref{prop:gaussian_approx}, along with the definition of the Levy function~\eqref{eq:levy}. In more detail:
$$
\begin{aligned}
\PP{ \tW \leq c^* - \beta} & \leq \PP{ \tW \leq c^* - \beta,\; \abs{\tW - W} \leq r_1} \, + \, \PP{ \abs{\tW - W} > r_1} \\
& \leq \PP{ \tW \leq c^* - \beta,\; \abs{\tW - W} \leq r_1, \; W \leq c^*-\beta} \,+\, \PP{0\leq W - c^*  + \beta \leq r_1 } \, +  \, r_2 \\
& \leq \PP{W \leq c^*-\beta}  \,+\, \levy{r_1} \,+\, r_2.  
\end{aligned}
$$
\\
\textbf{(iii)} follows from the definition of the Levy function~\eqref{eq:levy}.\\
\textbf{(iv)} follows by properties of the Bootstrap approximation. First, note that by definition of $c^*$, it holds 
$$ \PP{ W^* \leq c^* \cond \mathbf{Z}} \geq 1-\alpha$$
On the other hand, consider the event $\mathcal{E}$. On that event, 
$$
\begin{aligned}
\PP{\tW^* \leq c^* \cond \mathbf{Z}}  &\geq \PP{ W^* \leq c^* \cond \mathbf{Z}} - \levy{r_1^*} - r_2^*  \\
&\geq 1-\alpha - \levy{r_1^*} - r_2^*,
\end{aligned}
$$
where the first inequality follows from Proposition~\ref{prop:bootstrap_coupling} and the definition of the Levy function~\eqref{eq:levy} and the second inequality follows by definition of $c^*$. Hence,
$$ \PP{\tW^* \leq c^*} \geq \EE{ \PP{\tW^* \leq c^* \cond \mathbf{Z}};\, \mathcal{E}} \geq 1-\alpha - \levy{r_1^*} - r_2^* - \PP{\mathcal{E}^c}.$$ 
\textbf{(v)} From Propositions~\ref{prop:gaussian_approx} and~\ref{prop:bootstrap_coupling} it follows that \smash{$\PP{\mathcal{E}^c}, r_2,r_2^* = o(1)$}, so it remains to argue about the Levy terms. From the same propositions, it also holds that \smash{$r_1,r_1^* = O(1/\sqrt{\log(n)})$} and \smash{$\beta = O(1/\log(n)^{3/4})$} (Proposition~\ref{prop:kde_bias}). Thus, applying Proposition~\ref{prop:levy}, it also follows that the terms $\levy{r_1^*}$, $\levy{r_1}$ and $\levy{\beta}$ are $o(1)$.
\end{proof}

\subsection{Proofs of intermediate results}

\begin{proof}[Proof of Proposition~\ref{prop:gaussian_approx}]
This follows directly from~\citetsupplement[Proposition 3.1.]{chernozhukov2014gaussian} by noting that the first term, i.e., $O\p{ (nh)^{-1/6}\log(n)}$ is dominant for the choice $h = h_n = 1/\sqrt{\log(n)}$. Assumptions (B1), (B4), (B5) of~\citetsupplement{chernozhukov2014gaussian} hold trivially, (B2) holds by~\eqref{eq:vctype} and (B3), i.e., that $Z$ has a bounded Lebesgue density on $\RR$ follows from the fact that $f_{\gprior}(\zo) \leq 1/\sqrt{2\pi}$ for all $\zo$, since $f_{\gprior}$ is the convolution of $\gprior$ with a standard Gaussian pdf. Finally, we note that~\citetsupplement[Proposition 3.1.]{chernozhukov2014gaussian} is stated in a slightly different form than Proposition~\ref{prop:gaussian_approx}, however the result directly follows by inspecting the proof, which is based on~\citetsupplement[Corollary 2.2]{chernozhukov2014gaussian}.
\end{proof}

\begin{proof}[Proof of Proposition~\ref{prop:bootstrap_coupling}]
We seek to apply~\citetsupplement[Theorem 2.3]{chernozhukov2016empirical}. To this end, we will study the rescaled supremum $U^* = \sqrt{h} \cdot W^*$ and rescaled Gaussian process $\BB = \sqrt{h} \cdot \GG$. Then, note that as in the proof of Proposition~\ref{prop:gaussian_approx}, Assumptions (A)-(C) of~\citetsupplement[Theorem 2.3]{chernozhukov2016empirical} are satisfied. Furthermore, recalling~\eqref{eq:vctype}, in the notation of that paper, we can apply the result for $\eta \to 0$, $K_n = O(\log(n))$, $q=4$, $b$ bounded, $\sigma^2 = O(h)$ and $\gamma = 1/\log(n)$. It then follows that there exists a coupling 
$$ (\widetilde{U}^* \cond \mathbf{Z}) \, \, \stackrel{\mathcal{D}}{=} \, \sup_{\zo \in [-M,M]}  \abs{\UU(\zo)},$$
such that
$$\PP{  \abs{U^* - \widetilde{U}^* } \geq \tilde{r}_1^*} \leq \tilde{r}_2^*,$$ 
and $\tilde{r}_2^* = O(1/\log(n))$, while 
$$\tilde{r}_1^* = O\p{ \frac{\log(n)^{9/4}}{n^{1/4}} \, + \, \frac{\log(n)h^{1/3}}{n^{1/6}} \, + \, \frac{\log(n)^2 h^{1/4}}{n^{1/4}}} \, = \, O\p{\frac{\log(n)h^{1/3}}{n^{1/6}}}.$$
By applying Markov's inequality, we get that with probability at least $1-1/\sqrt{\log(n)}$, it holds that,
$$\PP{  \abs{U^* - \widetilde{U}^* } \geq \tilde{r}_1^* \cond \mathbf{Z}} \leq \tilde{r}_2^* \cdot \sqrt{\log(n)}.$$
This is the event $\mathcal{E}$ in the statement of Proposition~\ref{prop:bootstrap_coupling} and we can take $r_2^* = \tilde{r}_2^* \cdot \sqrt{\log(n)} = 1/\sqrt{\log(n)}$. Recalling that $W^*= U^*/\sqrt{h}$ and defining $\widetilde{W}^* = \widetilde{U}^*/\sqrt{h}$, we get,
$$ r_1^* = \frac{\tilde{r}_1^*}{\sqrt{h}}  = O\p{\frac{\log(n)}{(nh)^{1/6}}}.$$
\end{proof}

\begin{proof}[Proof of Proposition~\ref{prop:levy}]
First note that by Proposition~\ref{prop:kde_var}, we have that there exist $\ubar{\sigma}, \bar{\sigma}, n_0$ such that the following holds for the Gaussian process $\GG$ of Proposition~\ref{prop:gaussian_approx} (with covariance~\eqref{eq:gp_covar}):
$$\ubar{\sigma}^2 \, \leq \Var{\GG(\zo)} \leq \, \bar{\sigma}^2, \; \text{ for all } \zo \in [-M,M],\, n\geq n_0.$$
Hence, by~\citetsupplement[Lemma A.1.]{chernozhukov2014gaussian}, we have for a constant $C$ (that depends on $\ubar{\sigma}, \bar{\sigma}$) that:
\begin{equation}
\label{eq:levy_bd}
\levy{r} \leq C \cdot r \cdot \cb{ \EE{\tW} \, + \, \sqrt{\max\cb{1,\, \log(\ubar{\sigma}/r)}}}.
\end{equation}
We next bound, \smash{$\mathbb E[\tW] = \mathbb E[ \sup_{\zo \in [-M,M]} \abs{\GG(\zo)}]$}. To this end, we make the following observations: first, \smash{$\GG$} is a Gaussian process, and so in particular it is a sub-Gaussian process with respect to \smash{$d^2(\zo, \zo') = \mathbb E[\p{\GG(\zo) - \GG(\zo')}^2]$}. Applying Proposition~\ref{prop:kde_var} again, we find that \smash{$\sup_{\zo, \zo' \in [-M,M]} d(\zo, \zo')\leq D < \infty$} is finite (and $D$ can be chosen the same for all $n$). In addition, by~\eqref{eq:vctype}, we find that for some $A', v' \geq e$, 
$$N( [-M,M],\, d,\, \varepsilon) \leq \p{\frac{A'}{\varepsilon \cdot \sqrt{h}}}^{v'}.$$
and so,
$$ \int_{0}^D \sqrt{\log\p{N( [-M,M],\, d,\, \varepsilon) }}d\epsilon = O\p{ \sqrt{\log(1/h)}}.$$
By Dudley's Entropy integral~\citepsupplement[Corollary 2.2.8]{van1996weak}, we thus also get that,
$$ \EE{\tW} = O\p{ \sqrt{\log(1/h)}}.$$
Combing the above with~\eqref{eq:levy_bd}, we find that $\levy{r} \to 0$ for $r = o(1/\sqrt{\log(\log(n))})$.
\end{proof}

\section{Proofs for \Amari~inference}

\paragraph{Notation:} Throughout this supplement we omit the $M$ superscript, e.g., we write $f_{\gprior}$ instead of $f_{\gprior}^M$, $\lambda$ instead of $\lambda^M$ and so forth. 

\subsection{Properties of the modulus of continuity}
\label{sec:modulus_properties}

\begin{prop}
\label{prop:modulus_properties}
 Assume $\mathcal{G}_n$ is convex, $ \inf_{\zo} \barf(\zo) >0$ and $\sup_{\gprior \in \mathcal{G}_n} \abs{L(\gprior)} < \infty$. Then, the modulus $\omega_n(\cdot)$ defined in~\eqref{eq:continuous_modulus_problem}, as a function of $\delta>0$, has the following properties:

\begin{enumerate}[label=(\alph*)]
\item It is non-decreasing.
\item It is bounded and nonnegative.
\item It is concave.
\item For $\delta >0$, there exists an element $\omega_n'(\delta) \geq 0$ in the superdifferential of $\omega_n(\cdot)$ at $\delta$, i.e., $\omega_n'(\delta)$ satisfies the property defined in Footnote~\ref{footnote:superdifferential}. It holds that $\omega_n(\delta) \geq \delta \cdot \omega_n'(\delta)$.
\end{enumerate}
\end{prop}

\begin{proof}
(a) and (b) follow directly by the definition of $\omega_n(\cdot)$. For (c), let us take $\delta_a, \delta_b > 0$, $\lambda \in (0,1)$ and let \smash{$(\gprior_1^{\delta_a}, \gprior_{-1}^{\delta_a}), (\gprior_1^{\delta_b}, \gprior_{-1}^{\delta_b})$} solve the corresponding modulus problems. If solutions for either of these do not exist, we may take an approximate minimizer and use standard approximation arguments. Now for, $\lambda \in (0,1)$ and $\delta(\lambda) = \lambda \delta_a + (1-\lambda)\delta_b$, consider \smash{$\gprior_i^{\delta(\lambda)} = \lambda \gprior_i^{\delta_a} + (1-\lambda) \gprior_i^{\delta_b}$} with $i\in \cb{-1,1}$. Then \smash{$\gprior_i^{\delta(\lambda)} \in \mathcal{G}_n$} by convexity of \smash{$\mathcal{G}_n$} and furthermore by the triangle inequality,
$$  \cb{n \cdot \int \p{f_{\gprior_1^{\delta(\lambda)}}(\zo) - f_{\gprior_{-1}^{\delta(\lambda)}}(\zo)}^2/\barf(\zo)d\lambda(\zo)}^{1/2} \leq \lambda \delta_a + (1-\lambda) \delta_b = \delta(\lambda).$$
Hence:
$$ \omega_n(\delta(\lambda)) \geq L(\gprior_1^{\delta(\lambda)})-L(\gprior_{-1}^{\delta(\lambda)}) = \lambda \omega_n(\delta_a) + (1-\lambda) \omega_n(\delta_b).$$
To check (d), we note that the existence of $\omega_n'(\delta)$ follows from (b,c) and results from convex analysis~\citep{rockafellar1970convex}. $\omega_n'(\delta)$  satisfies the property defined in Footnote~\ref{footnote:superdifferential}, or equivalently,
\begin{equation}
\label{eq:superdifferential_rearranged}
\omega_n(\delta) -\omega_n(\tilde{\delta})  \geq \omega_n'(\delta)(\delta - \tilde{\delta}) \text{ for all }\tilde{\delta} >0.
\end{equation}
Suppose $\omega_n'(\delta) < 0$, then e.g., letting \smash{$\tilde{\delta} =2\delta$} in~\eqref{eq:superdifferential_rearranged}, it would follow that $\omega_n(\delta) -\omega_n(2\delta) > 0$, which would be a contradiction to part (a). Thus $\omega_n'(\delta) \geq 0$. Finally, by nonnegativity of $\omega_n(\cdot)$, it follows that 
$\omega_n(\delta) \geq  \omega_n(\delta) -\omega_n(\tilde{\delta})$, and taking \smash{$\tilde{\delta} \to 0$} in~\eqref{eq:superdifferential_rearranged}, we deduce that $\omega_n(\delta) \geq \delta \cdot \omega_n'(\delta)$.
\end{proof}

\subsection{Stein's heuristic}
\label{subsec:stein_heuristic}
In this Section we provide more details regarding optimization problem~\eqref{eq:minimax_problem_tractable1} and the modulus of continuity problem~\eqref{eq:continuous_modulus_problem} and provide rigorous arguments for the ideas sketched at the beginning of Section~\ref{subsec:linear_inference}

As already mentioned, at first sight, it is not obvious how to solve optimization problem~\eqref{eq:minimax_problem_tractable1}, since the problem is not concave in $\gprior$, hence standard min-max results for convex-concave problems are not applicable. Nevertheless, \citet{donoho1994statistical} provides a solution to this optimization problem by formalizing a powerful heuristic that goes back to Charles Stein. The key steps are as follows:
\begin{enumerate}
\item We search for the hardest 1-dimensional subfamily, i.e., we find $\gprior_1,\gprior_{-1} \in \gcal_n$, such that solving problem~\eqref{eq:minimax_problem_tractable1} over $\text{ConvexHull}(\gprior_1, \gprior_{-1})$ (instead of over all of $\gcal_n$) is as hard as possible. The precise definition of ``hardest'' is given by the modulus problem~\eqref{eq:continuous_modulus_problem}.
\item We find the minimax optimal estimator of problem~\eqref{eq:minimax_problem_tractable1} over the hardest 1-dimensional subfamily.
\item We then find that this solution is in fact optimal over all of $\gcal_n$.
\end{enumerate}

To implement Step 1 of the heuristic, we solve the modulus problem~\eqref{eq:continuous_modulus_problem} at $\delta$ (and assume it is solvable) and let $\gprior_1^{\delta}$, $\gprior_{-1}^{\delta}$ be solutions and $\omega_n'(\delta)$ an element of the superdifferential. Then $Q$ defined in~\eqref{eq:optimal_Q} solves the minimax problem~\eqref{eq:minimax_problem_tractable1} over \smash{$\text{ConvexHull}(\gprior_{1}^{\delta},\gprior_{-1}^{\delta})$} for $\Gamma_n = \Gamma_n(\delta) = \omega_n'(\delta)^2$ (Step 2). In fact, it solves this minimax problem over all of $\gcal_n$ (Step 3), as can be verified by the proposition below, and so~\eqref{eq:minimax_problem_tractable1} can be computed by solving the modulus problem~\eqref{eq:continuous_modulus_problem} (we postpone compoutational details to Supplement~\ref{sec:amari_computation}).

\begin{prop}[Properties of $Q$ in~\eqref{eq:optimal_Q}]
\label{prop:minimax_estimator}
Assume $\gcal_n$ is convex, \smash{$\sup_{\widetilde{\gprior} \in \gcal_n} |L(\widetilde{\gprior})| < \infty$} and that \smash{$f_{\widetilde{\gprior}}(\cdot) \in \mathcal{L}^2(\lambda)$} for all \smash{$\widetilde{\gprior} \in \gcal_n$}. Furthermore, assume that there exist $\gprior_1^{\delta},\;\gprior_{-1}^{\delta} \in \gcal_n$ that solve the modulus problem at $\delta>0$, i.e., are such that~\eqref{eq:modulus_solutions} holds and that $\inf_{\zo}\barf(\zo) >0$. Then:
\begin{enumerate}[label=(\alph*),leftmargin=*]
\item $Q$ defined by~\eqref{eq:optimal_Q}, achieves its worst case positive bias over $\gcal_n$ for estimating $L(\gprior)$ at $\gprior_{-1}^{\delta}$ and negative bias at $\gprior_1^{\delta}$, i.e., letting $\Bias[\gprior]{Q,L} = \int Q(\zo) f_{\gprior}(\zo) d\lambda(\zo) - L(\gprior)$, it holds that,
\begin{equation}
\begin{aligned}
\sup_{\gprior \in \gcal_n}\Bias[\gprior]{Q,L} &= \phantom{-} \Bias[\gprior_{-1}^{\delta}]{Q,L} \\
&= -\Bias[\gprior_{1}^{\delta}]{Q,L}\;\;\;  = - \inf_{\gprior \in \gcal_n}\Bias[\gprior]{Q,L}.
\end{aligned}
\end{equation}
\item  For $Q:\zz \to \RR$, write $\hVar[\barf]{Q} = \int Q^2(\zo) \barf(\zo) d\lambda(\zo) - \p{\int Q(\zo) \barf(\zo) d\lambda(\zo)}^2$. Let $\Gamma_n = \hVar[\barf]{Q}/n$, then for any other function $\tQ$ with $\hVar[\barf]{\tQ} \leq \Gamma_n \cdot n$, it holds that:
\begin{equation}
\sup_{\gprior \in \gcal_n} \Bias[\gprior]{\tQ,L}^2 \geq \sup_{\gprior \in \gcal_n}\Bias[\gprior]{Q,L}^2.
\end{equation}
\item  $\Gamma_n$ and the worst case bias have explicit expressions in terms of the modulus $\omega_n(\delta)$ and its superdifferential $\omega_n'(\delta)$:
\begin{align}
&\sup_{\gprior \in \gcal_n}\Bias[\gprior]{Q,L} = \frac{1}{2}\sqb{ \omega_n(\delta) - \delta \omega_n'(\delta)}, \label{eq:worst_case_bias_modulus_formula} \\
&\;\;\;\;\;\Gamma_n = \omega_n'(\delta)^2.
\end{align}
\end{enumerate}
\end{prop}

\begin{proof}
The arguments in this proof are well-known and appear in different forms for example in \citep{donoho1994statistical, low1995bias, armstrong2018optimal}. However, the statements there are provided in the context of Gaussian mean estimation and therefore we give a simplified, self-contained exposition.

\textbf{(a)} Below for notational convenience we will write $\gprior_1, \gprior_{-1}$ for $\gprior_1^{\delta}$ and $\gprior_{-1}^{\delta}$.  First let us check what the bias is at $\gprior_1$:
$$
\begin{aligned}
\Bias[\gprior_1]{Q,L} &=  \int Q(\zo) f_{\gprior_1}(\zo) d\lambda(\zo) - L(\gprior_1) \\
&=  -\frac{1}{2}(L(\gprior_{1})-L(\gprior_{-1})) \\
& \quad\;+\; \frac{n\omega'_n(\delta)}{\delta}\cb{ \displaystyle{\int} \frac{ f_{\gprior_1}(\zo) - f_{\gprior_{-1}}(\zo)}{\barf(\zo)}f_{\gprior_1}(\cdot)d\lambda(\zo) \, - \, \displaystyle{\int} \frac{\p{f_{\gprior_1}(\zo) - f_{\gprior_{-1}}(\zo)}f_{\gprior_{0}}(\zo)}{\barf(\zo)}d\lambda(\zo)} \\
&= -\frac{1}{2}\omega_n(\delta) + \frac{n\omega'_n(\delta)}{2\delta}  \displaystyle{\int} \frac{ (f_{\gprior_1}(\zo) - f_{\gprior_{-1}}(\zo))^2}{\barf(\zo)}d\lambda(\zo) \\
&= -\frac{1}{2}\omega_n(\delta) + \frac{\omega'_n(\delta)}{2\delta}\delta^2\\
&= -\frac{1}{2}\sqb{ \omega_n(\delta) - \delta \omega_n'(\delta)}
\end{aligned}
$$
Similarly, we get that: $\; \Bias[\gprior_{-1}]{Q,L} = \frac{1}{2}\sqb{ \omega_n(\delta) - \delta \omega_n'(\delta)}$. Let us now show that the worst case positive bias over $\gcal_n$ is indeed obtained at $\gprior_{-1}$. To this end take any other $\gprior \in \mathcal{G}_n$, $\gprior \neq \gprior_1, \gprior_{-1}$ and define for $\lambda \in [0,1]$:
$$
\begin{aligned}
&\Delta(\lambda) = \p{n\cdot \int \frac{\sqb{f_{\gprior_1}(\zo) - ((1-\lambda)f_{\gprior_{-1}}(\zo) + \lambda f_\gprior(\zo))}^2}{\barf(\zo)}\, d\lambda(\zo)}^{1/2} \\
&J(\lambda) = L(\gprior_{1}) - L((1-\lambda)\gprior_{-1} + \lambda \gprior) - \omega_n'(\delta)\Delta(\lambda)
\end{aligned}
$$
Observe that for any $\lambda \geq 0$:
$$
\begin{aligned}
J(\lambda) &\leq \omega_n(\Delta(\lambda)) - \omega_n'(\delta)\Delta(\lambda)\\
& \stackrel{(i)}{\leq} \sup_{\tilde{\delta} \geq 0}\cb{ \omega_n(\tilde{\delta}) - \omega_n'(\delta)\tilde{\delta}}\\
& \stackrel{(ii)}{\leq}  \omega_n(\delta) - \omega_n'(\delta)\delta\\
& = J(0)
\end{aligned}
$$
$(i)$ follows by definition of the modulus $\omega_n$ and $(ii)$ by noting that $\tilde{\delta} \mapsto \omega_n(\tilde{\delta}) - \omega_n'(\delta)\tilde{\delta}$ is concave and its superdifferential at $\delta$ includes the element $\omega_n'(\delta) - \omega_n'(\delta)=0$ (also compare to~\eqref{eq:superdifferential_rearranged}). The last equality holds by definition of $J(\lambda)$.\\
Continuing, by the chain rule and dominated convergence, it holds that $J(\lambda)$ is differentiable at $0$ and so $J'(0) \leq 0$. Furthermore,
$$ J'(0) = L(\gprior_{-1}) - L(\gprior) \, + \, \frac{n \cdot \omega_n'(\delta)}{\Delta(0)}\int\p{\frac{ f_{\gprior_1}(\zo) - f_{\gprior_{-1}}(\zo)}{\barf(\zo)}(f_{\gprior}(\zo)-f_{\gprior_{-1}}(\zo))}\, d\lambda(\zo)$$
And now also note that $\Delta(0)=\delta$ and:
$$
\begin{aligned}
\Bias[\gprior]{Q,L} - \Bias[\gprior_{-1}]{Q,L}  &= \p{\int Q(\zo) f_{\gprior}(\zo) d\lambda(\zo) - L(\gprior)} \, - \p{\int Q(\zo) f_{\gprior_{-1}}(\zo) d\lambda(\zo) - L(\gprior_{-1})} \, \\
&= L(\gprior_{-1}) -L(\gprior)  +  \frac{n \cdot \omega_n'(\delta)}{\delta}\int\p{\frac{ f_{\gprior_1}(\zo) - f_{\gprior_{-1}}(\zo)}{\barf(\zo)}(f_{\gprior}(\zo)-f_{\gprior_{-1}}(\zo))}\, d\lambda(\zo)\\
& = J'(0) \\
&\leq 0.
\end{aligned}
$$
By the above we conclude that:
$$ \Bias[\gprior]{Q,L} \leq \Bias[\gprior_{-1}]{Q,L}.$$
Finally, by repeating the same argument
$$ \Bias[\gprior]{Q,L} \geq \Bias[\gprior_{1}]{Q,L}.$$
\textbf{(b)} First let us write $\Gamma_n$ in terms of $\omega_n'(\delta)$. Note that \smash{$\int Q(\zo)\barf(\zo)\lambda(\zo)=0$}, since
$$ \int \frac{ f_{\gprior_1}(\zo) - f_{\gprior_{-1}}(\zo)}{\barf(\zo)} \barf(\zo) d\lambda(\zo) =  \int \p{ f_{\gprior_1}(\zo) - f_{\gprior_{-1}}(\zo)} \ d\lambda(\zo) = 1 -1 = 0,$$
and so,
$$
\begin{aligned}
\hVar[\barf]{Q} &= \frac{n^2 \cdot \omega_n'(\delta)^2}{\delta^2}\cdot \int \p{\frac{f_{\gprior_1}(\zo) - f_{\gprior_{-1}}(\zo)}{\barf(\zo)}}^2 \barf(\zo) d\lambda(\zo) \\
&=  \frac{n^2 \cdot \omega_n'(\delta)^2}{\delta^2} \cdot \int \frac{\p{f_{\gprior_1}(\zo) - f_{\gprior_{-1}}(\zo)}^2}{\barf(\zo)} d\lambda(\zo) \\
&=  \frac{n^2 \cdot \omega_n'(\delta)^2}{\delta^2} \cdot \frac{\delta^2}{n} \\
&= n \cdot \omega_n'(\delta)^2.
\end{aligned}
$$
Thus, $\Gamma_n = \omega_n'(\delta)^2$.

Now take any other function $\tQ(\cdot)$ and decompose it as, $\tQ(\cdot) = \tQ_0 + \tQ_1(\cdot)$, where \smash{$\tQ_0 = \int \tQ(\zo) \barf(\zo) d\lambda(\zo)$},
so that,
$$ \hVar[\barf]{\tQ} = \int \tQ_1(\zo)^2 \barf(\zo) d\lambda(\zo) \leq \Gamma_n \cdot n = n \cdot \omega_n'(\delta)^2.$$
Then:
$$
\begin{aligned}
&\Bias[\gprior_{-1}]{\tQ,L} - \Bias[\gprior_{1}]{\tQ,L} \\
=\;\; &  \p{\int \tQ_1(\zo) f_{\gprior_{-1}}(\zo) d\lambda(\zo) - L(\gprior_{-1})} \, - \p{\int \tQ_1(\zo) f_{\gprior_{1}}(\zo) d\lambda(\zo) - L(\gprior_{1})}\\
=\;\; &  L(\gprior_{1}) -  L(\gprior_{-1}) \,+ \, \int  \tQ_1(\zo)\p{ f_{\gprior_{-1}}(\zo) - f_{\gprior_{1}}(\zo)}d\lambda(\zo)  \\
=\;\; & \omega_n(\delta)  \,+ \, \int  \tQ_1(\zo)\barf(\zo)^{1/2} \frac{\p{ f_{\gprior_{-1}}(\zo) - f_{\gprior_{1}}(\zo)}}{\barf(\zo)^{1/2} } d\lambda(\zo) \\
\geq \;\; & \omega_n(\delta)  \,-\, \p{\int  \tQ_1(\zo)^2\barf(\zo)d\lambda(\zo)}^{1/2}\p{ \int \p{\frac{f_{\gprior_1}(\zo) - f_{\gprior_{-1}}(\zo)}{\barf(\zo)}}^2 \barf(\zo) d\lambda(\zo)}^{1/2}\\
\geq \;\; & \omega_n(\delta)  \,-\, \p{n \cdot \omega_n'(\delta)^2}^{1/2} \p{\delta^2 /n}^{1/2} \\
= \;\; & \omega_n(\delta)  \,-\,  \omega_n'(\delta)\delta \\
 =\;\; &2\sup_{\gprior \in \mathcal{G}_n}\cb{\abs{\Bias[\gprior]{Q, L}}}.
\end{aligned}
$$
Above we used properties of $\omega_n(\cdot)$ shown in Proposition~\ref{prop:modulus_properties}. Next,
$$
\begin{aligned}
\sup_{\gprior \in \gcal_n}\cb{\abs{\Bias[\gprior]{\tQ,L}}} &\geq \max\cb{\abs{\Bias[\gprior_{-1}]{\tQ,L}}, \abs{\Bias[\gprior_{1}]{\tQ,L}}} \\
&\geq \p{\abs{\Bias[\gprior_{-1}]{\tQ,L}} + \abs{\Bias[\gprior_{1}]{\tQ,L}}}\big/2 \\
&\geq \sup_{\gprior \in \gcal_n}\cb{\abs{\Bias[\gprior]{Q,L}}}.
\end{aligned}
$$
\textbf{(c)} We already proved these statements as intermediate steps while proving (a) and (b).

\end{proof}

\subsection{Proof of Theorem~\ref{theo:lin_functional_clt}}
\label{subsec:pf_theo_lin_clt}

\begin{proof}

A word on notation: We drop the dependence on $n$, $M$ and $\delta_n$, whenever this does not cause confusion, for example we write $\gprior_{1}$ instead of $\gprior_1^{\delta_n}$ and so forth. Furthermore, we write $\aEE{\cdot}$ for conditional expectations with respect to $\barf, \ff_n$, and $\aPP{\cdot}$,$\aVar{\cdot}$ for conditional probabilities, resp. variances.

Before embarking on the formal argument, we briefly sketch our proof strategy. Our proof makes heavy use of the representation of $Q(\cdot)$ in~\eqref{eq:optimal_Q}.  As a consequence of~\eqref{eq:optimal_Q}, it suffices to verify a central limit theorem for $\sum_{i=1}^n \tQ(\Zo_i)$, where,
\begin{equation}
\tQ(\cdot) = \frac{ f_{\gprior_1}(\cdot) - f_{\gprior_{-1}}(\cdot)}{\barf(\cdot)}.
\end{equation}
In other words, we drop the additive and multiplicative constants in front of $\tQ(\cdot)$ that appear in the expression of $Q(\cdot)$ in~\eqref{eq:optimal_Q}. A central limit theorem (CLT) for $\tQ$ directly implies a CLT for $Q$ and thus also for $\hL$. To prove the CLT for the sum of the $\tQ(\Zo_i)$, we note that $\tQ(\Zo_1), \dotsc, \tQ(\Zo_n)$ are i.i.d. conditionally on $\barf, \ff_n$, and so it suffices to verify Lindeberg's condition conditionally.

All our calculations of conditional expectations happen on the event $A_n$, which has asymptotic probability equal to $1$. By definition of the event $A_n$ in the statement of the theorem, there furthermore exists deterministic $n_0$ such that for all $n \geq n_0$, we also have that $4c_n \leq \eta$ and that $\cb{\inf_{\zo} \barf(\zo) > \eta/2, \; f_{\gprior}(\zo)/\barf(\zo)\in [1/2,2]} \subset A_n$. We assume $n \geq n_0$ henceforth.

We start by studying the (conditional) moments of $\tQ(\Zo_i)$. For the first moment, we want to argue that its square is negligible compared to the second moment. Our argument crucially depends on the following cancellation:
$$
\int \frac{ f_{\gprior_1}(\zo) - f_{\gprior_{-1}}(\zo)}{\barf(\zo)} \barf(\zo) d\lambda(\zo) = \int f_{\gprior_1}(\zo) d\lambda(\zo) \, - \, \int f_{\gprior_{-1}}(\zo) d\lambda(\zo)  = 1 -1 = 0.
$$
Using this cancellation, we get:
$$
\begin{aligned}
\abs{\aEE[\gprior]{\tQ(\Zo_i)}} &= \abs{ \int \frac{ f_{\gprior_1}(\zo) - f_{\gprior_{-1}}(\zo)}{\barf(\zo)} f_{\gprior}(\zo) d\lambda(\zo)} \\ 
&=  \abs{ \int \frac{\p{f_{\gprior_1}(\zo) - f_{\gprior_{-1}}(\zo)}}{\barf(\zo)}\p{f_{\gprior}(\zo) - \barf(\zo)}  d\lambda(\zo)} \ \\
&=  \abs{ \int \frac{\p{f_{\gprior_1}(\zo) - f_{\gprior_{-1}}(\zo)}}{\barf(\zo)}\frac{\p{f_{\gprior}(\zo) - \barf(\zo)}}{f_{\gprior}(\zo)} f_{\gprior}(\zo) \,d\lambda(\zo)} \\
& \leq \int \frac{\abs{f_{\gprior_1}(\zo) - f_{\gprior_{-1}}(\zo)}}{\barf(\zo)}\frac{\abs{f_{\gprior}(\zo) - \barf(\zo)}}{f_{\gprior}(\zo)} f_{\gprior}(\zo) \,d\lambda(\zo) \\
& \leq \frac{c_n}{\eta}  \int \frac{\abs{f_{\gprior_1}(\zo) - f_{\gprior_{-1}}(\zo)}}{\barf(\zo)}f_{\gprior}(\zo)d\lambda(\zo).
\end{aligned}
$$
We next turn to lower bound the second moment. Observe that almost surely, by Jensen's inequality:
$$
\aEE[\gprior]{\tQ(\Zo_i)^2} \geq  \aEE[\gprior]{\abs{\tQ(\Zo_i)}}^2 = \p{ \int \frac{\abs{f_{\gprior_1}(\zo) - f_{\gprior_{-1}}(\zo)}}{\barf(\zo)}f_{\gprior}(\zo)d\lambda(\zo)}^2.
$$
These two displays together imply that:
\begin{equation*}
\label{eq:first_moment_negligible}
\aEE[\gprior]{\tQ(\Zo_i)}^2 \leq \frac{c_n^2}{\eta^2}\aEE[\gprior]{\tQ(\Zo_i)^2} \leq \frac{1}{16}\aEE[\gprior]{\tQ(\Zo_i)^2}.
\end{equation*}
Next,
$$
\label{eq:second_moment_lower_bd}
\aEE[\gprior]{\tQ(\Zo_i)^2} =  \int \frac{\p{f_{\gprior_1}(\zo) - f_{\gprior_{-1}}(\zo)}^2}{\barf(\zo)} \frac{f_{\gprior}(\zo)}{\barf(\zo)}  \,d\lambda(\zo) \geq  \frac{1}{2}\int \frac{\p{f_{\gprior_1}(\zo) - f_{\gprior_{-1}}(\zo)}^2}{\barf(\zo)}  \,d\lambda(\zo) = \frac{\delta_n^2}{2n}.
$$
Hence,
$$\aVar[\gprior]{\tQ(\Zo_i)} = \aEE[\gprior]{\tQ(\Zo_i)^2} - \aEE[\gprior]{\tQ(\Zo_i)}^2 \geq \frac{15}{16} \aEE[\gprior]{\tQ(\Zo_i)^2} \geq \frac{15}{32}\frac{\delta_n^2}{n} \geq \frac{\delta_n^2}{4n}.$$
Furthermore,
$$
\begin{aligned}
\Norm{\tQ(\cdot) - \aEE[\gprior]{\tQ(\Zo_i)}}_{\infty} &\leq 2\Norm{\tQ(\cdot)}_{\infty} \\
& \leq \frac{2 \Norm{f_{\gpriorL}(\zo) - f_{\gpriorR}(\zo)}_{\infty}}{ \inf \barf(\zo)} \\
&\leq \frac{4 c_n}{\eta}.
\end{aligned}
$$
So:
\begin{equation}
\frac{\aEE[\gprior]{\abs{ \tQ(\Zo_i) - \aEE[\gprior]{\tQ(\Zo_i)}}^3}}{\aVar[\gprior]{\tQ(\Zo_i)}^{3/2}n^{1/2}} \leq \frac{\Norm{\tQ(\cdot) - \aEE[\gprior]{\tQ(\Zo_i)}}_{\infty}}{\aVar[\gprior]{\tQ(\Zo_i)}^{1/2}n^{1/2} } \leq \frac{8}{\eta}\frac{c_n}{n^{1/2}(\delta_n/n^{1/2})} \leq \frac{8c_n}{\eta \delta^{\ell}}.
\end{equation}
As argued in the beginning of the proof, this also implies the same bound for $Q(\cdot)$,
\begin{equation}
\label{eq:conditional_lyapunov}
\frac{\aEE[\gprior]{\abs{ Q(\Zo_i) - \aEE[\gprior]{Q(\Zo_i)}}^3}}{\aVar[\gprior]{Q(\Zo_i)}^{3/2}n^{1/2}}\leq \frac{8c_n}{\eta \delta^{\ell}}.
\end{equation}
We could in principle conclude now by applying the Lyapunov/Lindeberg CLT conditionally on $\ff_n, \barf$ along with Slutsky. To explain why the coverage of our intervals is uniform (under the conditions stated in the footnote of the theorem), we instead apply the Berry-Esseen bound (conditionally on $\ff_n, \barf)$ for Student's statistic~\citepsupplement[Theorem 1.1.]{bentkus1996berry}.  Recall the definition of $\hL$ in~\eqref{eq:Q}, and define
$$T = \frac{1}{n}\sum_{i=1}^n(Q(\Zo_i) - \aEE[\gprior]{Q(\Zo_i)})/\hV^{1/2} =  \left(\hL - L(\gprior) - \Bias[\gprior]{Q,L}\right)/\hV^{1/2},$$
and let $\Phi$ be the standard Normal CDF. Then, there exists a constant $C>0$, such that on the event $A_n$ and for $n$ sufficiently large,\footnote{Note that~\citetsupplement{bentkus1996berry}, do not apply the $(n-1)$ correction to the sample variance, in contrast to the definition of $\hV$ in~\eqref{eq:sample_var_est}. The additional error introduced due to this discrepancy is negligible and may be absorbed into $C$.} 
$$ \sup_{t \in \RR} \abs{ \aPP[\gprior]{T \leq t} - \Phi(t)} \leq \min\cb{1,\frac{C}{\eta}\frac{c_n}{\delta^{\ell}}}.$$
It follows (unconditionally) that,
$$
\begin{aligned}
\sup_{t \in \RR} \abs{ \PP[\gprior]{ T \leq t } - \Phi(t)} &= \sup_{t \in \RR} \abs{\EE[\gprior]{\aPP[\gprior]{ T \leq t } - \Phi(t)}} \\
&\leq  \sup_{t \in \RR} \abs{\EE[\gprior]{\ind(A_n) \cdot \p{\aPP[\gprior]{ T \leq t } - \Phi(t)}}} \, + \, \PP[\gprior]{A_n^c}\\
&\leq \frac{C}{\eta}\frac{c_n}{\delta^{\ell}} \, + \, \PP[\gprior]{A_n^c}
\end{aligned}
$$
The first part of the Theorem follows, since $c_n \to 0$ and $\PP[\gprior]{A_n^c} \to 0$ as $n \to \infty$, and so,
$$\p{\hat{L} - L(\gprior) - \Bias[\gprior]{Q,L}}/\sqrt{\hV} \xrightarrow[]{\mathcal{D}} \mathcal{N}(0,1).$$
From Proposition~\ref{prop:minimax_estimator}, we know that $\sup_{\tilde{\gprior} \in \gcal_n} \abs{\Bias[\tilde{\gprior}]{Q,L}} = \hB$, and so, on the event $\cb{\gprior \in \mathcal{G}_n} \subset A_n$ it also holds that $\abs{\Bias[\gprior]{Q,L}} \leq \hB$, i.e.,  
$$\PP[\gprior]{\abs{\Bias[\gprior]{Q,L}} \leq \hB} \geq \PP[\gprior]{A_n} = 1 - o(1).$$
It remains to prove coverage. Let $\tilde{t}_{\alpha} = t_{\alpha}(\hB/\hV^{1/2},1)$ and $\tilde{b}=\Bias[\gprior]{Q,L}/\hV^{1/2}$, then,
$$
\begin{aligned}
\PP[\gprior]{L(\gprior) \in \ii_{\alpha}} &= \PP[\gprior]{|\hat{L}-L(\gprior)| \leq t_{\alpha}(\hB,\hV)}  \\
&= \PP[\gprior]{\abs{\p{\hat{L}-L(\gprior)-\Bias[\gprior]{Q,L}}/\hV^{1/2} + \tilde{b}}  \leq \tilde{t}} \\
&= \PP[\gprior]{-\tilde{t} - \tilde{b} \leq \p{\hat{L}-L(\gprior)-\Bias[\gprior]{Q,L}}/\hV^{1/2} \leq \tilde{t} - \tilde{b}}\\
&=\EE[\gprior]{ \Phi(\tilde{t}-\tilde{b}) - \Phi(-\tilde{t}-\tilde{b})} -  o(1) \\
&\geq \EE[\gprior]{\p{\Phi(\tilde{t}-\tilde{b}) - \Phi(-\tilde{t}-\tilde{b})} \cdot  \ind\p{|\tilde{b}| \leq \hB/\hV^{1/2}}} - o(1)\\
&\stackrel{(i)}{\geq} (1-\alpha)\PP[\gprior]{|\Bias[\gprior]{Q,L}| \leq \hB} - o(1)\\
&\stackrel{(ii)}= 1-\alpha -o(1)
\end{aligned}
$$
In $(i)$ we used the definition of $t_{\alpha}(\cdot,\cdot)$ from~\eqref{eq:im_iw_ci} and in $(ii)$ we used the bound on the bias we derived above. 
\end{proof}

\subsection{Proof of Proposition~\ref{prop:applications}}

\begin{proof}
Throughout the proof we take $c_n = k^{-1/4} \log(k)$ with $k=k_n$. There are three things we need to check (for each case). \\

\noindent\text{(i)} First we check the quality of the pilot $\barf$. The key to our argument here is that
$$\sup_{t \in \RR} \abs{ \widehat{F}_n(t) - F_{\gprior}(t)} = O_{\mathbb P}(1/\sqrt{k}),$$
see e.g., below~\eqref{eq:DKW} for the justification (and note that here we use $k$ instead of $n$ samples). Hence,
$$
\begin{aligned}
t_n &:= \sup_{t \in \RR} \abs{ F_{\hG_n}(t)    - F_{\gprior}(t)} \\
	&\leq  \sup_{t \in \RR} \abs{ F_{\hG_n}(t)    -  \widehat{F}_n(t)}  \, + \, \sup_{t \in \RR} \abs{ \widehat{F}_n(t) - F_{\gprior}(t)} \\
	&\stackrel{(*)}{\leq} 2\sup_{t \in \RR} \abs{ \widehat{F}_n(t) - F_{\gprior}(t)} \\
	&= O_{\mathbb P}(1/\sqrt{k}).
\end{aligned}
$$
We note that $(*)$ follows by the definition of \smash{$\hG_n$} as the minimum distance estimator~\eqref{eq:minimum_ks_fhat}. In the remainder of the proof we seek to bound $\sup_{\zo}\abs{ f_{\gprior}(\zo) - \barf(\zo)}$ in terms of $t_n$.

For the discrete examples from part (b), we only use the fact that $\lambda$ is the counting measure on (a subset of) $\mathbb N_{\geq 0}$, so that e.g, $f_{\gprior}(0) = F_{\gprior}(0)$ and $f_{\gprior}(\zo) =  F_{\gprior}(\zo) -  F_{\gprior}(\zo-1)$ for $\zo > 0$. Hence, recalling that  \smash{$\barf(\zo) = f_{\hG_n}(\zo)$}, and by the triangle inequality, we conclude that,
$$\sup_{\zo}\abs{ f_{\gprior}(\zo) - \bar{f}(\zo)} \leq 2 \sup_{t \in \RR} \abs{ F_{\hG_n}(t)    - F_{\gprior}(t)} = 2t_n =  O_{\mathbb P}(1/\sqrt{k}).$$
Let us now turn to the Gaussian example from part (a), say with $\sigma=1$ (without loss of generality). We handle the density at $\zo=\triangleleft$, resp.$ \zo=\triangleright$ as in the discrete examples and focus now on the Lebesgue density at $\zo \in [-M,M]$.
We record the following fact. We have that,
$$ f_{\gprior}(\zo) = \EE[\gprior]{ \varphi(z-\mu)},$$
where $\varphi$ is the standard Gaussian pdf. Hence,
\begin{equation}
\label{eq:gaussian_convolution_derivative}
f_{\gprior}'(\zo) = \EE[\gprior]{ \varphi'(z-\mu)}.
\end{equation}
Consequently $f_{\gprior}'(\zo)$ is bounded, uniformly over all priors $\gprior$ and all $\zo$. Thus, by Taylor's theorem, there exists a constant $C>0$ such that for all $\gprior$, $\zo$ and $h>0$:
$$\abs{ f_{\gprior}(\zo) -   \frac{1}{h}\p{F_{\gprior}(\zo+h) - F_{\gprior}(\zo)}} \, \leq \, C \cdot h.$$
Arguing by the triangle inequality, we have that
$$\sup_{\zo}\abs{ f_{\gprior}(\zo) - \bar{f}(\zo)} \leq \frac{2t_n}{h} + 2Ch,$$
and by choosing $h = \sqrt{t_n}$, we conclude that:
$$\sup_{\zo}\abs{ f_{\gprior}(\zo) - \bar{f}(\zo)} \leq O_{\mathbb P}(\sqrt{t_n}) =  O_{\mathbb P}(k^{-1/4}).$$

\noindent \textbf{(ii)} Second we check that the localizations indeed include $F_{\gprior}$, i.e., that $\PP{ F_{\gprior} \in \ff_n} \to 1$. Here we use the fact, that all $F$-localizations considered in this proposition are nested in $\alpha$, i.e., $\ff_n(\alpha) \subset \ff_n(\alpha')$ for $\alpha' \leq \alpha$.

Fix $\varepsilon > 0$. Let $n_0$ be such that $\alpha_{n} < \varepsilon/2$ for all $n \geq n_0$ (recall that $\alpha_n \to 0$) and $n_1$ be such that,
$$ \PP[\gprior]{F_{\gprior} \in \ff_n(\varepsilon/2)} \geq 1 -\varepsilon/2 - \varepsilon/2 = 1-\varepsilon \text{ for all } n \geq n_1.$$
Such $n_1$ exists since all $F$-localizations are asymptotically valid (at a fixed confidence level) in the sense of~\eqref{eq:dbn_nbhood}. Thus, using the nestedness property, for $n \geq \max\cb{n_0, n_1}$, we have that,
$$ \PP[\gprior]{F_{\gprior} \in \ff_n(\alpha_n)} \geq \PP[\gprior]{F_{\gprior} \in \ff_n(\varepsilon/2)} \geq 1-\varepsilon.$$
Since $\varepsilon > 0$ was arbitrary, we conclude.\\

\noindent \textbf{(iii)} It remains to check that $ \Norm{f_{\gpriorL}(\cdot)-f_{\gpriorR}(\cdot)}_{\infty} \leq c_n$  with probability tending to $1$, where $\gpriorL,\gpriorR$ are the solutions to the modulus problem. Note that by definition, $F_{\gpriorL},F_{\gpriorR} \in \ff_n(\alpha_n)$. We consider each $F$-localization separately. \\

\noindent\textbf{DKW-$F$-localization:} By~\eqref{eq:DKW} (with the sample size $n$ replaced by $k$),
$$
\begin{aligned}
\sup_{t}\abs{ F_{\gpriorL}(t) - F_{\gpriorR}(t)} &\leq  \sup_{t}\abs{ F_{\gpriorL}(t) - \widehat{F}_n(t)} + \sup_{t}\abs{ F_{\gpriorR}(t) - \widehat{F}_n(t)} \\
&\leq 2\sqrt{\log\p{2/\alpha_n}\big/(2k)}.
\end{aligned}
$$
Then, we can argue as in part (i) of this proof that this implies bounds on  $ \Norm{f_{\gpriorL}(\cdot)-f_{\gpriorR}(\cdot)}_{\infty}$ in all cases (note that in the Gaussian case we use the bounded second derivative argument for $\zo \in [-M,M]$ and handle $\triangleleft,\triangleright$ as discrete).\\
\newline
\noindent\textbf{$\chi^2$-$F$-localization:}  Consider the event $C_n = \cb{\inf_{\zo}\barf(\zo) > \eta/2}$, where $\eta >0$ is such that  $\inf_{\zo} f_{\gprior}(\zo) > \eta$. $C_n$ has probability tending to $1$ as $n\to \infty$ (as follows from the proof of (i) above). Note that, on the event $C_n$, for any $\zo'$,
$$\p{f_{\gpriorR}(\zo') -  f_{\gpriorR}(\zo')}^2 \leq \frac{2}{\eta}\sum_{i} \frac{\p{f_{\gpriorR}(\zo) -  f_{\gpriorR}(\zo)}^2}{\barf(\zo)} \leq \frac{2}{\eta} \cdot \frac{\p{\delta^u}^2}{n}.$$
Thus,
$$\Norm{f_{\gpriorL}(\cdot)-f_{\gpriorR}(\cdot)}_{\infty} = O_{\mathbb P}(1/\sqrt{n}) = o_{\mathbb P}(1/\sqrt{k}).$$

\noindent\textbf{Gauss-$F$-localization:} Again all our calculations assume that the event $\cb{\inf_{\zo}\barf(\zo) > \eta/2}$ has occurred. Let us first treat $\triangleleft$ and $\triangleright$ separately. Arguing as in the case for the $\chi^2$-$F$-localization, we find that,
$$\p{f_{\gpriorR}(\triangleleft) -  f_{\gpriorR}(\triangleleft)}^2 \leq \frac{2}{\eta} \cdot \frac{\p{\delta^u}^2}{n},$$
and similarly for  $\triangleright$. Next, for $\zo \in [-M,M]$, we use the smoothness of the convolved densities. Namely, suppose,
$$\abs{f_{\gpriorL}(\zo) - f_{\gpriorR}(\zo)} =: \varepsilon > 0.$$
Then, $\abs{f_{\gpriorL}(\zo) - f_{\gpriorR}(\zo)} > \varepsilon/2$ in an interval of length $c \cdot \varepsilon$, for a small constant $c >0$, as follows by the boundedness of the derivative~\eqref{eq:gaussian_convolution_derivative}. This means that,
$$ \frac{\p{\delta^{u}}^2}{n} \geq \int_{[-M,M]}\frac{\p{f_{\gpriorL}(\zo) - f_{\gpriorR}(\zo)}^2}{\barf(\zo)}d\lambda^{\text{Leb}}(\zo) \geq c'\varepsilon^3,$$
for another constant $c'$. Note that we also used the fact that $\barf(\zo)$ is uniformly bounded. By rearranging, we find that,

$$\abs{f_{\gpriorL}(\zo) - f_{\gpriorR}(\zo)} = O_{\mathbb P}(n^{-1/3})=o_{\mathbb P}(k^{-1/3}).$$

\end{proof}

\subsection{Proof of Theorem~\ref{theo:amari}}

\begin{proof}
This proof is a continuation of the proof of Theorem~\ref{theo:lin_functional_clt}, and so we also use the notation used therein. In particular, all calculations take place on the event $A_n$ (which has probability tending to $1$ as $n\to\infty$. Furthermore, we write $\aCov{\cdot}$ for the covariance conditionally on $\ff_n, \barf$. Let \smash{$c^* = \theta_{\gprior}(\zo)$} and \smash{$\kappa^*$} be such that \smash{$c^* = \kappa^* c^{\ell} + (1-\kappa^*)c^{u}$} and \smash{$Q^{c^*} = \kappa^* Q^{\ell} + (1-\kappa^*)Q^{u}$}. Note that $\kappa^* \in [0,1]$, since on the event $A_n$, $F_{\gprior} \in \ff_n$, and so $c^{\ell} \leq c^* \leq c^u$.

We first note that our algorithm estimates the bias conservatively. To see, this, note that:
$$
\begin{aligned}
\Bias[\gprior]{Q^{c^*} ,L} &= \int Q^{c^*}(\zo) f_{\gprior}(\zo) d\lambda(\zo) - L(\gprior) \\
&=  \int \p{ \kappa^* Q^{\ell}(\zo) + (1-\kappa^*)Q^{u}(\zo)} f_{\gprior}(\zo) d\lambda(\zo) - L(\gprior) \\
&= \kappa^*\Bias[\gprior]{Q^{\ell} ,L} + (1-\kappa^*)\Bias[\gprior]{Q^{u} ,L}.
\end{aligned}
$$
Thus,
$$\sup_{\gprior}\abs{\Bias[\gprior]{Q^{c^*} ,L}} \leq \kappa^*\sup_{\gprior}\abs{\Bias[\gprior]{Q^{\ell} ,L}} + (1-\kappa^*)\sup_{\gprior}\abs{\Bias[\gprior]{Q^{u} ,L}},$$
and the RHS is precisely our bound for the worst-case bias.

We next seek to prove that the Lyapunov/Lindeberg bound~\eqref{eq:conditional_lyapunov} that holds for \smash{$Q^{\ell}, Q^u$} on the event $A_n$ also applies to \smash{$Q^{c^*}$} (with a larger constant) on a  smaller event, that however also has asymptotic probability equal to $1$ (just as $A_n$ does). The result will follow, using the Berry-Esseen bound of~\citetsupplement{bentkus1996berry} (as in the Proof of Theorem~\ref{theo:lin_functional_clt}) and the argument in the proof of Corollary~\ref{coro:naive_ar}.

All our conditional calculations occur on the event $A_n$ of Theorem~\ref{theo:lin_functional_clt}, and on the event, 
$$ \tilde{A}_n = \cb{ \aCov[\gprior]{Q^{\ell}(\Zo_i),\, Q^{u}(\Zo_i)}  \geq (-1+\varepsilon/2) \cdot \aVar[\gprior]{Q^{\ell}(\Zo_i)}^{1/2}\aVar[\gprior]{Q^{u}(\Zo_i)}^{1/2}}.$$
We will show below that $\mathbb P_{\gprior}[\tilde{A}_n] \to 1$ and so also $\mathbb P_{\gprior}[\tilde{A}_n \cap A_n] \to 1$ as $n \to \infty$. For now we seek to provide a Lyapunov/Lindeberg bound~\eqref{eq:conditional_lyapunov} for \smash{$Q^{c^*}$}. To this end, first note that on $\tilde{A}_n \cap A_n$,
$$
\begin{aligned}
&\aVar[\gprior]{Q^{c^*}(\Zo_i)} \\
=\;\;& \aVar[\gprior]{ \kappa^* Q^{\ell}(\Zo_i) + (1-\kappa^*)Q^{u}(\Zo_i)} \\
=\;\;& \p{\kappa^*}^2 \aVar[\gprior]{Q^{\ell}(\Zo_i)} +  \p{1-\kappa^*}^2 \aVar[\gprior]{Q^{u}(\Zo_i)}  + 2\kappa^*(1-\kappa^*) \aCov[\gprior]{Q^{\ell}(\Zo_i),\, Q^{u}(\Zo_i)} \\
\geq\;\;& \p{\kappa^*}^2 \aVar[\gprior]{Q^{\ell}(\Zo_i)} +  \p{1-\kappa^*}^2 \aVar[\gprior]{Q^{u}(\Zo_i)}  - 2\kappa^*(1-\kappa^*)(1-\varepsilon/2) \cdot \aVar[\gprior]{Q^{\ell}(\Zo_i)}^{1/2}\aVar[\gprior]{Q^{u}(\Zo_i)}^{1/2}\\
\geq \;\;& \p{\p{\kappa^*}^2 \aVar[\gprior]{Q^{\ell}(\Zo_i)} +  \p{1-\kappa^*}^2 \aVar[\gprior]{Q^{u}(\Zo_i)}}\cdot \varepsilon/2.
\end{aligned}
$$
In the last step we used the inequality $2ab \leq a^2 + b^2$. On the other hand,
$$
\begin{aligned}
&\aEE[\gprior]{\abs{ Q^{c^*}(\Zo_i) - \aEE[\gprior]{Q^{c^*}(\Zo_i)}}^3} \\
=\;\;& \aEE[\gprior]{\abs{ \kappa^*\p{Q^{\ell}(\Zo_i) - \aEE[\gprior]{Q^{\ell}(\Zo_i)}} + (1-\kappa^*)\p{Q^{u}(\Zo_i) - \aEE[\gprior]{Q^{u}(\Zo_i)}}}^3} \\
\leq\;\;& 8\p{\kappa^*}^3 \aEE[\gprior]{\abs{Q^{\ell}(\Zo_i) - \aEE[\gprior]{Q^{\ell}(\Zo_i)}}^3} + 8\p{1-\kappa^*}^3 \aEE[\gprior]{\abs{Q^{\ell}(\Zo_i) - \aEE[\gprior]{Q^{\ell}(\Zo_i)}}^3}. 
\end{aligned}
$$
Thus we now combine the two aforementioned inequalities,
$$
\begin{aligned}
&\frac{\aEE[\gprior]{\abs{ Q^{c^*}(\Zo_i) - \aEE[\gprior]{Q^{c^*}(\Zo_i)}}^3}}{\aVar[\gprior]{Q^{c^*}(\Zo_i)}^{3/2}n^{1/2}}\\
\leq \;\; &  \frac{8\p{\kappa^*}^3 \aEE[\gprior]{\abs{Q^{\ell}(\Zo_i) - \aEE[\gprior]{Q^{\ell}(\Zo_i)}}^3} + 8\p{1-\kappa^*}^3 \aEE[\gprior]{\abs{Q^{\ell}(\Zo_i) - \aEE[\gprior]{Q^{\ell}(\Zo_i)}}^3}}{\p{\p{\kappa^*}^2 \aVar[\gprior]{Q^{\ell}(\Zo_i)} +  \p{1-\kappa^*}^2 \aVar[\gprior]{Q^{u}(\Zo_i)}}^{3/2}\cdot (\varepsilon/2)^{3/2} n^{1/2}}\\
\leq \;\; &\frac{C}{\varepsilon^{3/2}}\cb{\frac{\aEE[\gprior]{\abs{ Q^{\ell}(\Zo_i) - \aEE[\gprior]{Q^{\ell}(\Zo_i)}}^3}}{\aVar[\gprior]{Q^{\ell}(\Zo_i)}^{3/2}n^{1/2}} \, + \, \frac{\aEE[\gprior]{\abs{ Q^{u}(\Zo_i) - \aEE[\gprior]{Q^{u}(\Zo_i)}}^3}}{\aVar[\gprior]{Q^{u}(\Zo_i)}^{3/2}n^{1/2}}        } \\
\leq \;\; & C'\frac{c_n}{\varepsilon^{3/2} \eta \delta^{\ell}}.
\end{aligned}
$$
Here $C,C'>0$ are some constants. The last step follows by applying~\eqref{eq:conditional_lyapunov} that holds for \smash{$Q^{\ell}, Q^u$} by the proof of Theorem~\ref{theo:lin_functional_clt}. It remains to prove that $\tilde{A}_n$ has asymptotic probability tending to $1$. To this end, let $\tilde{\varepsilon} >0$, and also let,
$$ W_n = \frac{1}{n}\sum_{i=1}^n \p{Q^{\ell}(\Zo_i)-\aEE[\gprior]{Q^{\ell}(\Zo_i)}}\p{Q^{u}(\Zo_i) - \aEE[\gprior]{Q^{u}(\Zo_i)}}.$$
Then:
$$
\begin{aligned}
&\aPP[\gprior]{ W_n - \aCov[\gprior]{Q^{\ell}(\Zo_i),\, Q^{u}(\Zo_i)} \geq \tilde{\varepsilon}\, \aVar[\gprior]{ Q^{\ell}(\Zo_i)}^{1/2}\aVar[\gprior]{Q^u(\Zo_i)}^{1/2}} \\
\leq \;\; & \aEE[\gprior]{ \p{Q^{\ell}(\Zo_i)-\aEE[\gprior]{Q^{\ell}(\Zo_i)}}^2  \p{Q^{u}(\Zo_i)-\aEE[\gprior]{Q^{u}(\Zo_i)}}^2}\Big/\p{n \tilde{\varepsilon}^2 \aVar[\gprior]{ Q^{\ell}(\Zo_i)}\aVar[\gprior]{Q^u(\Zo_i)}} \\
\leq \;\; &  \frac{\aEE[\gprior]{ \p{Q^{\ell}(\Zo_i)-\aEE[\gprior]{Q^{\ell}(\Zo_i)}}^4}^{1/2} \aEE[\gprior]{\p{Q^{u}(\Zo_i)-\aEE[\gprior]{Q^{u}(\Zo_i)}}^4}^{1/2}}{n \tilde{\varepsilon}^2 \aVar[\gprior]{ Q^{\ell}(\Zo_i)}\aVar[\gprior]{Q^u(\Zo_i)}}  \\
\leq \;\; &   \frac{C c_n^2}{\p{\tilde{\varepsilon} \eta \delta^{\ell}}^2},
\end{aligned}
$$
for another constant $C>0$. The argument for the last line is analogous to the argument that led up to~\eqref{eq:conditional_lyapunov}. This also means that unconditionally, 
$$
\begin{aligned}
&\PP[\gprior]{ W_n - \aCov[\gprior]{Q^{\ell}(\Zo_i),\,Q^{u}(\Zo_i)} \geq \tilde{\varepsilon} \,\aVar[\gprior]{ Q^{\ell}(\Zo_i)}^{1/2}\aVar[\gprior]{Q^u(\Zo_i)}^{1/2}} \\
\leq \;\; &  \frac{C c_n^2}{\p{\tilde{\varepsilon} \eta \delta^{\ell}}^2} \, + \, \PP[\gprior]{A_n^c} \to 0 \text{ as }  n \to \infty.
\end{aligned}
$$
By a similar argument we can prove that $n\hV^{\ell}/\aVar[\gprior]{Q^{\ell}(\Zo_i)} = 1+o_{\mathbb P}(1)$, $n\hV^{u}/\aVar[\gprior]{Q^{u}(\Zo_i)} = 1+o_{\mathbb P}(1)$ and that
$$\frac{1}{n}\sum_{i=1}^n\p{ Q^{\ell}(\Zo_i) - \aEE[\gprior]{Q^{\ell}(\Zo_i)}} \Big/\aVar[\gprior]{Q^{\ell}(\Zo_i)}^{1/2} = o_{\mathbb P}(1),$$
and similarly for $Q^{u}$. Combining the above results, with the assumption of the Theorem, it follows that $\mathbb P_{\gprior}[\tilde{A}_n \cap A_n] \to 1$.
\end{proof}

\section{Proofs for Section~\ref{sec:power} on asymptotic power}

A word on notation: We drop the dependence on $n$ and $M$, whenever this does not cause confusion. For example we may write $F_{\gprior}$ instead of $F_{\gprior}^M$. For the results for AMARI, we follow Proposition~\ref{prop:applications} and assume that the pilot quantities $\barf$ and $\ff_n$ are constructed based on $k$ samples with $k\to \infty$, $k/n \to 0$ and $k \cdot \alpha_n \to \infty$ as $n \to \infty$. To keep the notation lighter we suppose that we compute $\barf$ and $\ff_n$ based on $k$ fresh samples from model~\eqref{eq:EB} that are independent from $Z_1, \dotsc, Z_n$. The asymptotic confidence interval lengths remain the same as under the sample-splitting of Proposition~\ref{prop:applications}, because $(n-k)/n \to 1$ as $n \to \infty$.

\subsection{Poisson model (Section~\ref{subsec:power_poisson})}

\subsubsection{Moment space calculations}
We start with some preliminary definitions and a lemma that will be needed for the proofs of the theoretical results of Section~\ref{subsec:power_poisson}. For any measure $H$ supported on $[a,b]$, we write:

\begin{equation}
\label{eq:moment_def}
m_k(H) := \int_a^b \mu^k dH(\mu),\; k=0,\dotsc,M,
\end{equation}
for its moments.\footnote{$m_0(H)$ need not be $1$, since we do not only consider probability measures.} We define the moment space:
\begin{equation}
\label{eq:momentspace}
\mathcal{M} := \cb{ \p{ m_0(H),\dotsc, m_{M}(H)} \in \RR^{M+1} \cond H \text{ measure on } [a,b],\; \int \exp(\mu)dH(\mu) = 1}.
\end{equation}
The key lemma in this section is the following:
\begin{lemm}[Open in Moment space]
\label{lemm:open_moment}
Let $H$ be a measure supported on $[a,b]$, $-\infty < a < b < \infty$ with at least $M+2$ points of support, such that $\int \exp(\mu) dH(\mu) = 1$. Then, $(m_0(H), \dotsc, m_{M}(H))$ is an element of the interior of $\mathcal{M}$.
\end{lemm}

\begin{proof}
We will show at the end of the proof, that we may assume without loss of generality that there exist points $a \leq \xi_1 < \xi_2 < \dotsc < \xi_{M+2} \leq b$ such that $\min_{j=1}^{M+2} H(\cb{\xi_j}) > \zeta$ for some $\zeta >0$.

Take $\varepsilon >0$ (which we will specify later). Let $(m_0',\dotsc,m_{M}') \in \RR^{M+1}$ be such that $\Delta_k = m_k' - m_k(H)$ satisfies $\abs{\Delta_k} < \varepsilon$ for all $k=0,\dotsc,M$. We want to show that $(m_0',\dotsc,m_{M}') \in \mathcal{M}$. To this end, we will consider perturbations of $H$ of the following form. For $\mathbf{a} \in \RR^{M+2}$, we consider:
$$H_{\mathbf{a}} = H \; + \; \sum_{j=1}^{M+2} a_j \delta_{\xi_j},$$
where $\delta_{\xi}$ is the Dirac point mass at $\xi$. Our goal is to pick $\mathbf{a}$ by solving the following linear system:
$$
\begin{aligned}
&\sum_{j=1}^{M+2} a_j \exp(\xi_j) &&= \;\;0 \\
&\sum_{j=1}^{M+2} a_j \, \xi_j^k &&= \;\;\Delta_k \;\;\text{ for } \;\; k =0,\dotsc,M.
\end{aligned}
$$
Call $\bm{\Xi}$ the matrix of this linear system, then $\bm{\Xi} \mathbf{a} = (0, \Delta_0, \dotsc, \Delta_{M})^\top$, where
$$ \bm{\Xi} = \begin{pmatrix}
\exp(\xi_1) & \exp(\xi_2) & \exp(\xi_3) & \dots & \exp(\xi_{M+1}) & \exp(\xi_{M+2}) \\ 
1 & 1 & 1 & \dots & 1  & 1\\
\xi_1 & \xi_2 & \xi_3 & \dots & \xi_{M+1}  & \xi_{M+2}\\
\hdotsfor{6} \\
\xi_1^{M} & \xi_2^{M} & \xi_3^{M} & \dots & \xi_{M+1}^{M} & \xi^{M}_{M+2} \
\end{pmatrix}.$$
We will prove below that $\bm{\Xi}$ is invertible. We pick $\varepsilon >0$ small enough, so that:
$$\abs{\Delta_k} < \varepsilon \text{ for } k=0,\dotsc,M\;\;\Longrightarrow \;\; \Norm{\bm{\Xi}^{-1} (0, \Delta_0, \dotsc, \Delta_{M})^\top}_{\infty}  < \zeta/2.$$ 
With $\mathbf{a}$ picked as above, it then holds that:
$$H_{\mathbf{a}}(\cb{\xi_j}) = H(\cb{\xi_j}) \, + \, a_j \geq \zeta/2 > 0.$$
$$ \int \exp(\mu) dH_{\mathbf{a}}(\mu) = \int \exp(\mu)dH(\mu) + \sum_{j=1}^{M+2}a_j \exp(\xi_j) = 1 + 0 =1.$$
Hence $H_{\mathbf{a}}$ is a candidate measure with moments:
$$m_k(H_{\mathbf{a}}) = m_k(H) + \sum_{j=1}^{M+2}a_j \xi_j^k = m_k(H) + \Delta_k = m_k'.$$
Thus $(m_0',\dotsc,m_{M}') \in \mathcal{M}$, and so $(m_0(H), \dotsc, m_{M}(H))$ lies in the interior of $\mathcal{M}$.

We still need to prove the invertibility of $\bm{\Xi}$. To this end, we define the function

$$ D(\xi) = \begin{vmatrix}
\exp(\xi_1) & \exp(\xi_2) & \exp(\xi_3) & \dots & \exp(\xi_{M+1}) & \exp(\xi) \\ 
1 & 1 & 1 & \dots & 1  & 1\\
\xi_1 & \xi_2 & \xi_3 & \dots & \xi_{M+1}  & \xi\\
\hdotsfor{6} \\
\xi_1^{M} & \xi_2^{M} & \xi_3^{M} & \dots & \xi_{M+1}^{M} & \xi^{M} \
\end{vmatrix}$$
and want to prove that $D(\xi_{M+2}) = \abs{ \bm{\Xi}} \neq 0$. Suppose otherwise. Then $\xi_{M+2}$ is a root of $D(\cdot)$, and so are $\xi_1,\dotsc, \xi_{M+1}$. On the other hand, we can write $D(\xi)$ as:
$$ D(\xi) = b_{-1}\exp(\xi) + \sum_{j=0}^{M} b_j \xi^j,$$
for some $b_{-1}, \dotsc, b_{M} \in \RR$. If $b_{-1} =0$, then $D(\cdot)$ is a polynomial of degree $M$ and can have at most $M$ roots, which is a contradiction. If $b_{-1} \neq 0$, then by applying Rolle's theorem multiple times, we find that $b_j \exp(\cdot)$ must have a root, which is also a contradiction. Thus $D(\xi_{M+2}) \neq 0$ and $\bm{\Xi}$ is invertible. 

We also need to justify why we could assume in the beginning of the proof that there exist points $a \leq \xi_1 < \xi_2 < \dotsc < \xi_{M+2} \leq b$ such that $\min_{j=1}^{M+2} H(\cb{\xi_j}) > \zeta$ for some $\zeta >0$. We follow the proof idea of~\citet{Pinelis2017}. By assumption, the support set of $H$ consists of at least $M+2$ points in $[a,b]$. This means that there exist pairwise disjoint closed intervals $J_1, \dotsc, J_{M+2}$ in $[a,b]$ such that $H(J_j) > 0$ for all $j=1,\dotsc,M+2$. For each interval $J_j$ there exists a discrete measure $H_j$ supported on at most $M+3$ points in $J_j$ such that:
$$\int_{J_j} \exp(\mu) dH_j(\mu) = \int_{J_j} \exp(\mu) dH(\mu),\;\;\; \int_{J_j} \mu^k dH_j(\mu) = \int_{J_j} \mu^k dH(\mu),\; k=0,\dotsc,M.$$
This result follows e.g., from \citet[Theorem 1.2, Chapter II]{karlin1966tchebycheff} by noting that the proof of the invertibility of $\bm{\Xi}$ also demonstrates that $\mu \mapsto (1, \mu, \dotsc, \mu^{M}, \exp(\mu))$ is a Tchebycheff system \citep[Definition 1.1, Chapter I]{karlin1966tchebycheff} on every closed, nonempty interval that is a subset of $[a,b]$. Since $H_j(I_j) = H(I_j) > 0$, there exists $\xi_j \in I_j$ such that $H_j(\cb{\xi_j}) > 0$. 

Consider the measure $\widetilde{H}$ that is defined on Borel sets  $A \subset [a,b]$ as follows:
$$\widetilde{H}(A) = H\bigg(A \setminus \bigcup_{j=1}^{M+2} I_j\bigg) \, + \, \sum_{j=1}^{M+2} H_j(A \cap I_j).$$
Then $\widetilde{H}$ is also a measure supported on $[a,b]$ and:
$$\int_a^b \exp(\mu) d\widetilde{H}(\mu) = \int_a^b \exp(\mu) dH(\mu),\;\;\; m_k(\widetilde{H}) =m_k(H),\; k=1,\dotsc,M.$$
We may now repeat the argument of this proof with $\widetilde{H}$ replacing $H$ to arrive at the conclusion of the Lemma.
\end{proof}

\subsubsection{Proof for Proposition~\ref{prop:asymptotics_poisson_marginal_prob}: DKW-$F$-Localization.}
\label{subsubsec:poisson_pmf_dkw_proof}	

In this section we provide the proof of the statement for the DKW-$F$-Localization.

\begin{proof}
We divide the proof into three steps. Throughout we write $c_n = \sqrt{\log(2/\alpha)/(2n)}$ and $\abs{\ii}$ for the length of the DKW-$F$-localization interval.

\begin{itemize} 
\item \textbf{Step 1:} We first prove that $\abs{\ii} \leq 4c_n$ almost surely.
\item \textbf{Step 2:} Let $A_n$ be the event on which there exist distributions $F^M_{\ell}$, $F^M_{u}$ on $\cb{0,1,\dotsc,M-1,\, M, \triangleright}$ with $F^M_{\ell}$, $F^M_{u} \in \ff^{\text{DKW}}_n(\alpha)$ that make the inequalities used in Step 1 tight. We show that $\PP[G]{A_n} \to 1$ as $n \to \infty$.
\item \textbf{Step 3:} In Step 2, we ignored the fact that the distributions we constructed need not be marginal distributions in the empirical Bayes problem. Let $B_n^{u}$ be the event that the distribution $F^M_{u}$ from Step 2 may be represented as $F^M_{u} = F_{\tilde{G}}^M$ as in~\eqref{eq:marginal_density}, where $\tilde{G} \in \mathcal{G}$. $B_n^{\ell}$ is defined similarly for $F^M_{\ell}$. We then prove that also $\PP[\gprior]{B_n^{u} \cap B_n^{\ell}} \to 1$.  
\end{itemize}
Using the results from Steps 2 and 3, we see that with probability tending to $1$, it holds that $\abs{\ii} \geq 4c_n$, and so it also follows that $\abs{\ii}/(4c_n) = 1 + o_{\mathbb P_G}(1)$.

\noindent{\textbf{Step 1:}} Take any distribution $F \in \ff^{\text{DKW}}_n(\alpha)$. Then the following holds for its density at $\zo \in \cb{1,\dotsc,M}$.
\begin{equation}
\label{eq:naive_dkw_construction}
\begin{aligned}
f(\zo) &=  F(\zo) - F(\zo-1) \\
 &= \p{\widehat{F}_n(\zo) - \widehat{F}_n(\zo-1)} + \p{F(\zo) - \widehat{F}_n(\zo)} -  \p{F(\zo-1) - \widehat{F}_n(\zo-1)} \\
 &= \hat{f}_n(\zo) + \p{F(\zo) - \widehat{F}_n(\zo)} -  \p{F(\zo-1) - \widehat{F}_n(\zo-1)} \\
 &\leq \hat{f}_n(\zo) + 2c_n.
\end{aligned}
\end{equation}
In the last inequality we used the definition of the DKW band and $\hat{f}_n(\zo)=\#\cb{\Zo_i = \zo)}/n$. Similarly, we may conclude that $f(\zo) \geq \hat{f}_n(\zo) - 2c_n$. Combining these two results, we see that the DKW-$F$-Localization band for $L(\gprior) = f_{\gprior}(\zo)$ must satisfy,
$$ \ii \subset [\hat{f}_n(\zo) - 2c_n      ,\; \hat{f}_n(\zo) + 2c_n],$$
and so its length can be at most $4c_n$.\\

\noindent{\textbf{Step 2:}} We define $F_u$ as follows:
$$ F_{u}(\zo') =  \widehat{F}_n(\zo') \text{ for } \zo' \notin \cb{\zo-1,\zo},\; F_{u}(\zo-1)
 =  \widehat{F}_n(\zo-1) - c_n,\; F_{u}(\zo)=  \widehat{F}_n(\zo) + c_n.$$
$F_u$ is tight for the inequality in~\eqref{eq:naive_dkw_construction}, since:
$$ \p{F_u(\zo) - \widehat{F}_n(\zo)} -  \p{F_u(\zo-1) - \widehat{F}_n(\zo-1)} = 2c_n.$$
$F_u$ satisfies the constraints of the DKW-band, however $F_u$ is not necessarily a distribution function. Let us  define $A_{n}^{u}$ as the event on which $\widehat{F}_n(\zo-1) - c_n > \widehat{F}_n(\zo-2)$ (or $>0$ if $\zo=1$) and $\widehat{F}_n(\zo) + c_n < \widehat{F}_n(\zo+1)$. Since the true distribution is a Poisson mixture $F_G$ and $G$ is not only supported on the point $0$, it holds that $F_G(\zo-2) < F_G(\zo-1)$ (if $\zo \geq 2$ and $F_G(\zo-1) > 0$ if $\zo=1$) and that $F_G(\zo) < F_G(\zo+1)$. Since $c_n \to 0$ and by the Glivenko-Cantelli Theorem, we conclude that $\PP[\gprior]{A_{n}^{u}} \to 1$. We may similarly define and argue for $F_{\ell}$ (with a corresponding event $A_n^{\ell}$). Since $A_n \supset A_n^{\ell}\cap A_n^{u}$, we conclude.\\

\noindent{\textbf{Step 3:}} We define $H_{\gprior}$ as the measure that is absolutely continuous w.r.t. $\gprior$ with Radon-Nikodym derivative $dH_{\gprior}/d\gprior(\mu) = \exp(-\mu)$. Recall the definition of moments $m_k(\cdot)$ in~\eqref{eq:moment_def} and the moment space $\mathcal{M}$~\eqref{eq:momentspace}. For the measure $H_{\gprior}$ it holds that
$$ m_k(H_{\gprior}) =  \int_a^b \mu^k dH_{\gprior}(\mu) = \int_a^b \mu^k \exp(-\mu) \, d\gprior(\mu) = k! f_{\gprior}(k).$$
Furthermore, since $\gprior$ is supported on at least $M+2$ points, so is $H_{\gprior}$. By Lemma~\ref{lemm:open_moment}
$\p{m_0(H_{\gprior}), \dotsc, m_{M}(H_{\gprior})}$ lies in the interior of the moment space $\mathcal{M}$, i.e., there exists an open set $\mathcal{U} \subset \RR^{M+1}$ such that $\p{m_0(H_{\gprior}), \dotsc, m_{M}(H_{\gprior})} \in \mathcal{U}$ and $\mathcal{U} \subset \mathcal{M}$. We define the bijective mapping 
$$T = (T_0, \dotsc, T_{M}):  \RR^{M+1} \to \RR^{M+1},\;\; T_k((u_0,\dotsc, u_{M})) = \sum_{j=0}^k \frac{u_j}{j!},\;\;k=0,\dotsc,M.$$
Then, for two elements $\mathbf{u} =(u_0,\dotsc,u_{M}), \mathbf{v} = (v_0, \dotsc, v_{M}) \in \RR^{M+1}$, we let:
$$d(\mathbf{u},\, \mathbf{v}) = \max_{k=0}^{M} \abs{ T_k(\mathbf{u}) -  T_k(\mathbf{v})}.$$ 
$d(\cdot,\cdot)$ is a distance for $\RR^{M+1}$ that metrizes the standard topology. Thus there exists $\varepsilon >0$ so that:
$$\mathbf{u} \in \mathcal{U} \text{ for all } \mathbf{u} \text{ with } d\p{\mathbf{u},\;\; \p{m_0(H_{\gprior}), \dotsc, m_{M}(H_{\gprior})}} < \varepsilon.$$
Observing that $T_k\p{m_0(H_{\gprior}), \dotsc, m_{M}(H_{\gprior})} = F_G(k)$ for $k=0,\dotsc,M$, we conclude that:
$$
\begin{aligned}
 &d\p{T^{-1}(F_u(0), \dotsc,F_u(M)),\;\, \p{m_0(H_{\gprior}), \dotsc, m_{M}(H_{\gprior})}} \\
=\;\; &\max_{k=0}^{M}\abs{F_u(k) - F_G(k)} \;=\;\sup_{t} \abs{F_u^M(t) - F_G^M(t)}.
\end{aligned}
$$
By construction of $F_u$ in Step 2 and the Glivenko-Cantelli theorem, we see that the RHS above converges almost surely to $0$. In turn, this means that 
$$\PP{T^{-1}(F_u(0), \dotsc,F_u(M)) \in \mathcal{U}} \to 1.$$ 
On the event inside the probability above, we can find a measure $H_u$ supported on $[a,b]$, such that:
$$T\p{m_0(H_u), \dotsc, m_{M}(H_u)} = (F_u(0), \dotsc,F_u(M)).$$
Finally, let $G_u$ be the measure that is absolutely continuous w.r.t. $H_u$ with Radon-Nikodym derivative $dG_u(\mu)/dH_u(\mu) = \exp(\mu)$. $G_u$ is a probability measure, since:
$$\int dG_u(\mu) = \int \exp(\mu) dH_u(\mu) = 1,$$
by definition of the moment space $\mathcal{M}$. Furthermore, it holds that $F_u = F_{G_u}$ on $\cb{0,1\dotsc,M,\triangleright}$. We conclude after arguing analogously for $F_{\ell}$.
\end{proof}

\subsubsection{Proof for Proposition~\ref{prop:asymptotics_poisson_marginal_prob}: \Amari.}
\label{subsubsec:proof_amari_marginal_prob}

The modulus problem~\eqref{eq:continuous_modulus_problem} at $\delta>0$ takes the form:
\begin{equation}
\label{eq:modulus_problem_poisson_marginal_prob}
\underset{G_1, G_{-1} \in \gcal_n}{\text{maximize}}\;\; f_{\gprior_1}(\zo) - f_{\gprior_{-1}}(\zo)  \;\text{ s.t. }\;  \sum_{\zo' \in \cb{0,\dotsc,M}\cup\cb{\triangleright}} \frac{\p{f_{\gprior_1}(\zo') - f_{\gprior_{-1}}(\zo')}^2}{\barf(\zo')}  \leq  \frac{\delta^2}{n}.
\end{equation}
We first solve a relaxation of the above optimization problem (we will show that the relaxation is tight later), in which we introduce variables $h_j$ that formally correspond to $f_{\gprior_1}(j) -  f_{\gprior_{-1}}(j)$:
\begin{equation}
\label{eq:poisson_density_modulus_relax}
\begin{aligned}
&\underset{h_0,\dotsc,h_M, h_{\triangleright}}{\text{maximize}} 
& & h_{\zo}\\
& \text{subject to}
& & \sum_{j} \frac{h_{j}^2}{\barf(j)} \leq  \frac{\delta^2}{n} \\
&&& \sum_{j} h_{j} = 0.
\end{aligned} 
\end{equation}
This is a convex optimization problem.  Consider the Lagrangian with dual variables $\zeta \in \RR$ and $\xi < 0$ (the dual objective is unbounded for $\xi=0$):
$$\mathcal{L}(h; \xi, \zeta) = h_{\zo} \, + \, \xi \p{\sum_{j} \frac{h_{j}^2}{\barf(j)} -  \frac{\delta^2}{n}} \, + \, \zeta \sum_{j} h_{j}.$$
The derivative with respect to $h_z$ is equal to,
$$ \frac{\partial\mathcal{L}(h; \xi, \zeta)}{\partial h_z} = 1 + 2 \xi \frac{h_z}{\barf(z)} + \zeta,$$
and with respect to $h_j, j \neq z$:
$$ \frac{\partial\mathcal{L}(h; \xi, \zeta)}{\partial h_j} =  2 \xi \frac{h_j}{\barf(j)} + \zeta,$$
By the first order optimality conditions, we conclude that $h_{j}/\barf(j) = h_{j'}/\barf(j')$ for all $j, j' \neq \zo$, and so there exists a constant $t$ such that $h_{j} = t \barf(j)$ for all $j \neq z$. Furthermore, 
$$ h_{\zo} = - \sum_{j \neq \zo} h_j = - t \sum_{j \neq \zo} \barf(j) = -t (1 - \barf(\zo)).$$
Thus $t = -h_{\zo}/(1 - \barf(\zo))$ and we only need to optimize in~\eqref{eq:poisson_density_modulus_relax} with respect to a single parameter $h_{\zo}$, which we seek to maximize. The pseudo-$\chi^2$ constraint takes the form:
$$
\begin{aligned}
\sum_{j} \frac{h_{j}^2}{\barf(j)} &= \frac{h_z^2}{\barf(z)} + \sum_{j \neq z} t^2\barf(j) \\
&=\frac{h_z^2}{\barf(z)} + t^2(1-\barf(z)) \\
&= \frac{h_z^2}{\barf(z)}  + \frac{h_{z}^2}{1-\barf(z)} \\
&= \frac{h_z^2}{\barf(z)(1-\barf(z))}.
\end{aligned}
$$
We want the above to be equal to $\delta^2/n$, and so to maximize $h_z$ subject to the above constraint, we get
\begin{equation}
\label{eq:hz_optimal}
h_{\zo} = \frac{\delta}{\sqrt{n}} \sqrt{\barf(\zo)(1-\barf(\zo)) },
\end{equation}
as the optimal value of the relaxed modulus problem~\eqref{eq:poisson_density_modulus_relax}. To argue that the relaxation is tight (with probability tending to $1$), we need to exhibit priors $\gprior_1, \gprior_{-1} \in \gcal_n$ such that $f_{\gprior_1}(j) -  f_{\gprior_{-1}}(j) = h_j$, where $(h_j)_j$ is the maximizer of~\eqref{eq:poisson_density_modulus_relax} derived above. To do so, we proceed as follows. Let $\hat{f}_n^{\text{DKW}}(\zo)$ be the frequency of $\zo$ in the sample used to construct the pilot DKW-$F$-localization. Then define the pmf $f_u$ on $\cb{0,\dotsc,M, \triangleright}$ as:
$$f_u(\zo) =  \hat{f}_n^{\text{DKW}}(\zo) + \frac{h_{\zo}}{2}, \;\;\; f_u(\zo') = \hat{f}_n^{\text{DKW}}(\zo') - \frac{h_{\zo'}\barf(\zo')}{2(1 - \barf(\zo))} \text{ for } \zo' \in \cb{0,\dotsc,M, \triangleright}\setminus\cb{\zo}.$$
We make the following observation. First, $h_{\zo}$ in~\eqref{eq:hz_optimal} is of order $O_{\mathbb P}(1/\sqrt{n})$, which is of smaller order than the width of the DKW band $O_{\mathbb P}(\sqrt{\log(2/\alpha_n)}/\sqrt{k})$. Thus, arguing as in Steps 2 and 3 of Supplement~\ref{subsubsec:poisson_pmf_dkw_proof}, we can prove that the following event has probability tending to $1$: there exists a prior $G_1 = G_{1,n} \in \gcal$ such that $F_{G_1} \in \mathcal{F}_n$ and such that $f_{G_1}(\zo') = f_u(\zo')$ for all $\zo' \in \cb{0,\dotsc,M, \triangleright}$. The same argument also applies to $f_{\ell}$ defined as $f_{\ell}(\zo') = f_u(\zo') - h_{\zo'}$ for $\zo' \in \cb{0,\dotsc,M, \triangleright}$. We thus conclude that the optimal value of the modulus problem~\eqref{eq:modulus_problem_poisson_marginal_prob} is equal to (making the dependence of $\barf(\zo)$ on $n$ explicit):
$$ \omega_n(\delta) \, = \, \frac{\delta}{\sqrt{n}} \sqrt{\barf_n(z)(1-\barf_n(z)) }.$$
$\omega_n(\delta)$ is differentiable in $\delta$ with derivative
$$\omega_n'(\delta) = \frac{1}{\sqrt{n}}\cdot \sqrt{\bar{f}_n(\zo)(1-\bar{f}_n(\zo))} = \omega_n(\delta) /\delta.$$
Let us plug the above into~\eqref{eq:optimal_Q} to find the optimal $Q(\cdot)$. First, we consider the part of $Q$ that is a function of $z$.  For $j \neq \zo$ we get 
\begin{equation*}
\begin{aligned}
\frac{n\cdot \omega_n'(\delta)}{\delta}\cdot \frac{f_{\gprior_1}(j) - f_{\gprior_{-1}}(j)}{\barf_n(j)} &= -\frac{n\cdot \omega_n'(\delta)}{\delta}\cdot \frac{h_z }{1 - \barf_n(z)} \\
&= -\frac{n\cdot \omega_n'(\delta)}{\delta}\cdot \frac{\omega_n(\delta)}{1 - \barf_n(z)} \\
&= - \frac{n}{\delta^2} \frac{\omega_n(\delta)^2}{1 - \barf_n(z)}\\
&=  - \barf_n(z).
\end{aligned}
\end{equation*}
For $\zo$ we get:
\begin{equation*}
\frac{n\cdot \omega_n'(\delta)}{\delta}\cdot \frac{f_{\gprior_1}(z) - f_{\gprior_{-1}}(z)}{\barf_n(z)} = \frac{n\cdot \omega_n'(\delta)}{\delta} \cdot \frac{\omega_n(\delta)}{\barf_n(z)} = \frac{n \cdot \omega_n(\delta)^2}{\delta^2 \barf_n(z)} = \frac{\barf_n(1)(1-\barf_n(z))}{\barf_n(z)} = 1-\barf_n(z).
\end{equation*}
It remains to evaluate the additive component in~\eqref{eq:optimal_Q} that does not depend on $\zo$. Let $\gprior_0 = (\gprior_{1}+\gprior_{-1})/2$. We note that by construction $L(\gprior_0) = f_{\gprior_0}(\zo)$. Hence the constant term is equal to:
\begin{equation*}
\begin{aligned}
& -\frac{n\cdot \omega_n'(\delta)}{\delta}\cdot \sum_{j} \frac{\p{f_{\gprior_1}(j) - f_{\gprior_{-1}}(j)}f_{\gprior_{0}}(j)}{\barf_n(j)} \, + \, L(\gprior_0) \\
= \;\; &  -\frac{n\cdot \omega_n(\delta)^2}{\delta^2}\cdot \p{ \frac{ f_{\gprior_0}(\zo)}{\barf_n(\zo)} - \frac{1- f_{\gprior_0}(\zo)}{1-\barf_n(\zo)}}   \, + \, L(\gprior_0) \\
= \;\; & \barf_n(\zo)(1- f_{\gprior_0}(\zo)) - (1-\barf_n(\zo)) f_{\gprior_0}(\zo)  \, + \, L(\gprior_0) \\
= \;\; & \barf_n(\zo).
\end{aligned}
\end{equation*}
We conclude that for $j \neq \zo$, $Q(j) =  - \barf_n(\zo) + \barf_n(\zo) =0$ and for $\zo$, $Q(\zo) = 1-\barf_n(\zo) +  \barf_n(\zo) =1$. Thus $Q(\cdot) = \ind(\cdot = z)$. Hence $\hL = \hat{f}_n(\zo) = \#\cb{\Zo_i = \zo}/n$. The worst case absolute bias of $\hL$ is given by:
$$ \hB = \frac{1}{2}(\omega_n(\delta)-\delta \omega_n'(\delta)) = 0.$$
The confidence intervals of \Amari~in~\eqref{eq:im_iw_ci} ($\hV$ as in~\eqref{eq:sample_var_est}) hence take the form:
$$\hL \pm q_{1-\alpha/2} \cdot \sqrt{\hV},$$
with $q_{1-\alpha/2}$ the $1-\alpha/2$ quantile of the standard Normal distribution. Since $n \hV \to f_{\gprior}(\zo)(1-f_{\gprior}(\zo))$, we conclude.

\subsubsection{Proof for Proposition~\ref{prob:asymptotics_poisson_robbins}: DKW-$F$-Localization}
\label{subsubsec:floc_poisson_robbins_proof}
\begin{proof}
The proof will be structured very similarly to the proof in Supplement~\ref{subsubsec:poisson_pmf_dkw_proof} that concerned inference for $f_{\gprior}(\zo)$. In particular, we follow the same three steps as in that proof.	

\noindent{\textbf{Step 1:}} Take any distribution $F \in \ff^{\text{DKW}}_n(\alpha)$. Write $\theta(\zo) = (\zo+1)f(\zo+1)/f(\zo)$, we seek to provide lower and upper bounds on it. Note that when $F=F_{\gprior}$, then by~\eqref{eq:robbins_poisson}, we have that \smash{$\theta(\zo) = \EE[\gprior]{\mu \cond \Zo=\zo}$}.

\begin{equation}
\label{eq:naive_dkw_construction_postmean}
\begin{aligned}
(\zo+1)\frac{f(\zo+1)}{f(\zo)} &=  (\zo+1)\frac{F(\zo+1) - F(\zo)}{F(\zo) - F(\zo-1)} \\
 &= (\zo+1)\frac{\hat{f}_n(\zo+1) + \p{F(\zo+1) - \widehat{F}_n(\zo+1)} -  \p{F(\zo) - \widehat{F}_n(\zo)}}{\hat{f}_n(\zo) + \p{F(\zo) - \widehat{F}_n(\zo)} -  \p{F(\zo-1) - \widehat{F}_n(\zo-1)}} \\
 &\leq  (\zo+1)\frac{\hat{f}_n(\zo+1) + 2c_n}{\hat{f}_n(\zo) - 2c_n}.
\end{aligned}
\end{equation}
In the last inequality we used the definition of the DKW band. Similarly, we may conclude that
\begin{equation}
\label{eq:naive_dkw_construction_postmean_lowerbound}
(\zo+1)\frac{f(\zo+1)}{f(\zo)}  \geq  (\zo+1)\frac{\hat{f}_n(\zo+1) - 2c_n}{\hat{f}_n(\zo) + 2c_n}.
\end{equation}
Combining these two results, we see that the DKW-$F$-Localization band for $L(\gprior) = f_{\gprior}(\zo)$ must satisfy,
$$ \ii \subset \ii' := \sqb{(\zo+1)\frac{\hat{f}_n(\zo+1) - 2c_n}{\hat{f}_n(\zo) + 2c_n},\;\; (\zo+1)\frac{\hat{f}_n(\zo+1) + 2c_n}{\hat{f}_n(\zo) - 2c_n}}.$$
We will prove that the above inclusion is in fact an equality below (with high probability and for large $n$). For now, we verify that $\ii'$ has the claimed asymptotic length.
\begin{equation}
\label{eq:dkw_postmean_extended_length}
\begin{aligned}
\sqrt{n}\abs{\ii'} &= \sqrt{n}(\zo+1) \frac{4c_n \p{\hat{f}_n(\zo+1) + \hat{f}_n(\zo)}}{\p{\hat{f}_n(\zo)-2c_n}\p{\hat{f}_n(\zo)+2c_n}} \\
&= 2(\zo+1)\sqrt{2\log(2/\alpha)}\frac{f_{\gprior}(\zo) + f_{\gprior}(\zo+1)}{f_{\gprior}(\zo)^2} + o_{\mathbb P_{\gprior}}(1).
\end{aligned}
\end{equation}

\noindent{\textbf{Steps 2 and 3:}} We define $F_u$ as follows:
$$
\begin{aligned}
&F_{u}(\zo') =  \widehat{F}_n(\zo') \text{ for } \zo' \notin \cb{\zo-1,\zo, \zo+1},\\
&F_{u}(\zo-1) =  \widehat{F}_n(\zo-1) + c_n,\; F_{u}(\zo)=  \widehat{F}_n(\zo) - c_n,\; F_{u}(\zo+1)=  \widehat{F}_n(\zo+1) + c_n.
\end{aligned}
$$
Plugging $F_u$ into~\eqref{eq:naive_dkw_construction_postmean}, we see that the last inequality is an equality. Furthermore, arguing as in Steps 2 and 3 of Supplement~\ref{subsubsec:poisson_pmf_dkw_proof}, we can prove that the following event has probability tending to $1$: There exists a prior $G = G_{n} \in \gcal$ such that $F_{G} \in \mathcal{F}_n$ and such that $F_{G}(\zo') = F_u(\zo')$ for all $\zo' \in \cb{0,\dotsc,M, \triangleright}$. We may define $F_{\ell}$ that makes~\eqref{eq:naive_dkw_construction_postmean_lowerbound} tight analogously. Hence, on an event that has probability tending to $1$, it holds that $\ii = \ii'$. Thus the asymptotic length of $\ii$ is equal to the asymptotic length of $\ii'$ computed in~\eqref{eq:dkw_postmean_extended_length}.
\end{proof}

\subsubsection{Proof for Proposition~\ref{prob:asymptotics_poisson_robbins}: \Amari.}

\begin{proof}

In Supplement~\ref{subsubsec:proof_amari_marginal_prob} we solved the modulus problem in the Poisson problem, when $L(\gprior) = f_{\gprior}(\zo)$, and proved that the optimal $Q(\cdot)$ in~\eqref{eq:optimal_Q} takes the form $Q(\cdot) = \ind(\cdot = \zo)$ on an event with asymptotic probability $1$. Here we will start by proving a generalization of the above result.

Concretely, we will fix $a = (a_0, \dotsc, a_M, a_{\triangleright})$ and we will consider the linear functional $L(\gprior) = \sum_{\zo' \in \cb{0,\dotsc,M}\cup \cb{\triangleright}} a_{\zo'} f_{\gprior}(\zo')$. For notational convenience (and with some abuse of notation), we identity $f_{\gprior}$ with the vector $(f_{\gprior}(0), \dotsc, f_{\gprior}(M), f_{\gprior}(\triangleright))$. We also define the matrix $\bar{D} =   (\barf_{\gprior}(0), \dotsc, \barf_{\gprior}(M), \barf_{\gprior}(\triangleright))$ and write $h=(h_0,\dotsc,h_M, h_{\triangleright})$. The relaxed modulus problem  (compare to~\eqref{eq:poisson_density_modulus_relax}) takes the form:
\begin{equation}
\label{eq:poisson_linear_modulus_relax}
\begin{aligned}
&\underset{h}{\text{maximize}} 
& & a^\top h\\
& \text{subject to}
& & \sum_{j} \frac{h_{j}^2}{\barf(j)} \leq  \frac{\delta^2}{n} \\
&&& \sum_{j} h_{j} = 0.
\end{aligned} 
\end{equation}
Instead of the relaxed modulus problem, we consider the relaxed inverse modulus problem,\footnote{The proof of Theorem 3 in~\citet{cai2003note} uses a similar proof technique using the inverse modulus of continuity.} which is parameterized by $t>0$:
\begin{equation}
\label{eq:poisson_lincomb_inverse_modulus_relax}
\begin{aligned}
&\underset{h}{\text{minimize}} 
& & h^\top\bar{D}^{-1} h \\
& \text{subject to}
& & a^\top h = t \\
&&& \mathbf{1}^\top h = 0.
\end{aligned} 
\end{equation}
\eqref{eq:poisson_lincomb_inverse_modulus_relax} will enable us to also solve \eqref{eq:poisson_linear_modulus_relax} and then to compute $\omega_n(\delta)$. \eqref{eq:poisson_lincomb_inverse_modulus_relax} is also a convex problem, so we introduce the Lagrangian (with dual variables $\xi, \zeta \in \RR$)
$$\mathcal{L}(h; \zeta, \xi) = h^\top\bar{D}^{-1} h \, + \, 2\xi (a^\top h - t) \, + \, 2\zeta \mathbf{1}^\top h.$$
We multiply the dual variables by $2$ only for convenience. By the first order optimality conditions, we see that:
$$ \bar{D}^{-1}h = - \xi a - \zeta \mathbf{1} \, \Longrightarrow \, h = -\bar{D}\p{\xi a + \zeta \mathbf{1}}.$$
$\xi$ and $\zeta$ are determined by a system of two linear equations. Namely from $a^\top h=t, \mathbf{1}^\top h = 0$.
$$
\begin{aligned}
&\xi a^\top \bar{D} a \, &&+ \,\, \zeta  a^\top \bar{D} \mathbf{1} &&= \;\, -t \\
&\xi a^\top \bar{D} \mathbf{1} \, &&+ \,\, \zeta &&= 0. \;\, \\
\end{aligned}
$$
In deriving the second of the above inequalities, we used the fact that $\mathbf{1}^\top \bar{D} \mathbf{1} = 1$. It follows that:
$$  \xi = t \bigg / \cb{ \p{a^\top \bar{D} \mathbf{1}}^2 - a^\top \bar{D} a},\;\; \zeta = - \xi a^\top \bar{D} \mathbf{1}.$$
The objective value of~\eqref{eq:poisson_lincomb_inverse_modulus_relax} is then equal to:

$$
\begin{aligned}
h^\top\bar{D}^{-1}h &= \xi^2 a^\top\bar{D} a + \zeta^2 + 2\xi \zeta a \bar{D}\mathbf{1}\\
&= \xi^2\cdot \cb{a^\top\bar{D} a - \p{a^\top \bar{D} \mathbf{1}}^2}\\
&= t^2 \bigg / \cb{a^\top\bar{D} a - \p{a^\top \bar{D} \mathbf{1}}^2}.
\end{aligned}
$$
We have solved the relaxed inverse modulus problem. This yields the solution to the relaxed modulus problem~\eqref{eq:poisson_linear_modulus_relax} by choosing $t$ so that $h^\top\bar{D}^{-1}h = \delta^2/n$, and so, the optimal value of~\eqref{eq:poisson_linear_modulus_relax} is equal to:
$$ t^2 = \frac{\delta^2}{n}\cb{a^\top\bar{D} a - \p{a^\top \bar{D} \mathbf{1}}^2}.$$
Arguing as in Supplement~\ref{subsubsec:proof_amari_marginal_prob}, we find that the relaxed modulus problem is tight for the modulus problem (with probability tending to $1$) and on the latter event:
$$ \omega_n(\delta) = \frac{\delta^2}{n}\cb{a^\top\bar{D} a - \p{a^\top \bar{D} \mathbf{1}}^2}.$$
Consequently, $\omega_n$ is differentiable at $\delta>0$ with $\omega_n'(\delta) = \omega_n(\delta)/\delta$. We can continue as in the proof in Section~\ref{subsubsec:proof_amari_marginal_prob} by plugging the above into~\eqref{eq:optimal_Q}. The constant additive part of $Q(\cdot)$ is equal to $a^\top \bar{D} \mathbf{1}$. Hence, identifying $Q$ with the vector $(Q(0), \dotsc, Q(M), Q(\triangleright))$, we find that: 
$$
\begin{aligned}
Q &= \frac{n\cdot \omega_n'(\delta)}{\delta} \cdot \bar{D}^{-1} h \, + \, \p{a^\top \bar{D} \mathbf{1}}\mathbf{1} \\
&= -\frac{n \omega_n(\delta)}{\delta^2}   \p{\xi a +\zeta \mathbf{1}}  \, + \, \p{a^\top \bar{D} \mathbf{1}}\mathbf{1} \\ 
&= -\frac{n \omega_n(\delta)}{\delta^2}  \xi \cb{a - \p{a^\top \bar{D} \mathbf{1}} \mathbf{1}}  \, + \, \p{a^\top \bar{D} \mathbf{1}}\mathbf{1} \\
&= a
\end{aligned}
$$
In the last step, we used the fact that $-\frac{n \omega_n(\delta)}{\delta^2}\xi=1$, since:
$$
\begin{aligned}
-\frac{n \omega_n(\delta)}{\delta^2}\xi &= -\frac{n \omega_n(\delta)}{\delta^2} \omega_n(\delta)\bigg / \cb{ \p{a^\top \bar{D} \mathbf{1}}^2 - a^\top \bar{D} a}\\
&= -\frac{n}{\delta^2} \cdot \frac{\delta^2}{n}\cb{a^\top\bar{D} a - \p{a^\top \bar{D} \mathbf{1}}^2}\bigg / \cb{ \p{a^\top \bar{D} \mathbf{1}}^2 - a^\top \bar{D} a} =1.
\end{aligned}
$$
We note that the resulting estimator for $L(\gprior)$ is unbiased, i.e., the worst case bias in this case is equal to $0$.

We are ready to return to the study of Algorithm~\ref{alg:amari_ci}. Let $[c^{\ell}, c^u]$ be the pilot $F$-localization intervals for $\theta_{\gprior}(\zo)$. By construction $\PP[\gprior]{\theta_{\gprior}(\zo) \in [c^{\ell}, c^u]} \to 1$, and furthermore, by the proof in Supplement~\ref{subsubsec:floc_poisson_robbins_proof}, we also have that $c^{\ell} = \theta_{\gprior}(\zo) + o_{\mathbb P_{\gprior}}(1)$ and  $c^u = \theta_{\gprior}(\zo) + o_{\mathbb P_{\gprior}}(1)$. By the preceding argument, we get for $\zo' \in \cb{0,\dotsc,M,\triangleright}$:
$$Q^{\ell}(\zo')  = (\zo+1)\ind(\zo' = \zo+1) - c^{\ell} \ind(\zo' = \zo+1),\; Q^{u}(\cdot)  = (\zo+1)\ind(\zo' = \zo+1) - c^{u} \ind(\zo' = \zo+1).$$
In particular $\hL^{\ell} = (\zo+1)\hat{f}_n(\zo+1)-c^{\ell}\hat{f}_n(\zo)$, $\hL^{u} = (\zo+1)\hat{f}_n(\zo+1)-c^{u}\hat{f}_n(\zo)$, and for any $c \in [c^{\ell}, c^u]$, $\hL^c = (\zo+1)\hat{f}_n(\zo+1)-c\hat{f}_n(\zo)$. Next note that $\hB^c = 0$ and since $t_\alpha(0, V) = q_{1-\alpha/2}\sqrt{V}$ in~\eqref{eq:im_iw_ci}, to determine the \Amari~confidence interval for $\theta_{\gprior}(\zo)$, we need to determine all $c \in [c^{\ell},\,c^u]$ such that:
$$0 \in \widetilde{\ii}_{\alpha}(\zo; c) =  (\zo+1)\hat{f}_n(\zo+1)-c\hat{f}_n(\zo) \pm q_{1-\alpha/2} \sqrt{\hV^c}.$$
To do this it will be furthermore convenient to express $\hV^c$ in Algorithm~\ref{alg:amari_ci} in a slightly different form, namely
$$
\begin{aligned}
&\hV^c = (\zo+1)^2 \tilde{V}_2 + c^2\tilde{V}_1 - 2c(\zo+1) \tilde{V}_{12}, &&\tilde{V}_1 = \frac{1}{n-1} \hat{f}_n(\zo)(1-\hat{f}_n(\zo)),\\
&\tilde{V}_2 = \frac{1}{n-1} \hat{f}_n(\zo+1)(1-\hat{f}_n(\zo+1)),  &&\tilde{V}_{12} = -\frac{1}{n-1}\hat{f}_n(\zo+1) \hat{f}_n(\zo).
\end{aligned}
$$
Having rewritten $\hV^c$ as above and shortening $q=q_{1-\alpha/2}^2$, we see that:
$$ 0 \in \widetilde{\ii}_{\alpha}(\zo; c) \; \Longleftrightarrow \; \p{  (\zo+1) \hat{f}_n(\zo+1) - c \hat{f}_n(\zo)}^2 \leq q^2 \sqb{ (\zo+1)^2 \tilde{V}_2 + c^2\tilde{V}_1 - 2c(\zo+1) \tilde{V}_{12}}.$$
The latter condition is a quadratic inequality in $c$, that we may rearrange as:
$$
\p{ \hat{f}_n(\zo)^2 - q^2 \tilde{V}_1}c^2 \, + \, 2(\zo+1)\p{q^2 \tilde{V}_{12} - \hat{f}_n(\zo+1)\hat{f}_n(\zo)}c \, + \,  (\zo+1)^2 \p{\hat{f}_n(\zo+1)^2 - q^2 \tilde{V}_2} \leq 0.
$$
We make the observation that $c = (\zo+1)\hat{f}_n(\zo+1)/\hat{f}_n(\zo)$ is an interior point of the above inequality. Furthermore, on the event \smash{$\{\hat{f}_n(\zo)^2 > q^2 \tilde{V}_1\}$} the above is a convex quadratic, and so the set of $c$ satisfying the inequality must be a closed interval. Since \smash{$\PP[\gprior]{\hat{f}_n(\zo)^2 > q^2 \tilde{V}_1} \to 1$} as $n \to \infty$, we restrict attention to that event. On that event, the distance between the two roots of the quadratic is equal to:
$$ \frac{2(\zo+1)q}{ \hat{f}_n(\zo)^2 - q^2 \tilde{V}_1 } \sqrt{ q^2 \p{\tilde{V}_{12}^2 - \tilde{V}_1 \tilde{V}_2} +  \p{\tilde{V}_1 \hat{f}_n(\zo+1)^2 + \tilde{V}_2 \hat{f}_n(\zo)^2 - 2 \hat{f}_n(\zo)\hat{f}_n(\zo+1) \tilde{V}_{12}}}.$$
Noting that $n \tilde{V}_1 = f_{\gprior}(\zo)(1-f_{\gprior}(\zo)) + o_{\mathbb P_{\gprior}}(1)$,\; $n \tilde{V}_2 = f_{\gprior}(\zo+1)(1-f_{\gprior}(\zo+1)) + o_{\mathbb P_{\gprior}}(1)$ and  $n \tilde{V}_{12} = - f_{\gprior}(\zo) f_{\gprior}(\zo+1) + o_{\mathbb P_{\gprior}}(1)$, we conclude that the above is asymptotically equal to:
$$
\begin{aligned}
&\frac{2(\zo+1)q}{ \sqrt{n}f(\zo)^2}\sqrt{ f(\zo)(1-f(\zo))f(\zo+1)^2 + f(\zo+1)(1-f(\zo+1))f(\zo)^2  + 2 f(\zo)^2 f(\zo+1)^2}(1+o_{\mathbb P_{\gprior}}(1))\\
=\; &\frac{2(\zo+1)q}{ \sqrt{n}f(\zo)^2}\sqrt{ f(\zo)f(\zo+1)\p{ f(\zo)+f(\zo+1)}}(1+o_{\mathbb P_{\gprior}}(1)).
\end{aligned}
$$
This is the confidence interval length claimed in the statement of the Proposition.
\end{proof}

\subsection{Bernoulli model (Section~\ref{subsec:bernoulli_partial})}

In this section we consider model~\eqref{eq:EB} with $\Zo_i \cond \mu_i \; \sim \; \text{Bernoulli}(\mu_i)$, i.e., the Binomial model with a single ($N=1$) trial.  Furthermore, we do not impose additional structure on $\gcal$, i.e., we assume that $ G \in \gcal = \pp([0,1])$. Under the above model, $\Zo_i$ is supported on $\cb{0,1}$ and we can take $\lambda = \delta_0 + \delta_1$ to be the counting measure on $\cb{0,1}$ and $p(\zo \cond \mu) = \mu^{\zo}(1-\mu)^{1-\zo}$. The marginal distribution $F_G$ is fully determined by $f_G(1)$, since $f_G(0) = 1 - f_G(1)$. In this case, the $F$-localizations we consider take the following simplified form. First, the DKW $F$-localization~\eqref{eq:DKW} is equal to:
\begin{equation}
\label{eq:dkw_bernoulli}
\ff^{\text{DKW}}_n(\alpha) =  \cb{F \in \pp(\cb{0,1}) \text{ with pmf }f: \, \abs{f(1) - \hat{f}_n(1)} \leq   \sqrt{\log\p{2/\alpha}\big/(2n)}}.
\end{equation}
Second, for the $\chi^2$-$F$-localization~\eqref{eq:floc_chisq}, write $\tau^2 = \chi^2_{1,1-\alpha}$, then:
\begin{equation}
\label{eq:chi_squared_bernoulli}
\begin{aligned}
\ff_n^{\chi^2}(\alpha) = \bigg\{&F \in \pp(\cb{0,1}) \text{ with pmf }f: \\
&\abs{f(1) - \frac{\hat{f}_n(1) \, + \, \tau^2/(2n)}{1 \, + \,\tau^2 /n} } \leq   \frac{\sqrt{\tau^2/n}}{1 \, + \,\tau^2/n} \cdot \sqrt{ \hat{f}_n(1)(1-\hat{f}_n(1)) + \tau^2/(4n)} \bigg\}.
\end{aligned}
\end{equation}
An important observation that we will use throughout the following proofs, is that any distribution $F \in \pp(\cb{0,1})$ can be represented as $F_{\tilde{\gprior}}$ for some $\tilde{\gprior} \in \pp([0,1])$ in model~\eqref{eq:EB} with the Bernoulli likelihood.

\subsubsection{Proof of Proposition~\ref{prop:bernoulli_partial}: Second moment}
\label{subsubsec:bernoulli_second_moment_proof}
\begin{proof}
We study the second moment of the prior. As already mentioned in the main text, this is an example of a linear functional that is partially identified. We discuss the partial identification aspect first. Suppose we know the marginal distribution of $\Zo_i$ exactly, that is, we know $f_{\gprior}(1)$. Notice that $\EE[\gprior]{\mu} = \int \mu \, d\gprior(\mu) = f_{\gprior}(1)$. Then, the partial identification interval for $L(\gprior)$ is the following:
\begin{equation}
\label{eq:bernoulli_marginal_density_partial_id}
 L(\gprior) \in \sqb{ \p{ \int \mu \, d\gprior(\mu)}^2, \;  \int \mu \, d\gprior(\mu)} \; = \; \sqb{ f_{\gprior}(1)^2,\; f_{\gprior}(1)}.
\end{equation}
For example,when $f_{\gprior}(1) = 1/2$, then $L(\gprior)  \in [1/4, 1/2]$. Why is the above the partial identification interval? First note that $\int \mu^2 \, d\gprior(\mu) \leq \int \mu \, d\gprior(\mu)$ holds since $\gprior$ is supported on $[0,1]$ and  $\int \mu^2 \, d\gprior(\mu) \geq (\int \mu \, d\gprior(\mu))^2$ holds by Jensen's inequality. Furthermore, there exist choices of $\gprior$ that make both inequalities tight. In particular, if $\gprior = \delta_{\bar{\mu}}$ for some $\bar{\mu}$ then $\int \mu \, d\gprior(\mu) = \bar{\mu}, \; \int \mu^2 \, d\gprior(\mu) = \bar{\mu}^2$, while for $\gprior = (1-\bar{\mu})\delta_0 + \bar{\mu}\delta_1$, it holds that $\int \mu \, d\gprior(\mu) = \int \mu^2 \, d\gprior(\mu) = \bar{\mu}.$

We seek to determine the (asymptotic) length of the different confidence intervals we consider in this work. We start with the $F$-localization approaches.\\

\noindent \textbf{$F$-localization:} By the above discussion on partial identification, we find that the $F$-localization intervals take the form $[ \inf_{G} \cb{f_{\gprior}(1)^2},\, \sup_{G} \cb{f_{\gprior}(1)} ]$ where the extrema are taken over all $\gprior$ such that $F_{\gprior} \in \ff_n(\alpha)$. We further restrict attention to the event wherein $\inf_{G} \cb{f_{\gprior}(1)} \in (0,1)$ and $\sup_{G} \cb{f_{\gprior}(1)} \in (0,1)$. Since $f_{\gprior}(1) \in (0,1)$ under the assumptions of Proposition~\ref{prop:applications}, this event will occur with asymptotic probability $1$ for both $F$-localizations. For the DKW-$F$-localization, in view of~\eqref{eq:dkw_bernoulli}, we then get the interval:
$$ \ii^{\text{DKW}} = \sqb{\p{ \hat{f}_n(1) \,-\,\sqrt{\log\p{2/\alpha}\big/(2n)}}^2,\; \hat{f}_n(1) \,+\,\sqrt{\log\p{2/\alpha}\big/(2n)}}.$$
Similarly, for the $\chi^2$-$F$-localization~\eqref{eq:chi_squared_bernoulli} we get the interval (with $\tau^2 = \chi^2_{1,1-\alpha}$):
$$
\begin{aligned}
\ii^{\chi^2} = \Bigg[&\p{  \frac{\hat{f}_n(1) \, + \, \tau^2/(2n)}{1 \, + \,\tau^2 /n} \,-\,\frac{\sqrt{\tau^2/n}}{1 \, + \,\tau^2/n} \cdot \sqrt{ \hat{f}_n(1)(1-\hat{f}_n(1)) + \tau^2/(4n)}}^2,\\
&\frac{\hat{f}_n(1) \, + \, \tau^2/(2n)}{1 \, + \,\tau^2 /n} \,+\,\frac{\sqrt{\tau^2/n}}{1 \, + \,\tau^2/n} \cdot \sqrt{ \hat{f}_n(1)(1-\hat{f}_n(1)) + \tau^2/(4n)}\;\;\;\;\,\Bigg].
\end{aligned}
$$
Since $\hat{f}_n(1) = f_{\gprior}(1) + o_{\mathbb P}(1)$, as $n \to \infty$, it follows for both $\ii = \ii^{\text{DKW}}$ and $\ii = \ii^{\chi^2}$, that:
$$\abs{\ii} =  f_{\gprior}(1) -  f_{\gprior}(1)^2 + o_{\mathbb P}(1) = f_{\gprior}(1)(1 -  f_{\gprior}(1)) + o_{\mathbb P}(1),$$
as claimed.\\

\noindent \textbf{AMARI:} The modulus problem~\eqref{eq:continuous_modulus_problem} at $\delta>0$ takes the form:
\begin{equation}
\label{eq:modulus_problem_bernoulli_supplement}
\underset{G_1, G_{-1} \in \gcal_n}{\text{maximize}}\;\; L(\gprior_1) - L(\gprior_{-1}) \;\text{ s.t. }\;  \abs{f_{\gprior_1}(1) - f_{\gprior_{-1}}(1)}  \leq  \frac{\delta}{\sqrt{n}}\cdot \sqrt{\bar{f}_n(1)(1-\bar{f}_n(1))}.
\end{equation}
By~\eqref{eq:bernoulli_marginal_density_partial_id}, it may be simplified as:
\begin{equation}
\label{eq:modulus_problem_bernoulli_second_moment_simplified}
\underset{G_1, G_{-1} \in \gcal_n}{\text{maximize}}\;\; f_{\gprior_1}(1) - f_{\gprior_{-1}}(1)^2 \;\text{ s.t. }\;  \abs{f_{\gprior_1}(1) - f_{\gprior_{-1}}(1)}  \leq  \frac{\delta}{\sqrt{n}}\cdot \sqrt{\bar{f}_n(1)(1-\bar{f}_n(1))}.
\end{equation}
We write $p = f_{\gpriorR}(1)$ and $f_{\gpriorL}(1) = p + \varepsilon$, for a choice of $\varepsilon \geq 0$ that we will make below. Then we seek to find (feasible choices of $p$, $\varepsilon$) so that $(p+\varepsilon) - p^2$ is maximized. We seek to solve this problem for some $\delta = \delta_n \geq \delta^{\ell} > 0$. Throughout the rest of the proof we assume that $f_{\gprior}(1) \in [1/2,\, 1)$; the case $f_{\gprior}(1)  \in (0, 1/2)$ being analogous. We define:
$$\ubar{p}_n = \max\cb{1/2,\, \inf\cb{ f_{\tilde{\gprior}}(1) \mid \tilde{\gprior} \in \gcal_n}}.$$
We note that $\ubar{p}_n = f_{\gprior}(1) + o_{\mathbb P}(1)$. Since $p \mapsto (p+\varepsilon) - p^2$ is decreasing in $p$ for $p \geq 1/2$, it follows that~\eqref{eq:modulus_problem_bernoulli_second_moment_simplified} is optimized for the choice $f_{\gprior_1}(1) = p = \ubar{p}_n$. Furthermore, $\varepsilon \mapsto (p+\varepsilon) - p^2$ is increasing in $\varepsilon$, and so it is maximized for the largest admissible value of $\varepsilon$. By the constraint in~\eqref{eq:modulus_problem_bernoulli_second_moment_simplified}, we see that 
$$\varepsilon \leq (\delta/\sqrt{n})\cdot \sqrt{\bar{f}_n(1)(1-\bar{f}_n(1))}.$$
The above constraint may be replaced by an equality, since the RHS above is of order $O_{\mathbb P}(1/\sqrt{n})$ and so, the distribution $F$ with $f(1)=\ubar{p}_n + \varepsilon$ would be included in both $F$-localizations~\eqref{eq:dkw_bernoulli} and \eqref{eq:chi_squared_bernoulli} for $n$ large enough. We conclude that
$$ \omega_n(\delta) \, = \, \ubar{p}_n(1-\ubar{p}_n) \, + \,  \frac{\delta}{\sqrt{n}}\cdot \sqrt{\bar{f}_n(1)(1-\bar{f}_n(1))},$$
which is differentiable in $\delta$ with:
$$\omega_n'(\delta) = \frac{1}{\sqrt{n}}\cdot \sqrt{\bar{f}_n(1)(1-\bar{f}_n(1))}.$$
Let us plug the above into~\eqref{eq:optimal_Q} to find the optimal $Q(\cdot)$. First, we consider the part of $Q$ that is a function of $z$.  For $z=1$ we get 
\begin{equation*}
\frac{n\cdot \omega_n'(\delta)}{\delta}\cdot \frac{f_{\gprior_1}(1) - f_{\gprior_{-1}}(1)}{\barf_n(1)} = \frac{\barf_n(1)(1-\barf_n(1))}{\barf_n(1)} = 1-\barf_n(1).
\end{equation*}
Similarly, for $z=0$, we get $-\barf_n(1)$. It remains to evaluate the additive component in~\eqref{eq:optimal_Q} that does not depend on $\zo$. Let $\gprior_0 = (\gprior_{1}+\gprior_{-1})/2$, where $\gprior_{-1},\gprior_{1}$ are any priors that have marginal pmf at $\zo=1$ equal to $\ubar{p}_n$, respectively $\ubar{p}_n + \varepsilon$. Then
$$  L(\gprior_0)  = \frac{1}{2}\cb{(\ubar{p}_n+\varepsilon) + \ubar{p}_n^2} = \frac{\ubar{p}_n(1+\ubar{p}_n)}{2} + \frac{\delta}{2\sqrt{n}}\cdot \sqrt{\barf_n(1)(1-\barf_n(1))}.$$
$$ f_{\gprior_0}(1) = \ubar{p}_n + \frac{\varepsilon}{2} = \ubar{p}_n + \frac{\delta}{2\sqrt{n}}\cdot \sqrt{\barf_n(1)(1-\barf_n(1))}.$$
Using the above two results, we find that the constant term is equal to:
\begin{equation*}
\begin{aligned}
&-\frac{n\cdot \omega_n'(\delta)}{\delta}\cdot \sum_{\zo=0}^1 \frac{\p{f_{\gprior_1}(\zo) - f_{\gprior_{-1}}(\zo)}f_{\gprior_{0}}(\zo)}{\barf_n(\zo)} \, + \, L(\gprior_0)  \\ 
= \; &  \barf_n(1) - f_{\gprior_{0}}(1) +  L(\gprior_0) \\
= \; & \barf_n(1) -\ubar{p}_n +\frac{\ubar{p}_n(1+\ubar{p}_n)}{2} \\
=\; & \barf_n(1) - \frac{\ubar{p}_n(1-\ubar{p}_n)}{2}.
\end{aligned}
\end{equation*}
Hence we now have an explicit expression for $Q(\zo)$ in~\eqref{eq:optimal_Q} for $\zo \in \cb{0,1}$:
$$ Q(\zo) \; = \; \ind(\zo = 1) \,-\, \frac{\ubar{p}_n(1-\ubar{p}_n)}{2}.$$
This means that $\hL = \hat{f}_n(1) \,-\, \frac{\ubar{p}_n(1-\ubar{p}_n)}{2}$, where $\hat{f}_n(1) = \#\cb{\Zo_i=1}/n$. The worst case absolute bias of $\hL$ is given by:
$$ \hB = \frac{1}{2}(\omega_n(\delta)-\delta \omega_n'(\delta)) = \frac{\ubar{p}_n(1-\ubar{p}_n)}{2}.$$
With $\hV$ as in~\eqref{eq:sample_var_est}, we finally get the confidence interval~\eqref{eq:im_iw_ci}:
$$\ii_\alpha = \hL \pm t_\alpha(\hB, \hV) = \hat{f}_n(1) \,-\, \frac{\ubar{p}_n(1-\ubar{p}_n)}{2}  \, \pm \,t_\alpha(\hB, \hV).$$
Under the given asymptotics $\hV = o_{\mathbb P}(1)$, $\hB = f_{\gprior}(1)(1-f_{\gprior}(1))/2 +  o_{\mathbb P}(1)$ and so it follows that for $\alpha \in (0,1)$, $t_\alpha(\hB, \hV) = f_{\gprior}(1)(1-f_{\gprior}(1))/2 +  o_{\mathbb P}(1)$. We conclude that the left endpoint of the confidence interval converges in probability to:
$$f_{\gprior}(1) \,- \, f_{\gprior}(1)(1-f_{\gprior}(1))/2 \, -\, f_{\gprior}(1)(1-f_{\gprior}(1))/2  = f_{\gprior}(1)^2.$$ 
The right point of the \Amari~confidence interval converges in probability to:
$$f_{\gprior}(1) \,- \, f_{\gprior}(1)(1-f_{\gprior}(1))/2 \, +\, f_{\gprior}(1)(1-f_{\gprior}(1))/2  = f_{\gprior}(1).$$ 
Hence we conclude that the length of the\Amari~confidence intervals converges to the length of the partial identification interval.
\end{proof}

\subsubsection{Proof of Proposition~\ref{prop:bernoulli_partial}: Posterior mean}

\begin{proof}
We now turn to study,
$$\theta_{\gprior}(1) = \EE[\gprior]{\mu \cond \Zo=1} = \frac{ \int \mu^2\, d\gprior(\mu)}{f_{\gprior}(1)}.$$
For a fixed value of the denominator $f_{\gprior}(\mu)(1)$, we derived partial identification intervals for the numerator in~\eqref{eq:bernoulli_marginal_density_partial_id}. It directly follows that the partial identification intervals for $\theta_{\gprior}(1)$ are equal to:
$$ \theta_{\gprior}(1)  \in \sqb{ f_{\gprior}(1) ,\; 1}.$$

\noindent \textbf{$F$-localization:} The argument now is very similar to that for the second moment. With the DKW-$F$-localization, in view of~\eqref{eq:dkw_bernoulli}, we get the interval
$$ \ii^{\text{DKW}}(1) = \sqb{\hat{f}_n(1) \,-\,\sqrt{\log\p{2/\alpha}\big/(2n)},\;\, 1} \bigcap \; [0,\,1].$$
Similarly, for the $\chi^2$-$F$-localization~\eqref{eq:chi_squared_bernoulli} we get the interval (with $\tau^2 = \chi^2_{1,1-\alpha}$):
$$
\ii^{\chi^2}(1) = \sqb{ \frac{\hat{f}_n(1) \, + \, \tau^2/(2n)}{1 \, + \,\tau^2 /n} \,-\,\frac{\sqrt{\tau^2/n}}{1 \, + \,\tau^2/n} \cdot \sqrt{ \hat{f}_n(1)(1-\hat{f}_n(1)) + \tau^2/(4n)},\;\,1} \bigcap \; [0,\,1].
$$
Since $\hat{f}_n(1) = f_{\gprior}(1) + o_{\mathbb P}(1)$, as $n \to \infty$, it follows for both $\ii(1) = \ii^{\text{DKW}}(1)$ and $\ii(1) = \ii^{\chi^2}(1)$, that:
$$\abs{\ii(1)} =  1 - f_{\gprior}(1) + o_{\mathbb P}(1).$$
\\
\noindent \textbf{\Amari:} For fixed $c \in [0,1]$, we start by studying the modulus problem~\eqref{eq:modulus_problem_bernoulli_supplement} for the linear functional
$$L(\gprior) = \theta_{\gprior}^{\text{lin}}(\zo; c) = \int \mu^2 \,dG(\mu) \,-\, c \int \mu \, dG(\mu).$$
Following precisely the derivation in the proof for \Amari~in Supplement~\ref{subsubsec:bernoulli_second_moment_proof} and the notation used therein, we find that on an event with asymptotic probability $1$, it holds that the optimal $Q^c(\cdot)$~\eqref{eq:optimal_Q} for all $c\in [0,1]$ takes the form:
$$ Q^c(\zo) \; = \; (1-c)\ind(\zo = 1) \,-\, \frac{\ubar{p}_n(1-\ubar{p}_n)}{2},$$
with worst-case bias $\hB=\ubar{p}_n(1-\ubar{p}_n)/2$ and $\hV =  o_{\mathbb P}(1)$. Write $\widetilde{\ii}_{\alpha}(\zo; c)$ for the confidence interval for $\theta_{\gprior}^{\text{lin}}(\zo; c)$ and $\htheta_{-}^{\text{lin}}(\zo; c)$, resp.  $\htheta_{+}^{\text{lin}}(\zo; c)$ for its left and right endpoints. All $o_{\mathbb P}(1)$ terms above are uniform with respect to $c$, and so arguing again as in Supplement~\ref{subsubsec:bernoulli_second_moment_proof}, it follows that:
$$\sup_{c \in [0,1]} \abs{ \htheta_{-}^{\text{lin}}(\zo; c) \, -  \, f_{\gprior}(1)(f_{\gprior}(1) - c)} = o_{\mathbb P}(1),\;\; \sup_{c \in [0,1]} \abs{ \htheta_{+}^{\text{lin}}(\zo; c) \, -  \, f_{\gprior}(1)(1 - c)} = o_{\mathbb P}(1).$$
Recall that in Algorithm~\ref{alg:amari_ci} we seek to find all $c \in [c^{\ell},\, c^u]$ such that $0 \in \widetilde{\ii}_{\alpha}(\zo; c)$, where $[c^{\ell},\, c^u] \subset [0, 1]$ is the pilot interval. Fix $\zeta >0$ small, then by the above uniform convergence, we have that:
$$ \PP{ 0 \in \widetilde{\ii}_{\alpha}(\zo; c) \text{ for any } 0 \leq c \leq f_{\gprior}(1) - \zeta} \to 0 \text { as n } \to \infty,$$
and that:
$$ \PP{ 0 \in \widetilde{\ii}_{\alpha}(\zo; c) \text{ for all } 1 \geq c \geq f_{\gprior}(1) + \zeta} \to 1 \text { as n } \to \infty.$$
Since $\zeta >0$ was arbitrary, we thus we find that the left-most endpoint of $\ii_\alpha(\zo)$ converges in probability to $f_{\gprior}(1)$ and the right-most endpoint converges to $1$, i.e., the asymptotic confidence interval length is equal to $1-f_{\gprior}(1)$.

\end{proof}

\section{Computational aspects for $F$-localization}
\label{sec:floc_computation}

\subsection{Parametric convex programming for $F$-localization intervals}
\label{sec:worst_case_convex}

We explain how to compute $\htheta^+_\alpha(\zo)$ in~\eqref{eq:nbhood_worst_case} (the steps for $\htheta^-_\alpha(\zo)$ being analogous) when $\gcal$ and $\ff_n$ are convex, but not necessarily representable through linear constraints. Recall that $\theta_{\gprior}(\zo) = a_{\gprior}(\zo)/f_{\gprior}(\zo)$.  We first compute confidence intervals for $f_{\gprior}(\zo)$ using the same $F$-localization, i.e.,
$$\hf^-_\alpha(\zo) = \inf \cb{ f_{\gprior}(\zo) \mid \gprior \in \gcal\p{\ff_n(\alpha)}},\, \hf^+_\alpha(\zo) = \sup \cb{ f_{\gprior}(\zo) \mid \gprior \in \gcal\p{\ff_n(\alpha)}}.$$
The objective here is linear and the constraints are convex, and so the above is a convex programming problem. We then observe that
$$
\begin{aligned}
\htheta^+_\alpha(\zo) &= \sup \cb{ \theta_{\gprior}(\zo) \mid \gprior \in \gcal\p{\ff_n(\alpha)}} \\
        &= \sup \sqb{ \sup \cb{ a_{\gprior}(\zo)/t \mid \gprior \in \gcal\p{\ff_n(\alpha)}, \, f_{\gprior}(\zo)=t} \mid \; t \in [\hf^-_\alpha(\zo), \hf^+_\alpha(\zo)]}.
\end{aligned}
$$
Hence, we can proceed as follows. Let $\mathcal{T}$ be a fine discretization of $[\hf^-_\alpha(\zo), \hf^+_\alpha(\zo)]$, say with $100$ equidistant points. Then, for each $t \in \mathcal{T}$, compute:
$$\htheta^+_\alpha(\zo; t) = \sup \cb{ a_{\gprior}(\zo)/t \mid \gprior \in \gcal\p{\ff_n(\alpha)}, \, f_{\gprior}(\zo)=t}.$$
Note that this problem also has an objective that is linear in $G$, and specifies convex constraints on $\gprior$, i.e., it is a convex programming problem. Finally, we report $\htheta^+_\alpha(\zo) = \max_{t \in \mathcal{T}} \htheta^+_\alpha(\zo; t)$.

\subsection{Considerations for discretization of $\gcal$}
\label{sec:gcal_discretization}

If an infinite-dimensional $\gcal$ is specified, then it is important to guarantee that the error incurred when solving~\eqref{eq:nbhood_worst_case} or~\eqref{eq:continuous_modulus_problem} with a discretized class \smash{$\widetilde{\gcal}$} instead of $\gcal$, is negligible compared to e.g., the confidence interval width. This typically requires $\gcal$ to be tight, and in our applications we used the prior classes~\eqref{eq:all_dbns},~\eqref{eq:normal_mixing_class} and \eqref{eq:normal_scale_class} with $\Ksupport$ chosen as a compact set. 

The following proposition can be used to verify the accuracy of a discretization \smash{$\widetilde{\gcal}$}.
\begin{prop}
Consider the linear functional $L(\gprior) = \int \psi(\mu)\dgprior$ for a function $\psi(\cdot)$.
\begin{enumerate}[leftmargin=*, label=\alph*)]
\item If $\abs{\psi(\mu)} \leq C_{\psi}$, then, $\inf_{\tilde{\gprior} \in \widetilde{\gcal}}\abs{L(G) - L(\tilde{G})} \leq C_{\psi}/2\inf_{\tilde{\gprior} \in \widetilde{\gcal}} \{\TV(\gprior,\tilde{\gprior})\}$, where \\ 
\smash{$\TV(\gprior, \tilde{\gprior}) = \sup_{A}\lvert \gprior(A)-\tilde{\gprior}(A)\rvert$} is the total variation distance between $\gprior$ and $\tilde{\gprior}$.
\item  If $\psi(\mu)$ is $C_{\psi}$-Lipschitz continuous, then, $\inf_{\tilde{\gprior} \in \widetilde{\gcal}}\abs{L(G) - L(\tilde{G})} \leq C_{\psi}\inf_{\tilde{\gprior} \in \widetilde{\gcal}} \{ W_1(\gprior,\tilde{\gprior})\},$\\
with \smash{$W_1(G, \tilde{G}) = \inf\{ \EEInline{\abs{\mu - \tilde{\mu}}}:\; (\mu, \tilde{\mu}) \text{ random variables s.t. } \mu\sim G, \tilde{\mu} \sim \tilde{G}\}$} the Wasserstein distance between $G$ and $\tilde{G}$ (cf. \citetsupplement{panaretos2019statistical} and references therein).
\end{enumerate}
\label{prop:discretization}
\end{prop}
\begin{proof}
\textbf{a)}\; Recall that \smash{$\TV(G,\tilde{\gprior}) = \frac{1}{2}\int |dG(\mu)-d\tilde{G}(\mu)|$}. Thus,
$$ \abs{L(\gprior)- L(\tilde{\gprior)}} = \abs{\int \psi(\mu) \p{dG(\mu)-d\tilde{G}(\mu)}} \leq C_{\psi} \int \abs{dG(\mu)-d\tilde{G}(\mu)} \leq 2C_{\psi} \TV(G,\tilde{\gprior}).$$
\textbf{b)}\; Letting $\mu \sim G, \tilde{\mu} \sim \tilde{G}$, the optimal Wasserstein coupling, we get
$$ \abs{L(\gprior)- L(\tilde{\gprior)}} = \abs{\EE{\psi(\mu)} - \EE{\psi(\tilde{\mu})}} \leq \EE{\abs{\psi(\mu) - \psi(\tilde{\mu})}} \leq C_{\psi} \EE{\abs{\mu - \tilde{\mu}}} = C_{\psi} W_1(G,\tilde{G}).$$
\end{proof}

For example, when part b) of the Proposition is applicable, then it suffices for $\widetilde{\gcal}$ to be a cover of $\gcal$ in terms of the Wasserstein distance. In some cases, part b) is not applicable. For example, when constructing intervals for the local false sign rate in the standard Gaussian empirical Bayes problem, then the numerator $a_{\gprior}(\zo)$ in~\eqref{eq:ratio} takes the form $a_{\gprior}(z) = \int \psi(\mu)\dgprior$ with $\psi(\mu) = \ind( \mu \geq 0) \varphi(z-\mu)$, and so $\psi(\cdot)$ is not Lipschitz continuous. Instead, part a) of the Proposition is applicable, and so a cover in total variation suffices.

Below we provide details for the discretization of $\mathcal{P}(\mathcal{K})$~\eqref{eq:all_dbns} and $\law\nn(\tau, \mathcal{K})$~\eqref{eq:normal_mixing_class}, when $\mathcal{K}$ is a compact interval.

\subsubsection{Discretization of compactly supported distributions} 
Consider $\mathcal{P}([L,U])$~\eqref{eq:all_dbns}, the class of all distributions supported on the compact interval $\mathcal{K}=[L,U]$. We first discretize $\mathcal{K}=[L,U]$ as the finite grid
\begin{equation}
\label{eq:finite_grid}
\mathcal{K}(\Bdisc, L, U) = \cb{ L,\, L + \frac{U-L}{\Bdisc},\, L + 2\frac{U-L}{\Bdisc},\,\dotsc,\, U},\;\; \Bdisc  \in \mathbb N.
\end{equation}
Then $\mathcal{P}([L,U])$ may be discretized by considering $\mathcal{P}(\mathcal{K}(\Bdisc,L,U))$, the class of all distributions supported on the grid $\mathcal{K}(\Bdisc, L, U)$. This class is amenable to our optimization tasks. By enumerating the grid elements as $\mathcal{K}(\Bdisc,L, U) = \cb{\mu_1, \dotsc, \mu_{\Bdisc+1}}$, we may represent every $\gprior \in \mathcal{P}(\mathcal{K}(\Bdisc,L,U))$ by the probabilities \smash{$\pi_j = \PP[\gprior]{\mu = \mu_j}$} assigned to $\mu_j$, and so we may identify $\mathcal{P}(\mathcal{K}(\Bdisc,L,U))$ with the probability simplex:
\begin{equation}
\label{eq:prob_simplex}
S^{\Bdisc+1} = \cb{ (\pi_1, \dotsc, \pi_{\Bdisc+1}) \in [0,1]^{\Bdisc+1} \; \cond \; \sum_{j=1}^{\Bdisc + 1 } \pi_j = 1}.
\end{equation}
$S^{\Bdisc+1}$ is a linear polytope. See~\eqref{eq:dkw_loc_opt} for an explicit example of how it is used for computations. In Section~\ref{subsec:lord_dataset} we discretized $\mathcal{P}([0,1])$ as above with $\Bdisc + 1= 300$.

Finally, we note that the discretization $\mathcal{P}(\mathcal{K}(\Bdisc,L,U))$ provides a $O((U-L)/\Bdisc)$ covering of $\mathcal{P}([L,U])$ in Wasserstein distance, but not in total variation distance. For this discretization scheme, Proposition~\ref{prop:discretization} justifies inference for functionals satisfying b) of its statement. The implication for inference on empirical Bayes estimands is that we may use the above discretization to conduct inference for the posterior mean, e.g., in the Gaussian empirical Bayes problem.

\subsubsection{Discretization of Gaussian location mixtures}

Consider the class of distributions \smash{$\law\nn\p{\tau,\, \mathcal{K}}$} from~\eqref{eq:normal_mixing_class} in the special case where \smash{$\mathcal{K} = [L,U]$} is a compact interval. Then letting \smash{$\mathcal{K}(\Bdisc, L, U)$}~\eqref{eq:finite_grid} the equidistant discretization of $[L,U]$, we use \smash{$\law\nn\p{\tau, \mathcal{K}(\Bdisc, L, U)}$} as a discretization of \smash{$\law\nn\p{\tau,\, \mathcal{K}}$} that is amenable to efficient computation. We have the following numerical representation of \smash{$\law\nn\p{\tau, \mathcal{K}(\Bdisc, L,U)}$}:
\begin{equation}
\label{eq:discretized_gaussian}
\gprior \in \law\nn\p{\tau, \mathcal{K}(\Bdisc, L,U)} \, \Longleftrightarrow \;\, \gprior = \sum_{j =1}^{\Bdisc+1} \pi_j \nn(\mu_j, \tau^2),\,\,\, (\pi_1, \dotsc, \pi_{\Bdisc + 1}) \in S^{\Bdisc +1}.  
\end{equation}
Here~\eqref{eq:discretized_gaussian} refers to the probability simplex~\eqref{eq:prob_simplex} and so $\law\nn\p{\tau, \mathcal{K}(\Bdisc, L,U)}$ can also be represented in term of $S^{\Bdisc +1}$. Furthermore in the Gaussian empirical Bayes problem~\eqref{eq:EB} with $\Zo \cond \mu \sim \nn(\mu, \sigma^2)$, and $G$ discretized as in~\eqref{eq:discretized_gaussian} we typically do not need to resort to numerical quadrature. To see this, note that for $\gprior \in  \law\nn\p{\tau, \mathcal{K}(\Bdisc, L,U)}$:
$$ L(\gprior) = \sum_{j} \pi_j L( \nn(\mu_j, \tau^2)), $$
and for many functionals of interest there exist explicit expressions for $L( \nn(\mu_j, \tau^2))$. A special case of the above result is marginalization:
$$ \Zo \, \sim \, \sum_{j} \pi_j \nn\p{\mu_j, \tau^2 + \sigma^2}.$$

Finally, we note that by a direct calculation it follows that \smash{$\law\nn\p{\tau, \mathcal{K}(\Bdisc, L,U)}$} provides a \smash{$O((U-L)/\Bdisc)$} covering of \smash{$\law\nn\p{\tau, [L,U]}$} in both total variation and Wasserstein distance.  Inference based on the discretized class will thus be valid for linear functionals satisfying a) or b) of Proposition~\ref{prop:discretization} as long as $\Bdisc$ is large enough.

\section{Computational aspects for \Amari}
\label{sec:amari_computation}

\subsection{Discretization}
\label{subsec:amari_discretization}

\paragraph{Discretization of $\gcal$:} Here the same considerations apply as in Supplement~\ref{sec:gcal_discretization}.

\paragraph{Discretization of $\Zo$:} Computing~\eqref{eq:continuous_modulus_problem} for a continuous likelihood, such as \smash{$\nn(\mu, \sigma^2)$}, requires numerical integration, e.g., to compute \smash{$\int  (f^M_{\gprior^{\delta}_1}(\zo) - f^M_{\gprior^{\delta}_{-1}}(\zo))^2\big/\barf_n^M(\zo)\ d\lambda^M(\zo)$}. In this section we explain how to conduct the discretization rigorously. Our goal is to allow an arbitrary discretization (that may be coarse) by accounting for the discretization in the calculation of the worst-case bias. In particular, even if the discretization is too coarse, our intervals will have correct coverage, although they may be overly wide.

Fix $M > 0$ as in~\eqref{eq:tilde_Z}, and consider the grid,\footnote{In practice, to guarantee shorter confidence intervals, the grid should be made as dense as possible, subject to computational constraints, and should also become denser as $n$ increases. In the Gaussian problem for example, we discretize \smash{$[-M,M]$} as a dense equidistant grid, with step size \smash{$\ll \sigma$}. While we do not pursue this further here, existing theory for discretization in statistical inverse problems~\citepsupplement{johnstone1991discretization} suggests that even relatively coarse griding suffices to maintain the minimax risk. The result of Theorem~\ref{theo:lin_functional_clt} allows for an arbitrary discretization of \smash{$[-M,M]$} (by applying the results of that theorem to the `discretized' likelihood) that can also change with $n$.}
\begin{equation}
\label{eq:fine_grid}
\ii = \ii_n = \cb{-M = t_{1,n} < t_{2,n} < \dotsc < t_{K_n-1,n} = M},
\end{equation}
where the grid may depend on $n$. Also let us define $t_{0,n}=-\infty$, $t_{K_n,n} =+\infty$ and $I_{k,n} = [t_{k-1,n}, t_{k,n})$ for $k \in \{1,\dotsc, K_n\}$ and $K_n\in \NN$. In analogy to~\eqref{eq:tilde_Z}, we define:
\begin{equation}
\Zo_i^{\ii}  := \sum_{k=1}^{K_n} k\ind(\Zo_i \in I_{k,n}) \in \{1,\dotsc,K_n\}.
\end{equation}
Also let $\lambda^{\ii}$ the counting measure on $\cb{1,\dotsc,K_n}$ and analogously to the development after~\eqref{eq:tilde_Z}, define  $f_{\gprior}^{\ii}$ to be the marginal density of $\Zo_i^{\ii}$ with respect to $\lambda^{\ii}$, i.e., $f_{\gprior}^{\ii}(k) = \int f_{\gprior}(\zo) \ind(\zo \in I_{k,n})d\lambda(\zo)$ for $k \in \cb{1,\dotsc,K_n}$ and also define $\bar{f}^{\ii}(k)$ analogously.

In view of the above considerations, the modulus problem~\eqref{eq:continuous_modulus_problem} takes on the following discrete form:
\begin{equation}
\label{eq:discrete_modulus_problem}
 \sup\cb{ L(\gprior_1) - L(\gprior_{-1})\;\mid\; \gprior_1, \gprior_{-1} \in \gcal_n,\; n\cdot \sum_{k=1}^{K_n} \frac{(f^{\ii}_{\gprior_1}(k) - f^{\ii}_{\gprior_{-1}}(k))^2}{\barf^{\ii}(k)} \; \leq \; \delta^2}.
\end{equation}
Below we discuss the solution of this discretized form of the modulus.

\subsection{Computing the affine minimax estimator}
\label{subsec:computation_minimaxaffine}

\subsubsection{Direct form of the modulus problem}

To solve~\eqref{eq:discrete_modulus_problem} with modern convex optimization solvers, it is convenient to represent it as follows
\begin{subequations}
\label{eq:modulus_problem_socp}
\begin{align}
\sup_{\gprior_1, \gprior_{-1}} \quad & L(\gprior_1) - L(\gprior_{-1}) \label{eq:modulus_max_linear_diff}\\
\textrm{s.t.} \quad & \sqrt{\sum_{k=1}^{K_m} \frac{(f^{\ii}_{\gprior_1}(k) - f^{\ii}_{\gprior_{-1}}(k))^2}{\barf^{\ii}(k)}} \; \leq \; \frac{\delta}{\sqrt{n}}  \label{eq:modulus_modulus}\\
  &\gprior_1, \gprior_{-1} \in \gcal  \label{eq:modulus_priors}  \\ 
  &\gprior_1, \gprior_{-1} \textrm { are F-localized, i.e., } F_{\gprior_1}, F_{\gprior_{-1}} \in \ff_n  \label{eq:modulus_localization}
\end{align}
\end{subequations}

We make the following observations:
\begin{itemize}[label={--},noitemsep]
	\item The optimization variables are $\gprior_1, \gprior_2 \in \gcal$. With $\gcal$ suitably discretized, as in Section~\ref{sec:gcal_discretization}, these have finite-dimensional representations. The choices of (discretized) $\gcal$ considered in this work may be represented using a finite number of linear constraints.
	\item The objective~\eqref{eq:modulus_max_linear_diff} is linear in the optimization variables.
	\item The maps $\gprior_{\ell} \mapsto f^{\ii}_{\gprior_{\ell}}(k)$ are linear in $\gprior_{\ell}$ and so~\eqref{eq:modulus_modulus} corresponds to a second order cone constraint.
	\item The localization constraints in~\eqref{eq:modulus_localization} may be implemented as a finite number of constraints on $\gprior_1$, resp. $\gprior_{-1}$. These can be either linear (DKW and Gauss-$F$-localizations) or quadratic ($\chi^2$-$F$-localization) in the optimization variables. To see these two claims, first note that the maps $\gprior_{\ell} \mapsto F_{\gprior_{\ell}}$ are also linear. Furthermore, inspecting the proof of Theorem~\ref{theo:lin_functional_clt}, we see that~\eqref{eq:pilot_and_localization_quality} only needs to hold for the discretized distributions, $F^{\ii}_{\gprior_1}, F^{\ii}_{\gprior_{-1}},F^{\ii}_{\gprior}, \bar{F}^{\ii}$ (with triangular array asymptotics accounting for $\ii$ changing with $n$).
\end{itemize}
As a consequence of the above observations, the discretized modulus problem may be represented as a finite-dimensional second order conic program (SOCP)~\citep{boyd2004convex}, which in turn is efficiently solvable by modern convex optimization solvers such as Mosek~\citep{mosek} or Hypatia~\citep{coey2020towards}. In our numerical examples we use Mosek; our implementation can also use Hypatia.

\subsubsection{Superdifferential of the modulus problem and duality}
Evaluation of the estimator~\eqref{eq:optimal_Q} requires access to $\omega_n'(\delta)$, an element of the superdifferential of the modulus of continuity $\omega_n(\delta)$ at $\delta = \delta_n$. An element $\omega_n'(\delta) \cdot \sqrt{n}$ may be directly extracted upon solving~\eqref{eq:modulus_problem_socp} as the dual variable associated to the constraint~\eqref{eq:modulus_modulus}, provided that strong duality holds, and the primal and dual optima are attained. Many convex solvers, including Mosek and Hypatia, return the dual variables.

\begin{proof}[Argument sketch]
Define the Lagrangian of~\eqref{eq:modulus_problem_socp} for $\lambda \geq 0$:
$$ \mathcal{L}(\gprior_1, \gprior_{-1}, \lambda; \delta) =  L(\gprior_1) - L(\gprior_{-1}) - \lambda\sqb{\p{\sum_{k=1}^{K_m} \frac{(f^{\ii}_{\gprior_1}(k) - f^{\ii}_{\gprior_{-1}}(k))^2}{\barf^{\ii}(k)}}^{1/2} - \frac{\delta}{\sqrt{n}}}.$$
Note that we parametrize the optimization problem and also the Lagrangian by $\delta$. For any feasible $\gprior_1, \gprior_{-1}$ and $\lambda \geq 0$
\begin{equation}
\label{eq:duality_basic}
\mathcal{L}(\gprior_1, \gprior_{-1}, \lambda; \delta) \geq L(\gprior_1) - L(\gprior_{-1}).
\end{equation}
Let $G_{1}^{\delta}, G_{-1}^\delta$ be primal optimal solutions to~\eqref{eq:modulus_problem_socp} and let $\lambda^{\delta}$ be the optimal dual variable, then~\citep[Chapter 5.5.2]{boyd2004convex}:
\begin{equation}
\label{eq:strong_duality}
\omega_n(\delta) = L(\gprior_1^{\delta}) - L(\gprior_{-1}^{\delta}) = \mathcal{L}(\gprior_1^{\delta}, \gprior_{-1}^{\delta}, \lambda^{\delta}; \delta) =  \sup_{\gprior_1, \gprior_{-1}} \mathcal{L}(\gprior_1, \gprior_{-1}, \lambda^{\delta}; \delta)
\end{equation}
Now fix $\delta>0$ and take any $\Delta\delta$ such that $\Delta\delta > -\delta$. Also let $\gprior_1^{\delta + \Delta\delta}, \gprior_{-1}^{\delta + \Delta\delta}$ solutions to~\eqref{eq:modulus_problem_socp} at $\delta + \Delta\delta$. Putting all results together
$$
\begin{aligned}
\omega_n(\delta + \Delta\delta) &= L(\gprior_1^{\delta + \Delta\delta}) - L(\gprior_{-1}^{\delta + \Delta\delta}) \\
&\stackrel{\eqref{eq:duality_basic}}{\leq} \mathcal{L}(\gprior_1^{\delta + \Delta\delta}, \gprior_{-1}^{\delta + \Delta\delta}, \lambda^{\delta}; \delta + \Delta\delta) \ \\
&= L(\gprior_1^{\delta + \Delta\delta}) - L(\gprior_{-1}^{\delta + \Delta\delta}) - \lambda^{\delta}\sqb{\p{\sum_{k=1}^{K_m} \frac{(f^{\ii}_{\gprior_1^{\delta + \Delta\delta}}(k) - f^{\ii}_{\gprior_{-1}^{\delta + \Delta\delta}}(k))^2}{\barf^{\ii}(k)}}^{1/2} - \frac{\delta+ \Delta\delta}{\sqrt{n}}} \\
&= L(\gprior_1^{\delta + \Delta\delta}) - L(\gprior_{-1}^{\delta + \Delta\delta}) - \lambda^{\delta}\sqb{\p{\sum_{k=1}^{K_m} \frac{(f^{\ii}_{\gprior_1^{\delta + \Delta\delta}}(k) - f^{\ii}_{\gprior_{-1}^{\delta + \Delta\delta}}(k))^2}{\barf^{\ii}(k)}}^{1/2} - \frac{\delta}{\sqrt{n}}} + \frac{\lambda^{\delta}\Delta\delta}{\sqrt{n}}\\
&=  \mathcal{L}(\gprior_1^{\delta + \Delta\delta}, \gprior_{-1}^{\delta + \Delta\delta}, \lambda^{\delta}; \delta)  + (\lambda^{\delta}/\sqrt{n})\Delta\delta\ \\ 
&\stackrel{\eqref{eq:strong_duality}}{\leq} \omega_n(\delta) + (\lambda^{\delta}/\sqrt{n})\Delta\delta
\end{aligned}
$$
Thus $\lambda^{\delta}/\sqrt{n} \in \partial\omega_n(\delta)$, that is, $\lambda^{\delta}/\sqrt{n}$ is an element of the superdifferential of $\omega_n(\cdot)$ at $\delta$.
\end{proof}

\subsection{Bias-aware Normal confidence interval}
\label{subsec:bias_aware_formula}
Recall that for constructing the confidence intervals from~\eqref{eq:im_iw_ci}, we need to calculate (with $W \sim \mathcal{N}(0,1)$):
$$ t_\alpha(B,V) = \inf\cb{t : \PP{\abs{B + V^{1/2}W} \leq t} \geq 1 - \alpha \text{ for all } \abs{b} \leq B}$$
This is the same as:
$$ t_\alpha(B,V) = V^{1/2} \inf\cb{t : \PP{\abs{B/V^{1/2} + W} \leq t} \geq 1 - \alpha \text{ for all } \abs{b} \leq B}$$
It is not directly obvious how to calculate this, however here we will argue that the calculation reduces to calculating the quantile of the absolute value of a Normal distribution (and hence can be efficiently computed); this expression is also given in~\citet{armstrong2018optimal}:
\begin{prop}
\label{prop:bias_aware_formula}
Under the above setting it holds that:
$$ t_{\alpha}(B,V) = V^{1/2}\mathop{cv_{\alpha}}(B/V^{1/2} )$$
Here $\mathop{cv_{\alpha}}(u)$ is the $1-\alpha$ quantile of the absolute value of a $\mathcal{N}(u,1)$ distribution.
\end{prop}

\begin{proof}
For convenience of notation and without any loss of generality, let us assume $V=1$. First let us note that $\abs{b + W} \stackrel{\mathcal{D}}{=} \abs{-b + W}$ for any $b$, hence:
$$ t_\alpha(B,V) = \inf\cb{t : \PP{\abs{b + W} \leq t} \geq 1 - \alpha \text{ for all } 0 \leq b \leq B}$$
Next, observe that for $b=B$, \smash{$\abs{B + W} \sim |\mathcal{N}(B,1)|$}, and thus by definition:
$$ \inf\cb{t : \PP{\abs{B + W} \leq t} \geq 1 - \alpha} = \mathop{cv_{\alpha}}(B)$$
We now just need to check what happens for $0 \leq b \leq B$, and indeed we will need some stochastic dominance argument. It suffices to argue that for any fixed $t>0$ and $0 \leq b \leq B$:

$$ \PP{ |B+W| \leq t} \leq  \PP{ |b+W| \leq t }$$
Thus, if we let $h(b) = \PP{ |b+W| \leq t }$ it suffices to show $h'(b)  \leq 0 \text{ for all } b\geq 0$, so that it is decreasing. A direct calculation yields (with $\Phi, \varphi$ the standard Normal CDF and pdf respectively):
$$h(b) = \Phi(t-b) - \Phi(-t-b).$$
So:
$$h'(b)= -\varphi(t-b) + \varphi(-t-b) \leq 0.$$
The last inequality holds since $\abs{t-b} \leq \abs{-t-b}$ for $t,b\geq0$.
\end{proof}

\section{Exponential family (logspline) G-modeling}
\label{sec:exp_family}
In this section we summarize the empirical Bayes approach introduced by~\citet{efron2016empirical} and \citet{narasimhan2016g}. The key idea is to specify $\mathcal{H}$ as a flexible exponential family of effect size distributions with natural parameters \smash{$\alpha = (\alpha_1,\dotsc,\alpha_p)$}, sufficient statistic \smash{$Q(\mu): \RR \mapsto \RR^{p}$} and base measure $H$. Concretely, distributions \smash{$G \in \mathcal{H}$} are parametrized by $\alpha$ with Radon-Nikodym derivative \smash{$g_{\alpha}(\mu) = dG/dH(\mu)$} defined as
\begin{equation}
    \label{eq:efron_exp}
     g_{\alpha}(\mu) =\exp(Q(\mu)^\top \alpha - A(\alpha)).
\end{equation}
$A(\alpha)$ is such that $\int g_{\alpha}(\mu)dH(\mu) = 1$. It is worth pointing out, that in contrast to our setting, $\mathcal{H}$ is not a convex class. $\alpha$ is estimated by $\hat{\alpha}$, the maximizer of the log (marginal) likelihood $\ell(\alpha)$ in model~\eqref{eq:EB}:
\begin{equation}
\ell(\alpha) = \sum_{i=1}^n\log\p{ \int p(\Zo_i \mid \mu)  g_{\alpha}(\mu)dH(\mu)}.
\end{equation}
\citet{efron2016empirical} further recommends to maximize the penalized likelihood $\ell(\alpha)-s(\alpha)$ instead, where \smash{$s(\alpha) = c_0 ||\alpha||_2$ for some $c_0 > 0$}. The empirical Bayes quantity \smash{$\theta_{\gprior}(\zo) = \EE[\gprior]{h(\mu_i) \cond \Zo_i=\zo}$} can then be estimated by the plug-in estimator \smash{$\hat{\theta}(\zo) = \EE[\hG]{h(\mu_i) \cond \Zo_i=\zo}$}, where \smash{$\hG$} is the prior with $dH$-density \smash{$g_{\hat{\alpha}}(\cdot)$}. Standard delta method calculations and maximum-likelihood asymptotics can then be used to estimate standard errors and correct bias due to the penalization (but not due to misspecification). \citet{efron2016empirical} demonstrates that even under misspecification, such a family of effect size distributions leads to practical (albeit biased) empirical Bayes point estimates.

We use the following parameters for the method in our numerical results.

\begin{itemize}[label={--}, leftmargin=*, nosep]
\item We take the base measure $H$ to be the uniform measure $U[-4, 4]$, the sufficient statistic to be a natural spline with 5 degrees of freedom with equidistant knots on the above grid and $c_0 =0.001$.
\item In Figure~\ref{fig:gmodel_varying_dof} we use the same settings as above but vary the degrees of freedom from $2$ to $12$.
\end{itemize}

\section{Sensitivity analysis for the prostate dataset}
\label{sec:supplement_sensitivity}

In this Supplement, we explore some of the issues raised in Section~\ref{subsec:sensitivity} by revisiting the prostate data analysis of Section~\ref{subsec:prostate}. There we we posited that \smash{$\gprior \in \gcal = \law\nn(\tau^2, [-3.3])$}~\eqref{eq:normal_mixing_class} with $\tau=0.25$. The first question we ask, is whether a goodness-of-fit test can guide the choice of $\tau$ in a data-driven way. To test \smash{$H_0(\tau): \gprior \in  \law\nn(\tau^2, [-3.3])$}, we use the Split Likelihood-Ratio (SLR) test of~\citet*{wasserman2020universal}, of which we provide a brief explanation. 

First, we randomly split our observations into two folds, $I_0$ and $I_1$. Then, let \smash{$\widehat{G}_1$} be the nonparametric maximum likelihood estimator of $\gprior$ in the class \smash{$\mathcal{P}(\RR)$} using $\Zo_i,\, i \in I_1$ and \smash{$\widehat{G}_0$} the nonparametric maximum likelihood estimator of $\gprior$ in the class \smash{$\law\nn(\tau^2, [-3.3])$} using $\Zo_i,\, i \in I_0$. The Split Likelihood-Ratio (SLR) is defined as:
$$ \text{SLR} = \prod_{i \in I_0} \frac{ f_{\widehat{G}_1}(\Zo_i)}{ f_{\widehat{G}_0}(\Zo_i)}.$$
\citet*{wasserman2020universal} prove that the test \smash{$\cb{\text{SLR} > 1/\alpha}$} is a finite-sample valid level $\alpha$ test for the null hypothesis \smash{$H_0(\tau)$}. The second column of Table~\ref{tab:sensitivity} shows the SLR for \smash{$\tau \in \cb{0.02, 0.1, 0.25, 0.5, 55}$}. The SLR test at level $\alpha=0.05$ only rejects the model with $\tau=0.55$, and the SLR statistic becomes smaller as $\tau$ decreases. 

Next, we consider inference for the local false sign rate \smash{$\PP[\gprior]{\mu \geq 0 \cond \Zo=2}$} and posterior mean \smash{$\EE[\gprior]{\mu \cond \Zo=2}$} at $\Zo=2$ using the Gauss-$F$-Localization approach. The last two columns of Table~\ref{tab:sensitivity} show the confidence intervals for each choice of $\tau$ (that was not rejected by the SLR test). We observe that as $\tau$ becomes smaller, the confidence intervals for  \smash{$\PP[\gprior]{\mu \geq 0 \cond \Zo=2}$} become substantially wider, while the confidence intervals for the posterior mean are less sensitive. One way of determining how pessimistic a given choice of $\tau$ may be, is to inspect the worst-case priors in~\eqref{eq:nbhood_worst_case}. Figure~\ref{fig:worstcase} shows these for the local false sign rate \smash{$\PP[\gprior]{\mu \geq 0 \cond \Zo=2}$}.

\begin{table}[H]
\centering
\begin{tabular}{@{}ccccc@{}} \toprule
$\tau$ &  SLR & Goodness of fit rejected & CI for $\PP[\gprior]{\mu \geq 0 \cond \Zo=2}$ & CI for $\EE[\gprior]{\mu \cond \Zo=2}$ \\ \midrule
0.02 & 0.0030 & $\text{\sffamily X}$ & 0.1859 -- 0.9996 & 0.0196 -- 0.7156\\
0.10 & 0.0032 & $\text{\sffamily X}$ & 0.4491 -- 0.9746 & 0.0302 -- 0.7127\\
0.25 & 0.0042 & $\text{\sffamily X}$ & 0.6396 -- 0.8905 & 0.0922 -- 0.6960\\
0.50 & 1.8847 & $\text{\sffamily X}$ & 0.8043 -- 0.8325 & 0.3834 -- 0.5291\\
0.55 & 74.106 & $\checkmark$ &&
\\ \bottomrule
\end{tabular}
\caption{Goodness-of-fit testing and sensitivity analysis for the Prostate data.}
\label{tab:sensitivity}
\end{table}

\begin{figure}
\centering
\begin{tabular}{c}
\begin{adjustbox}{width=\linewidth}\input{tikz_figures/worst_case_priors.tikz}\end{adjustbox} 
\end{tabular}
\caption{\textbf{Lebesgue density of worst-case priors} in the Gauss-$F$-Localization approach~\eqref{eq:nbhood_worst_case} for the local false sign rate \smash{$\PP[\gprior]{\mu \geq 0 \cond \Zo=2}$} in the Prostate dataset. Each panel corresponds to a different specification of the prior class $\mathcal{G}$, namely \smash{$\gcal=\law\nn(\tau^2, [-3.3])$} where $\tau$ varies across panels.}
\label{fig:worstcase}
\end{figure}

\end{appendix}

\end{document}